\documentclass[11pt]{article}

\usepackage[letterpaper,margin=1in]{geometry}
\usepackage[T1]{fontenc}
\usepackage{lmodern}
\usepackage{amsmath,amssymb,amsthm,mathtools}
\usepackage{booktabs}
\usepackage{array}
\usepackage{enumitem}
\usepackage{float}
\usepackage{microtype}
\usepackage{needspace}
\usepackage{placeins}
\usepackage{tikz}
\usepackage[round,authoryear]{natbib}
\usepackage{xcolor}
\usepackage[colorlinks=true,allcolors=blue!55!black]{hyperref}
\usepackage[nameinlink,noabbrev]{cleveref}

\usetikzlibrary{arrows.meta}

\newtheorem{theorem}{Theorem}[section]
\newtheorem{lemma}[theorem]{Lemma}
\newtheorem{proposition}[theorem]{Proposition}
\newtheorem{corollary}[theorem]{Corollary}

\theoremstyle{definition}
\newtheorem{definition}[theorem]{Definition}
\newtheorem{example}[theorem]{Example}
\theoremstyle{remark}
\newtheorem{remark}[theorem]{Remark}

\numberwithin{equation}{section}
\setlist{leftmargin=*}
\allowdisplaybreaks[2]
\newcolumntype{L}[1]{>{\raggedright\arraybackslash}p{#1}}

\newcommand{\R}{\mathbb R}

\newcommand{\N}{\mathbb N}
\newcommand{\one}{\mathbf 1}
\newcommand{\id}{\operatorname{id}}
\newcommand{\supp}{\operatorname{supp}}
\newcommand{\esssup}{\operatorname*{ess\,sup}}

\newcommand{\re}{\operatorname{Re}}

\newcommand{\pc}{\mathcal A_{\mathrm{pc}}}
\newcommand{\rt}{\mathcal A_{\mathrm{rt}}}
\newcommand{\readouts}{\mathcal H}
\newcommand{\cost}{\mathcal C}
\newcommand{\storage}{\mathcal S}
\newcommand{\laplace}[1]{\widehat{#1}}
\newcommand{\dd}{\,\mathrm d}
\newcommand{\transpose}{\mathsf T}
\newcommand{\weakstar}{\mathop{\rightharpoonup}\limits^{*}}
\newcommand{\DeltaT}{\Delta_T}

\title{Price manipulation in nonlinear transient impact models:\\
rigidity before memory and complete positivity after memory}
\author{Minhyeok Lee\\[-0.1em]
\small Independent Researcher, Republic of Korea\\[-0.1em]
\small \texttt{borrownotime@gmail.com}}
\date{2 September 2026}

\hypersetup{
  pdftitle={Price manipulation in nonlinear transient impact models: rigidity before memory and complete positivity after memory},
  pdfauthor={Minhyeok Lee},
  pdfsubject={No-manipulation classifications for nonlinear transient market impact: rigidity before memory, complete positivity after memory},
  pdfkeywords={transient market impact, price manipulation, square-root law, power-law decay, nonlinear propagators, complete positivity, Prony kernels, dissipativity, round trip, Volterra equation}
}

\begin{document}
\maketitle

\begin{abstract}
Transient impact models compose a nonlinearity with a memory kernel, and
the order of composition determines the criterion for absence of price
manipulation.  We classify both orders.  If an arbitrary instantaneous law
$f$ acts on the trading rate before any nonzero integrable Volterra kernel,
nonnegative cost on every finite piecewise-constant round trip forces $f$
to be affine, and linear for every nonzero convolution kernel.  In
particular, the power law $\operatorname{sgn}(x)|x|^\delta$, $\delta>0$,
combined with power-law decay $t^{-\gamma}$, $0<\gamma<1$, admits
manipulation if and only if $\delta\ne1$: square-root impact is manipulable
at every decay exponent, and the region left open by Gatheral's slow-rate
two-block bound $\delta+\gamma\ge1$ collapses to the line $\delta=1$.
Earlier rigidity theorems require a kernel that is bounded at zero; the
argument here is a zero-volume chattering pump read out by two thin
baseline trades, and it applies to singular kernels.  If instead a monotone
readout acts on the impact state after the kernel, safety for all inputs
and all readouts is equivalent to complete positivity of the kernel, with a
constructive converse; in particular, square-root impact after power-law
memory is manipulation-free.  A remote compensating block shows that
round-trip safety and all-input safety coincide for kernels with uniformly
vanishing tails and differ, for permanent memory, by an explicit storage
quotient.  These mechanisms classify every two-mode Prony kernel,
first-order time-inhomogeneous memory, and stable fully actuated matrix
memory, and they quantify the friction, the two-block phase, and the
switching complexity behind the power-law case.  Calibrated exponent pairs
all lie in the manipulable set: absent friction, concavity has to enter
after the memory, not before it.
\end{abstract}



\section{Introduction}\label{sec:introduction}

The propagator model of \citet{BouchaudGefenPottersWyart2004} and
\citet{Gatheral2010} describes the price response to a trading rate $v$ as
the convolution of a memory kernel with an instantaneous impact law.  Two
empirical regularities are usually imposed on it: the response to a single
trade is concave in size, well approximated by a square root
\citep{TothEtAl2011,AbiJaberEtAl2025}, and it decays as a power law
\citep{BouchaudGefenPottersWyart2004}.  Whether these two facts can coexist
in one model without allowing a trader to make money on a round trip has
been an open question since \citet{Gatheral2010} derived a slow-rate
power-law bound and a separate maximal-rate bound,
\citet{CuratoGatheralLillo2017} found negative-cost strategies numerically
inside the resulting schematic region, and \citet{SchneiderLillo2019} and
\citet{AbiJaberEtAl2025} recorded
the question as open.  This paper answers it, and shows that the answer
depends on where the nonlinearity is placed relative to the memory.

Three architectures are natural:
\begin{align}
\text{rate-inside:}\qquad
 D(t)&=\int_0^t H(t,s)f(v(s))\dd s,
 &\cost[v]&=\int_0^T v(t)D(t)\dd t,
 \label{eq:intro-inside}\\
\text{state-outside:}\qquad
 D(t)&=(G*v)(t),
 &\cost_h[v]&=\int_0^T v(t)h(D(t))\dd t,
 \label{eq:intro-outside}\\
\text{mixed:}\qquad
 D(t)&=(G*f(v))(t),
 &\cost_{f,h}[v]&=\int_0^T v(t)h(D(t))\dd t.
 \label{eq:intro-mixed}
\end{align}
Here $G$ or $H$ is a causal memory kernel, $f$ acts on the contemporaneous
rate, and $h$ reads the accumulated impact state.  With a convolution kernel
$H(t,s)=G(t-s)$, the rate-inside model is the model of \citet{Gatheral2010},
which the mixed model recovers at $h=\id$; the state-outside model is the
nonlinear propagator of \citet{AbiJaberEtAl2025}.  The economic test is
nonnegative expected impact cost on every round trip $\int_0^Tv=0$.  We
impose it on ordinary controls with finitely many constant pieces; no
relaxed or distributional strategy is admissible.  When a theorem quantifies
over readouts, the class is
\begin{equation}\label{eq:intro-readouts}
 \readouts
 =\{h:\R\to\R: h\text{ is continuous and nondecreasing},\ h(0)=0\}.
\end{equation}
Quantifying over all monotone readouts is what makes complete positivity,
rather than positive definiteness, the relevant kernel condition.

\subsection{Main results}\label{sec:main-results}

\begin{enumerate}[label=(\arabic*),leftmargin=2em]
\item \emph{Rigidity before memory} (\Cref{thm:exact}, \Cref{cor:convolution}, \Cref{thm:fractional-main}, \Cref{cor:all-powers}).
  Let $H$ be any nonzero real kernel in $L^1$ on the Volterra triangle and
  $f$ any real function.  The rate-inside cost is nonnegative on every finite
  piecewise-constant round trip if and only if $f$ is affine, its intercept
  is compatible with the row integral of $H$, and the slope gives the linear
  quadratic form the right sign on round trips.  For a nonzero convolution
  kernel the intercept vanishes.  For $H(t,s)=(t-s)^{-\gamma}$,
  $0<\gamma<1$, safety is equivalent to $f(x)=\lambda x$ with
  $\lambda\ge0$; hence the power law $f_\delta(x)=\operatorname{sgn}(x)|x|^\delta$
  is safe if and only if $\delta=1$.  Every witness is a bounded finite
  round trip and fits inside any prescribed rate cap.
\item \emph{Complete positivity after memory}
  (\Cref{thm:outside-cp}, \Cref{thm:mixed-classification}).  For a scalar
  state-outside model, nonnegative cost for every input and every
  $h\in\readouts$ is equivalent to complete positivity of $G$: the Volterra
  resolvents $r_a+aG*r_a=aG$ and $s_a+aG*s_a=1$ are nonnegative for every
  $a>0$.  The converse is constructive and is already detected by one smooth
  dead-zone readout.  In the mixed model, universal safety forces
  $f=\lambda\id$ and $\lambda G$ completely positive.  Since every completely
  monotone kernel, in particular $t^{-\gamma}$, is completely positive,
  square-root impact applied after power-law memory is safe for every input
  (\Cref{cor:concave-after-powerlaw}).
\item \emph{Round trips versus all inputs}
  (\Cref{thm:remote-compensation}, \Cref{thm:permanent-quotient}).  If the kernel
  tail vanishes uniformly, a remote slow block closes any input at
  asymptotically zero cost, so round-trip safety on all horizons equals
  all-input safety.  If $G\to g\ne0$, the two notions differ by the storage
  quotient $F_h(gm)/g$, $F_h'=h$, $m=\int v$, which keeps the sign of $g$.
\item \emph{Permanent shifts and Prony kernels}
  (\Cref{thm:licm-shifts}, \Cref{thm:stable-complement}, \Cref{thm:two-mode}, \Cref{prop:signed-three-mode}, \Cref{thm:finite-prony-hierarchy}).
  Every real permanent shift of a completely monotone transient is safe on
  round trips for every readout, including shifts under which the kernel
  changes sign.  Within a stable inverse class, safety is equivalent to a
  nonnegative nonincreasing inverse density.  For
  $G=g+a_1e^{-\lambda_1t}+a_2e^{-\lambda_2t}$ the safe set is described
  completely: either $a_1,a_2\ge0$, or $gG(0+)>0$ and the inverse density
  passes an explicit residue test.  Signed safe kernels exist; a decaying
  completely positive tail with residue signs $+,-,+$ exists at three modes
  and not at two; for general finite-Prony kernels with nonzero
  instantaneous gain and no imaginary-axis inverse poles, safety is an
  infinite hierarchy of finite cut linear programs, with every failure
  certified by a finite round trip.
\item \emph{Nonstationary and matrix memory}
  (\Cref{thm:time-inhomogeneous}, \Cref{thm:vector-remote}, \Cref{thm:stable-matrix}).
  For $D'=-\rho(t)D+b(t)v$, $b>0$, safety for all readouts is equivalent to
  $\rho\ge0$ and $b'+\rho b\ge0$ pointwise.  For stable fully actuated
  matrix dynamics $D'=-AD+Bv$ and a fixed $C^1$ readout $h$, round-trip
  safety on all horizons is equivalent to $B^{-\transpose}h$ being
  conservative with $(B^{-\transpose}h(x))^\transpose Ax\ge0$; the
  round-trip statement, not the all-input one, is what is new.
\item \emph{Fixed readouts, friction, and complexity}
  (\Cref{thm:vertical-reachability}, \Cref{thm:multipower-repair}, \Cref{thm:fractional-two-block-phase}, \Cref{prop:constructive-block-bound}).
  A bounded dead zone cannot hide a nonlinear $f$ under power-law memory.
  Under a common rate cap, a power penalty $\kappa\int\Psi(v)$ repairs the
  model with a finite coefficient exactly when the smallest exponent of
  $\Psi$ is at most $1+\delta$; without a common cap, such a finite coefficient
  exists exactly when the exponents of $\Psi$ bracket $1+\delta$.  A bid--ask
  spread alone cannot repair the uncapped model, a quadratic cost alone cannot
  repair concave impact, and the two together repair every concave law.  Within
  two-block strategies the
  manipulable set is $\delta<1-\gamma$ or $\delta>\delta_+(\gamma)$, with
  $\delta_+(\gamma)=\infty$ exactly when $\gamma\ge2-\log_23$.  The
  sublinear boundary is Gatheral's slow-rate bound, while the same constant
  $2-\log_23$ occurs in his separate maximal-rate blow-up regime.  Every
  manipulable point outside the two-block regions needs at least three blocks;
  near $\delta=1$,
  $O(|\delta-1|^{-2/(1-\gamma)})$ blocks suffice.
\end{enumerate}

\Cref{fig:phase} draws the consequence for the power-law family.  The
calibrated pairs $(\delta,\gamma)\approx(0.6,0.4)$ and $(0.5,0.5)$ of
\citet{Gatheral2010}, with $\gamma\approx0.4$ from
\citet{BouchaudGefenPottersWyart2004} and $\delta\approx1/2$ from the
square-root law, all lie in the manipulable set, as does every other pair
with $\delta\ne1$.  The rate-inside model with concave impact is therefore
not a description of a manipulation-free market.  The theorems also say
what repairs it without giving up either stylized fact: apply the concave
law to the impact state rather than to the rate.  For the power-law kernel, or any completely monotone
kernel, every monotone readout is safe for every input, and the kernel
condition that this architecture needs in general is complete positivity.

\subsection{Relation to prior work}\label{sec:prior-work}

\paragraph{Rigidity before memory.}
That nonlinear instantaneous impact creates manipulation is known for
permanent impact \citep{HubermanStanzl2004}, for exponential decay
\citep[Lemma~4.1]{Gatheral2010}, for nonincreasing kernels with
$G(0+)<\infty$ \citep[Proposition~1]{GatheralSchiedSlynko2011}, and for
bounded kernels with cross impact \citep[Lemma~3.5]{SchneiderLillo2019}.
These continuous-time arguments compress a two-block round trip into a
horizon on which the kernel is effectively permanent.  In discrete time,
\citet*[Theorem~2.4]{HeyNeumanTuschmann2025} obtain an all-input rigidity
result for kernels continuous at zero by perturbing a three-trade
configuration; the resulting negative witness has nonzero total volume and
is not a round trip.  Both mechanisms rely on a finite value or continuity
at the origin and are unavailable for the power-law kernel $t^{-\gamma}$,
which is unbounded at zero and scale invariant.  \citet*{HeyNeumanTuschmann2025}
explicitly restrict their result to nonsingular kernels;
\citet{CuratoGatheralLillo2017} report negative expected costs numerically
inside the region $\delta+\gamma\ge1$ allowed by Gatheral's slow-rate
condition and conclude that the concave model is misspecified;
\citet{SchneiderLillo2019} record the consistency of power-law decay with a
nonlinear impact function as an open problem; \citet{AbiJaberEtAl2025}
describe the compatibility of square-root impact with power-law decay as a
long-standing open problem.  \Cref{thm:exact} replaces the short-horizon
argument by a zero-volume chattering pump read out by two thin baseline
trades.  It works for every nonzero $L^1$ Volterra kernel, singular or not,
signed or not, and for every finite-valued $f$ without regularity.  For the
homogeneous power-law family, \citet[Lemma~5.1]{Gatheral2010} derives the
slow-rate two-block bound $\delta+\gamma\ge1$.  His Lemma~5.2 and
Appendix~A treat a different impact law that diverges at a maximal rate and
produce the constant $2-\log_23$; \Cref{thm:fractional-two-block-phase}
shows where that same constant enters the superlinear two-block phase of the
fixed power-law family.  Chattering controls and their weak-star limits are
classical in relaxed control theory \citep{Young1969,Artstein1989}; the new
point is that a zero-rate, nonzero-impact moment survives an arbitrary
integrable kernel and can be read out by ordinary trades.

\paragraph{Complete positivity after memory.}
Completely positive kernels, their resolvent characterization, and the
complementary kernel $L*G=1$ with nonnegative nonincreasing density are
classical \citep[Definition~1.1 and Theorem~2.2]{ClementNohel1981}, as is the
convolution chain rule used for sufficiency \citep{Zacher2008,VergaraZacher2015}.
Complete accretivity and the characterization of monotone-cone
inequalities by truncations go back to \citet{BenilanCrandall1991}, and the
finite isotone cone we use for cut generation is described in
\citet{Ubhaya2001}.  Nonlinear models in which the nonlinearity follows a
resilient state include the general-shape limit-order-book models of
\citet{AlfonsiFruthSchied2010} and \citet{AlfonsiSchied2010}, the concave
impact model with decay of \citet*{HeyEtAl2025}, the nonlinear propagator of
\citet{AbiJaberEtAl2025}, the market-resistance model of
\citet{DeCarvalhoEtAl2026}, and the concave cross-impact model of
\citet{HeyMastromatteoMuhleKarbe2026}.  These works optimize or establish
absence of manipulation for a selected law and kernel.  The contribution
here is the converse: nonnegative cost for every finite input and every
monotone readout forces complete positivity, and each failure is witnessed
by a finite negative-cost input, which remote compensation closes to a
finite round trip whenever the kernel tail vanishes uniformly.

\paragraph{Linear and multivariate kernels.}
Linear resilient-state benchmarks include \citet{ObizhaevaWang2013} and
\citet{AlfonsiSchiedSlynko2012}.  More generally, absence of manipulation is
governed by positive definiteness of the kernel
\citep{GatheralSchiedSlynko2012}, by
matrix-valued positive definite functions in several assets
\citep{AlfonsiKlockSchied2016}, and by operator positivity for general
propagators \citep{AbiJaberNeuman2022}.  The matrix-valued positive-type
and Laplace--Fourier characterization is classical
\citep[Chapter~16, Theorems~2.4 and~2.6]{GripenbergLondenStaffans1990};
see \citet{GatheralSchied2013} and \citet{Webster2023} for surveys.
\Cref{thm:exact} shows that, on the
rate-inside side, this quadratic problem restricted to the round-trip
subspace is the only one that survives nonlinearity.

\paragraph{Time-varying memory.}
For first-order memory with time-dependent gain and resilience,
\citet{FruthSchoenebornUrusov2014} and \citet{AlfonsiInfante2014} give
coefficient conditions for absence of manipulation under linear or
prescribed order-book shapes.  \Cref{thm:time-inhomogeneous} shows that
quantifying over all monotone readouts separates their combined conditions
into $\rho\ge0$ and $b'+\rho b\ge0$ pointwise, for round trips and
without a forgetting assumption.

\paragraph{Matrix memory and dissipativity.}
The storage--supply formulation of dissipativity is due to
\citet{Willems1972}, and the all-input equations $B^{\transpose}\nabla V=h$,
$\nabla V^\transpose Ax\ge0$ for fully actuated nonlinear systems are the
zero-feedthrough passive case of \citet{Moylan1974},
\citet{HillMoylan1976}, and \citet[Theorem~17]{HillMoylan1980}.
\citet[Definition~8]{HillMoylan1980} require the internal state to return in
a dissipative cycle; a trading round trip instead closes input volume and
can leave the memory state open.
\citet{HeyMastromatteoMuhleKarbe2026} derive fast-loop conservativity
restrictions in an additive concave cross-impact model.  The new step here
is \Cref{thm:vector-remote}: vector round trips, which leave the memory
state open, are as strong as all inputs once every horizon is allowed, so
the fixed-readout matrix classification is global and uncapped.

\subsection{Scope}\label{sec:scope}

Each theorem states its quantifiers: all inputs or round trips, one
readout or the class $\readouts$, one horizon or all horizons, and whether a
witness is cap-local.  The two-mode, time-inhomogeneous, and stable matrix
theorems are exhaustive under their hypotheses.  The finite-Prony hierarchy
is an infinite characterization with finite certificates on the unsafe side
only.  The conic-complement and storage-matching conditions of
\Cref{sec:matrix} are sufficient certificates.  Nothing here rejects
empirically concave impact; the results locate where concavity can sit in a
frictionless transient model and what friction is needed elsewhere.

\subsection{Organization}

\Cref{sec:framework} fixes the admissible classes and the quantifier map.
\Cref{sec:rate-inside} proves rigidity before memory and gives explicit
witnesses at the calibrated exponents.  \Cref{sec:outside} proves the
complete-positivity classification.  \Cref{sec:roundtrip-storage} relates
round trips to all inputs, and \Cref{sec:prony} classifies permanent shifts
and Prony kernels.  \Cref{sec:time-inhomogeneous,sec:matrix} treat
nonstationary and matrix memory.  \Cref{sec:friction-complexity} records
what survives with a fixed readout, which frictions repair the model, and
how many blocks a manipulation needs.  Proof details, separators, and the
computational supplement are described in the appendices.

\paragraph{AI-assisted research.}
This paper was developed with substantive assistance from generative AI.
The systems used and their roles are described in \Cref{sec:ai-use}.

\section{Framework and quantifiers}\label{sec:framework}

\subsection{Controls, horizons, and round trips}

For $T>0$, let $\pc(T;\R^d)$ denote the functions
$v:[0,T]\to\R^d$ that are constant on each member of a finite interval
partition.  The scalar case is $d=1$.  The round-trip class is
\begin{equation}\label{eq:rt-class}
 \rt(T;\R^d)
 =\left\{v\in\pc(T;\R^d):\int_0^T v(t)\dd t=0\right\}.
\end{equation}
If a cap $c>0$ is imposed, the scalar constraint is $|v|\le c$ and the
vector constraint is $\|v\|\le c$ for a fixed norm.  ``Cap-local'' means
that a negative witness can be constructed inside every prescribed positive
cap; it does not mean that safety inside one cap determines a readout outside
the reachable state region.

All states start from zero prehistory.  Unless a theorem explicitly fixes a
horizon, safety is required for every finite $T$.  This all-horizon
quantifier is needed whenever a remote compensation block is used.

\subsection{Scalar architectures}

Let $H\in L^1(\Delta_T)$ be a real Volterra kernel on
$\Delta_T=\{(t,s):0<s<t<T\}$.  For a finite-valued function $f$ on the
accessible rate set, define the rate-inside cost
\begin{equation}\label{eq:rate-inside-cost}
 \cost^{\mathrm{in}}_{H,f,T}[v]
 =\int_0^T v(t)\int_0^tH(t,s)f(v(s))\dd s\dd t.
\end{equation}
No measurability or continuity of $f$ is needed because an individual
$v\in\pc$ takes only finitely many values.

For a real $G\in L^1_{\mathrm{loc}}(0,\infty)$, let
\begin{equation}\label{eq:convolution-state}
 D_v(t)=(G*v)(t)=\int_0^tG(t-s)v(s)\dd s.
\end{equation}
The state-outside and mixed costs are
\begin{align}
 \cost^{\mathrm{out}}_{G,h,T}[v]
 &=\int_0^Tv(t)h(D_v(t))\dd t,
 \label{eq:outside-cost}\\
 \cost^{\mathrm{mix}}_{G,f,h,T}[v]
 &=\int_0^Tv(t)h\bigl((G*f(v))(t)\bigr)\dd t.
 \label{eq:mixed-cost}
\end{align}
The universal scalar readout class is $\readouts$ from
\eqref{eq:intro-readouts}.  We also use the analytic subclass
\begin{equation}\label{eq:analytic-readouts}
 \readouts_{\mathrm{an}}
 =\{h\in\readouts:h\text{ is real analytic and strictly increasing}\}.
\end{equation}

\begin{definition}[Safety quantifiers]\label{def:safety}
For a fixed architecture, kernel, and readout class:
\begin{enumerate}[label=(\roman*)]
\item \emph{all-input safety on $[0,T]$} means nonnegative cost for every
  $v\in\pc(T)$;
\item \emph{round-trip safety on $[0,T]$} means nonnegative cost for
  every $v\in\rt(T)$;
\item \emph{universal-readout safety} means that the corresponding
  inequality holds for every $h\in\readouts$;
\item \emph{global safety} means the printed inequality holds on every
  finite horizon.
\end{enumerate}
\end{definition}

\subsection{Complete positivity}

For $G\in L^1(0,T)$ and $a>0$, let $r_a$ and $s_a$ be the Volterra
resolvents
\begin{equation}\label{eq:resolvent-def}
 r_a+aG*r_a=aG,\qquad s_a+aG*s_a=1.
\end{equation}

\begin{definition}[Complete positivity]\label{def:cp}
The kernel $G$ is \emph{completely positive on $[0,T]$} if, for every
$a>0$, the resolvents in \eqref{eq:resolvent-def} satisfy
$r_a\ge0$ and $s_a\ge0$ almost everywhere on $(0,T)$.
It is globally completely positive if this holds on every finite horizon.
\end{definition}

For nonzero scalar kernels, we use the fixed-horizon complementary-kernel
characterization of \citet[Theorem~2.2]{ClementNohel1981}; see also
\citet[Lemma~2.2]{FengLi2023}:
\begin{equation}\label{eq:complement-form}
 L=\beta\delta_0+\ell(t)\dd t,\qquad
 \beta\ge0,\quad \ell\ge0,\quad \ell\text{ nonincreasing},\qquad
 L*G=1.
\end{equation}
The atom is allowed to vanish, and locally integrable singular kernels are
included.  The zero kernel is completely positive in the resolvent sense but
has no convolution inverse satisfying \eqref{eq:complement-form}; it is
therefore split off whenever the complementary representation is invoked.
Every completely monotone kernel is completely positive
\citep{ClementNohel1981}; \Cref{thm:licm-shifts} below reproves this in the
form needed here.

For $h\in\readouts$, define the convex primitive and its Bregman divergence
\begin{equation}\label{eq:F-bregman}
 F_h(x)=\int_0^x h(z)\dd z,\qquad
 \mathcal B_h(x,y)=F_h(x)-F_h(y)-h(y)(x-y)\ge0.
\end{equation}
The inequality is monotonicity of $h$.  It is the basic dissipation term
throughout the scalar theory.

\subsection{The quantifier map}

\Cref{tab:quantifiers} lists, for each classification, which inputs,
horizons, and readouts are quantified over and which object is
characterized.  Moving between its columns requires a separate argument;
\Cref{sec:roundtrip-storage} supplies the one between round trips and all
inputs.

\begin{table}[ht]
\centering
\small
\caption{Quantifiers controlling the classifications.}
\label{tab:quantifiers}
\begin{tabular}{@{}L{0.21\textwidth}L{0.24\textwidth}
                L{0.24\textwidth}L{0.20\textwidth}@{}}
\toprule
Branch & Inputs and horizons & Readouts & Characterized object \\
\midrule
Rate-inside & round trips; fixed horizon & identity output & affine $f$, intercept condition, and quadratic-form gate \\
Scalar state-outside & all inputs; fixed horizon & all $\readouts$, all $\readouts_{\rm an}$, or one dead-zone law & complete positivity \\
Uniformly vanishing tail & all horizons; all inputs versus round trips & one fixed normalized law, then universalized & the two input classes coincide \\
Permanent limit & all horizons & one fixed normalized law, then universalized & storage quotient $F_h(gm)/g$ \\
Two-mode Prony & round trips; all horizons & every $h\in\readouts$ & two-branch classification \\
Time-inhomogeneous & all inputs and round trips; all horizons & universal or analytic class & $\rho\ge0$, $b'+\rho b\ge0$ \\
Stable matrix first order & all inputs and vector round trips; all horizons & one fixed $C^1$ vector law & conservative one-form and drift dissipation \\
\bottomrule
\end{tabular}
\end{table}

\subsection{Two closure principles}

Two closure facts are used repeatedly.

\begin{lemma}[Finite-step closure]\label{lem:finite-step-closure}
Let $G\in L^1(0,T)$, let $h$ be continuous, and suppose a bounded measurable
input $v$ has strictly negative state-outside cost.  Then some finite
piecewise-constant input has strictly negative cost.  If $v$ has zero
integral, the approximants can be corrected on one interval so that they are
round trips, without losing strict negativity.  The same statement
holds componentwise for a locally integrable matrix kernel and continuous
vector readout.
\end{lemma}

\begin{proof}
Choose interval-step functions $v_n$ with
$\|v_n-v\|_1\to0$ and a common $L^\infty$ bound.  For every bounded kernel $K$, if $\|v_n\|_\infty,\|v\|_\infty\le B$,
\[
 \limsup_n\|G*(v_n-v)\|_\infty
 \le\limsup_n\bigl(\|K\|_\infty\|v_n-v\|_1
   +2B\|G-K\|_1\bigr)=2B\|G-K\|_1.
\]
Now approximate $G$ in $L^1$ by bounded kernels; the state error tends to zero.
Continuity of $h$ on the common compact state range then gives convergence
of the costs.  An $o(1)$ volume error is removed on an interval whose length
is fixed and whose rate correction is $o(1)$; the same estimate shows that
the cost changes by $o(1)$.  The vector proof uses operator and dual norms.
\end{proof}

\begin{lemma}[Cap locality by joint scaling]\label{lem:cap-scaling}
Whenever a strict negative witness is obtained from a state path and a
continuous monotone readout whose relevant threshold can be scaled, the
witness can be placed inside every prescribed positive symmetric rate cap by
scaling the control, state, and readout together.  Statements involving one
fixed dead-zone threshold are excluded unless separately proved.
\end{lemma}

\begin{proof}
For $\varepsilon>0$, replace $v$ by $\varepsilon v$ and define
$h_\varepsilon(x)=\varepsilon h(x/\varepsilon)$, or the corresponding
vector radial scaling.  The state scales by $\varepsilon$, the sign of the
cost is unchanged up to the positive factor $\varepsilon^2$, and
$\varepsilon$ is chosen to meet the cap.  A fixed readout cannot in general
be rescaled, which explains the printed exception.
\end{proof}

One further observation is used for all-horizon mixed statements.  If a
locally integrable convolution kernel is nonzero on some horizon and a mixed
theorem forces $f(x)=\lambda x$ there, the same $\lambda$ applies on every
horizon, because $f$ is common to all horizons; substitution reduces every
other horizon to the effective linear kernel $\lambda G$, and horizons on
which the kernel vanishes have identically zero state.

\section{Rate-inside rigidity}\label{sec:rate-inside}

\subsection{Model and results}\label{sec:rate-inside-model}

Fix a horizon $T>0$ and set
\[
 \DeltaT=\{(t,s):0<s<t<T\}.
\]
The kernel $H:\DeltaT\to\R$ is measurable and belongs to
$L^1(\DeltaT)$.  Unless explicitly specialized, it need not be positive,
translation invariant, monotone, or continuous.  We extend $H$ by zero to
$(0,T)^2$ whenever a product-space argument is used.

A control belongs to $\pc=\pc(T;\R)$ if it is real-valued and constant on
each member of some finite partition of $[0,T]$, and $\rt=\rt(T;\R)$ is the
round-trip class \eqref{eq:rt-class}; both are written without their
arguments throughout this section.  All plateau values and partition points
are arbitrary real numbers.  There is no common bound on the number of
pieces and no common minimum dwell time.

For any function $f:\R\to\R$ and $v\in\pc$, the composition $f\circ v$
has finite range and is measurable.  The impact cost
\eqref{eq:rate-inside-cost}, abbreviated in this section to
\begin{equation}\label{eq:cost}
 C_{H,f}[v]
 =\int_0^T v(t)\int_0^t H(t,s)f(v(s))\dd s\dd t
 =\int_{\DeltaT}H(t,s)v(t)f(v(s))\dd s\dd t,
\end{equation}
is therefore an absolutely convergent integral even if $f$ has no
regularity.  In the price model \eqref{eq:intro-inside}, \eqref{eq:cost} is
the impact component of implementation shortfall.  The unaffected initial price cancels on a round
trip; an unaffected martingale also contributes zero in expectation under
the usual predictability and integrability assumptions.  We isolate
\eqref{eq:cost}, since the theorem is deterministic and does not require a
probabilistic price model.

\begin{definition}[Price manipulation]\label{def:manipulation}
A control $v\in\rt$ is a price manipulation for $(H,f)$ if
$C_{H,f}[v]<0$.  The pair $(H,f)$ is \emph{safe} on a class of round trips if
the cost is nonnegative on every member of that class.
\end{definition}

Introduce the row integral and the linear quadratic form
\begin{equation}\label{eq:AHQH}
 A_H(t)=\int_0^t H(t,s)\dd s\in L^1(0,T),
 \qquad
 Q_H[v]=\int_{\DeltaT}H(t,s)v(t)v(s)\dd s\dd t.
\end{equation}

\begin{theorem}[Affine rigidity]\label{thm:exact}
Let $H\in L^1(\DeltaT;\R)$ be nonzero and let $f:\R\to\R$ be arbitrary.
The following are equivalent:
\begin{enumerate}[label=\textup{(\roman*)}]
\item $C_{H,f}[v]\ge0$ for every $v\in\rt$;
\item there exist $c,\lambda,\kappa\in\R$ such that
\begin{equation}\label{eq:affine-classification}
 f(x)=c+\lambda x\quad(x\in\R),
 \qquad cA_H(t)=\kappa\quad\text{for a.e. }t,
\end{equation}
and
\begin{equation}\label{eq:slope-sign}
 \lambda Q_H[v]\ge0\qquad(v\in\rt).
\end{equation}
\end{enumerate}
Necessarily $c=f(0)$.  In particular, if $f(0)=0$, safety forces
$f(x)=\lambda x$ on all of $\R$.
\end{theorem}

\begin{corollary}[Nonzero convolution kernels]\label{cor:convolution}
Suppose $H(t,s)=G(t-s)$ with $G\in L^1(0,T)$ and $G\ne0$ in
$L^1(0,T)$.  Then every safe impact law has zero intercept.  Consequently,
\begin{equation}\label{eq:convolution-classification}
 C_{H,f}\ge0\ \text{on }\rt
 \quad\Longleftrightarrow\quad
 f(x)=\lambda x\ \text{on }\R
 \ \text{and}\ \lambda Q_H\ge0\ \text{on }\rt.
\end{equation}
\end{corollary}

Indeed, $A_H(t)=\int_0^tG(r)\dd r$ is absolutely continuous.  If
$cA_H$ is constant and $c\ne0$, then $G=A_H'=0$ almost everywhere, a
contradiction.  Thus an affine intercept can survive only for a
time-inhomogeneous kernel, such as $H(t,s)=1/t$, whose row integral is
constant.

The slope condition matters for signed kernels.  If $Q_H\ge0$ on round trips
and is not identically zero there, the safe slopes are $\lambda\ge0$; if
$Q_H\le0$, they are $\lambda\le0$; if $Q_H$ takes both signs on round
trips, only $\lambda=0$ is safe; and if $Q_H\equiv0$ on round trips, every
slope is harmless.

For a common speed cap $B>0$, put
\begin{equation}\label{eq:capclass}
 \rt(B)=\{v\in\rt:\|v\|_\infty\le B\}.
\end{equation}

\begin{theorem}[Common speed cap]\label{thm:cap}
Let $H\ne0$ and $B>0$.  The cost is nonnegative on $\rt(B)$ if and only
if there exist $c,\lambda_B,\kappa\in\R$ such that
\begin{equation}\label{eq:cap-affine}
 f(x)=c+\lambda_Bx\quad(|x|\le B),
 \qquad cA_H=\kappa\quad\text{a.e.},
\end{equation}
and $\lambda_BQ_H[v]\ge0$ for every $v\in\rt$.  A common speed cap thus
localizes affine rigidity to the accessible rate interval; it does not make a
nonlinear law safe on that interval.
\end{theorem}

The financially central specialization is the power-law kernel.

\begin{theorem}[Fractional-kernel classification]\label{thm:fractional-main}
Let $0<\gamma<1$, let $H_\gamma(t,s)=(t-s)^{-\gamma}$, and let
$f:\R\to\R$ be arbitrary.  Then
\begin{equation}\label{eq:fractional-equivalence}
 C_{H_\gamma,f}[v]\ge0\quad\text{for every }v\in\rt
 \quad\Longleftrightarrow\quad
 f(x)=\lambda x\quad(x\in\R),\quad\lambda\ge0.
\end{equation}
Under a common speed cap $B$, the corresponding condition is
$f(x)=\lambda_Bx$ on $[-B,B]$ with $\lambda_B\ge0$.
\end{theorem}

\begin{corollary}[All homogeneous powers]\label{cor:all-powers}
For $0<\gamma<1$ and
\[
 f_\delta(x)=\operatorname{sgn}(x)|x|^\delta,\qquad\delta>0,
\]
there is no price manipulation if and only if $\delta=1$.  For every
$\delta\ne1$ there is a finite bounded piecewise-constant manipulation,
even under any prescribed positive common speed cap.
\end{corollary}

\paragraph{Relation to earlier rigidity theorems.}
For kernels that are bounded near zero, the conclusion of \Cref{thm:exact}
is known from the continuous-time results reviewed in
\Cref{sec:prior-work}: each of those proofs compresses a two-block round
trip into a horizon $[0,\nu]$ on which $G\approx G(0)$, so that the
permanent-impact argument of \citet{HubermanStanzl2004} applies.  The
discrete-time all-input result of
\citet*[Theorem~2.4]{HeyNeumanTuschmann2025} likewise needs continuity at
zero, and its negative witness is not a round trip.  For
$H_\gamma(t,s)=(t-s)^{-\gamma}$ there is no finite $G(0)$ or continuous
origin to exploit: the kernel is unbounded at zero and invariant under time
scaling.  Within the homogeneous power-law family,
\citet[Lemma~5.1]{Gatheral2010} shows that slow-accumulation two-block
round trips are safe only if $\delta+\gamma\ge1$, while his Lemma~5.2 and
Appendix~A obtain $\gamma\ge2-\log_23$ for a different impact law that
blows up at a maximal rate.  That constant reappears below as the
transition at which the fixed-power-law superlinear threshold
$\delta_+(\gamma)$ becomes infinite; it is not a second necessary condition
on every fixed power law.  \Cref{thm:exact} settles the singular case with
a different mechanism: a chattering pump with zero volume and nonzero mean
impact, read out by two thin baseline trades.  The mechanism uses only
integrability of the kernel, so the same theorem covers signed,
nonconvolution, and integrably singular kernels, with no regularity of $f$.

\Cref{fig:phase} places this result in the exponent plane.  The slow-rate
bound of \citet[Lemma~5.1]{Gatheral2010} leaves
$\delta+\gamma\ge1$.  \Cref{thm:fractional-two-block-phase} below proves
that this boundary is sharp on the sublinear side and that, on the
superlinear side, two blocks manipulate exactly when
$\delta>\delta_+(\gamma)$.  The horizontal level
$\gamma=2-\log_23$ marks where $\delta_+(\gamma)$ becomes infinite.
Arbitrary finite switching removes every point of the plane except the line
$\delta=1$.

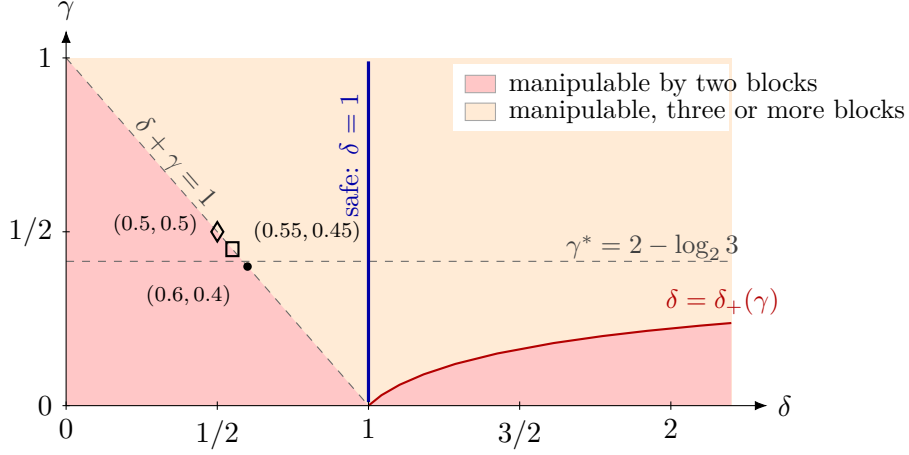
\begin{figure}[t]
\centering
\begin{tikzpicture}[x=4.0cm,y=4.6cm,>=Latex]
  \fill[orange!16] (0,0) rectangle (2.2,1);
  \fill[red!22] (0,0)--(1,0)--(0,1)--cycle;
  \fill[red!22] (1,0)--(1.0429,0.03)--(1.1030,0.06)--(1.1826,0.09)
    --(1.2874,0.12)--(1.4261,0.15)--(1.6125,0.18)--(1.7736,0.20)
    --(1.9207,0.215)--(2.0967,0.23)--(2.2,0.2376)--(2.2,0)--cycle;
  \draw[dashed,gray!80!black] (0,1)--(1,0);
  \draw[dashed,gray!80!black] (0,0.415)--(2.2,0.415);
  \node[gray!50!black,anchor=west,font=\small] at (1.62,0.455)
    {$\gamma^*=2-\log_23$};
  \node[gray!50!black,font=\small,rotate=-49] at (0.36,0.70)
    {$\delta+\gamma=1$};
  \draw[thick,red!70!black] (1,0)--(1.0429,0.03)--(1.1030,0.06)--(1.1826,0.09)
    --(1.2874,0.12)--(1.4261,0.15)--(1.6125,0.18)--(1.7736,0.20)
    --(1.9207,0.215)--(2.0967,0.23)--(2.2,0.2376);
  \node[red!70!black,font=\small,anchor=west] at (1.95,0.30) {$\delta=\delta_+(\gamma)$};
  \draw[very thick,blue!65!black] (1,0.01)--(1,0.99);
  \node[blue!65!black,rotate=90,anchor=south,font=\small] at (1.0,0.72) {safe: $\delta=1$};
  \draw[->] (0,0)--(2.32,0) node[right] {$\delta$};
  \draw[->] (0,0)--(0,1.08) node[above] {$\gamma$};
  \foreach \x/\lab in {0/0,0.5/{1/2},1/1,1.5/{3/2},2/2}
    \draw (\x,0.012)--(\x,-0.012) node[below] {$\lab$};
  \foreach \y/\lab in {0/0,0.5/{1/2},1/1}
    \draw (0.012,\y)--(-0.012,\y) node[left] {$\lab$};
  \fill[black] (0.6,0.4) circle (1.6pt);
  \node[anchor=north east,font=\scriptsize] at (0.58,0.375) {$(0.6,0.4)$};
  \draw[black,thick] (0.5,0.5) ++(0,-0.028) -- ++(0.02,0.028) -- ++(-0.02,0.028) -- ++(-0.02,-0.028) -- cycle;
  \node[anchor=east,font=\scriptsize] at (0.47,0.52) {$(0.5,0.5)$};
  \draw[black,thick] (0.53,0.43) rectangle (0.57,0.47);
  \node[anchor=west,font=\scriptsize] at (0.585,0.50) {$(0.55,0.45)$};
  \fill[white] (1.28,0.79) rectangle (2.19,0.98);
  \filldraw[fill=red!22,draw=gray!70] (1.32,0.90) rectangle (1.42,0.95);
  \node[anchor=west,font=\small] at (1.43,0.925) {manipulable by two blocks};
  \filldraw[fill=orange!16,draw=gray!70] (1.32,0.82) rectangle (1.42,0.87);
  \node[anchor=west,font=\small] at (1.43,0.845) {manipulable, three or more blocks};
\end{tikzpicture}
\caption{No-manipulation diagram for $H_\gamma(t,s)=(t-s)^{-\gamma}$ and
$f_\delta(x)=\operatorname{sgn}(x)|x|^\delta$.  The diagonal dashed line is
the slow-rate two-block boundary $\delta+\gamma=1$ of
\citet[Lemma~5.1]{Gatheral2010}.  The horizontal dashed line marks
$\gamma^*=2-\log_23\approx0.415$, where the superlinear threshold
$\delta_+(\gamma)$ becomes infinite; the same constant occurs in
\citet[Lemma~5.2 and Appendix~A]{Gatheral2010} for a separate maximal-rate
blow-up law.  Dark shading is the set manipulated by a two-block round trip
(\Cref{thm:fractional-two-block-phase}); $\delta_+(\gamma)\to\infty$ as
$\gamma\uparrow\gamma^*$.  Light shading is manipulable by
\Cref{thm:exact} but needs at least three blocks.  The safe set is the line
$\delta=1$.  The dot and diamond are the calibrated pairs
$(\delta,\gamma)\approx(0.6,0.4)$ and $(0.5,0.5)$ discussed in
\citet{Gatheral2010}; the square marks $(0.55,0.45)$, where
\citet{CuratoGatheralLillo2017} found negative execution costs numerically.
\Cref{tab:certificates} gives a certified finite manipulation at
$(0.5,0.5)$ and at two additional nearby exponent pairs; \Cref{cor:all-powers}
covers every plotted point off the line $\delta=1$.}
\label{fig:phase}
\end{figure}

The proofs begin with the algebraic obstruction that makes chattering
possible.

\subsection{The two-rate impact pump}\label{sec:pump}

The first lemma isolates the only algebra required of the impact law.  Its
contrapositive is what links the positive and negative half-lines and forces a
single global slope.

\begin{lemma}[Two-rate separation]\label{lem:two-rate}
Let $g:\R\to\R$ satisfy $g(0)=0$.  If $g$ is not proportional to the
identity, there exist $x>0>y$ such that, with
\begin{equation}\label{eq:duty-cycle}
 \theta=\frac{-y}{x-y}\in(0,1),
 \qquad
 n_0=\theta g(x)+(1-\theta)g(y),
\end{equation}
one has
\begin{equation}\label{eq:two-moments}
 \theta x+(1-\theta)y=0,
 \qquad n_0\ne0.
\end{equation}
If $g$ is not proportional to the identity on $[-B,B]$, the two rates can
be chosen in that interval.
\end{lemma}

\begin{proof}
Suppose instead that the second moment in \eqref{eq:two-moments} vanished
for every $x>0>y$.  Multiplying by $x-y$ gives
\[
 -y g(x)+xg(y)=0,
 \qquad\text{hence}\qquad
 \frac{g(x)}x=\frac{g(y)}y.
\]
Fixing one negative point shows that $g(x)/x$ is the same for all $x>0$;
fixing one positive point gives the same conclusion for all $y<0$, with the
same constant on both sides.  Together with $g(0)=0$, this makes $g$
proportional to the identity.  The capped argument is identical.
\end{proof}

Fix an open source interval $P=(p_0,p_1)\Subset(0,T)$.  Divide $P$ into
$M$ equal cells.  In each cell, assign rate $x$ on the initial fraction
$\theta$ and rate $y$ on the remaining fraction $1-\theta$; set the control
equal to zero outside $P$.  Denote the resulting pump by $u_M$.  Every cell
has zero net volume and mean impact $n_0$.  \Cref{fig:pump} shows the
finite witness assembled from this pump.  Rapidly oscillating controls and
their weak-star limits are the chattering controls of relaxed control theory
\citep{Young1969,Artstein1989}; here the limit is never used as a control,
only to select a finite switching count.

\begin{figure}[t]
\centering
\begin{tikzpicture}[x=0.78cm,y=0.55cm,>=Latex]
  \draw[->] (0,0) -- (13.1,0) node[right] {$t$};
  \draw[->] (0,-2.0) -- (0,2.2) node[above] {$v(t)$};
  \draw[dashed,gray] (0,1.15) -- (13,1.15);
  \draw[dashed,gray] (0,-1.15) -- (13,-1.15);
  \draw[very thick,blue!65!black]
    (0.8,0)--(1.1,0)--(1.1,1.15)--(2.0,1.15)--(2.0,0)--(3.0,0);
  \draw[very thick,orange!80!black]
    (3.0,0)--(3.0,1.75)--(3.35,1.75)--(3.35,-1.45)--(3.8,-1.45)
    --(3.8,1.75)--(4.15,1.75)--(4.15,-1.45)--(4.6,-1.45)
    --(4.6,1.75)--(4.95,1.75)--(4.95,-1.45)--(5.4,-1.45)
    --(5.4,1.75)--(5.75,1.75)--(5.75,-1.45)--(6.2,-1.45)
    --(6.2,1.75)--(6.55,1.75)--(6.55,-1.45)--(7.0,-1.45)
    --(7.0,1.75)--(7.35,1.75)--(7.35,-1.45)--(7.8,-1.45)
    --(7.8,0)--(10.7,0);
  \draw[very thick,blue!65!black]
    (10.7,0)--(10.7,-1.15)--(11.6,-1.15)--(11.6,0)--(12.6,0);
  \draw[<->] (1.1,-1.75)--(2.0,-1.75) node[midway,below] {$I_-$};
  \draw[<->] (3.0,-1.75)--(7.8,-1.75) node[midway,below] {$P$};
  \draw[<->] (10.7,-1.75)--(11.6,-1.75) node[midway,below] {$I_+$};
  \node[left] at (0,1.15) {$a$};
  \node[left] at (0,-1.15) {$-a$};
  \node[orange!80!black] at (5.4,2.05) {zero mean rate, nonzero mean impact};
\end{tikzpicture}
\caption{The finite witness of \Cref{prop:normalized}.  Equal early and
late baseline intervals close the round trip.  On the source interval $P$,
the pump alternates between two bounded rates $x>0>y$ with zero mean in
every cell but mean impact $n_0\ne0$.  The display assumes $\sigma=1$; the
proof reverses both baseline signs when required.}
\label{fig:pump}
\end{figure}
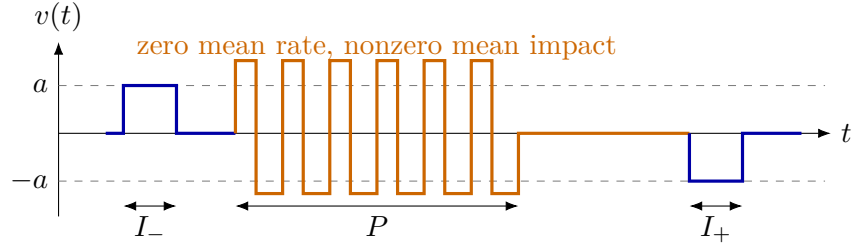

\begin{lemma}[Periodic pump averaging]\label{lem:pump-average}
As $M\to\infty$,
\begin{equation}\label{eq:pump-limits}
 u_M\weakstar0,
 \qquad
 g(u_M)\weakstar n_0\one_P
 \quad\text{in }L^\infty(0,T).
\end{equation}
\end{lemma}

\begin{proof}
For a test function in $L^1(P)$, replace the function on every pump cell by
its cell average.  The integral of the replacement against $u_M$ vanishes
exactly cell by cell, as does its integral against $g(u_M)-n_0$.  The
piecewise cell-average approximation converges in $L^1(P)$ as the mesh tends
to zero.  Uniform boundedness of the two sequences controls the replacement
error.  The extension outside $P$ uses $g(0)=0$.
\end{proof}

Both factors in the cost oscillate, so separate one-dimensional weak-star
convergence is not by itself enough.  The following tensor lemma supplies the
needed product limit for every integrable kernel.

\begin{lemma}[Product weak-star averaging]\label{lem:tensor}
Suppose $a_M\weakstar a$ and $b_M\weakstar b$ in $L^\infty(0,T)$, with
both sequences uniformly bounded.  For every
$K\in L^1((0,T)^2)$,
\begin{equation}\label{eq:tensor-limit}
 \iint K(t,s)a_M(t)b_M(s)\dd s\dd t
 \longrightarrow
 \iint K(t,s)a(t)b(s)\dd s\dd t.
\end{equation}
\end{lemma}

\begin{proof}
For $K(t,s)=\phi(t)\psi(s)$, the integral factors into the product of two
weak-star pairings, so \eqref{eq:tensor-limit} follows.  It follows by
linearity for finite sums of such tensors.  Finite tensor sums are dense in
$L^1((0,T)^2)$, and the uniform $L^\infty$ bounds control the approximation
error independently of $M$.
\end{proof}

Applying \Cref{lem:tensor} with $a_M=u_M$ and $b_M=g(u_M)$ shows that the
pump's self-cost tends to zero: its limiting target-rate factor is zero.  At
the same time, a later ordinary trade can interact with the nonzero impact
moment $n_0\one_P$.  Integrability of the kernel is what makes the pump's
self-cost vanish; a nonintegrable singularity may retain a microscopic
contribution and lies outside the theorem.

\subsection{Ordered localization and thin baselines}\label{sec:localization}

It remains to place ordinary trades before and after the pump so that one of
them reads a nonzero interaction of $H$, while all baseline-only costs are
negligible at the same scale.  For $z\in(0,T)$ and small $h>0$, let
\[
 I_h(z)=(z-h/2,z+h/2).
\]
In this section, every $K\in L^1((0,T)^2)$ is extended by zero to
$\R^2$, so these intervals require no boundary convention.

\begin{lemma}[Diagonal dilution and ordered localization]\label{lem:thin}
Let $K\in L^1((0,T)^2)$.  There exists a sequence $h_n\downarrow0$ such
that, for almost every $z$,
\begin{equation}\label{eq:diagonal-dilution}
 \frac1{h_n}\iint_{I_{h_n}(z)^2}|K(t,s)|\dd s\dd t\longrightarrow0.
\end{equation}
For almost every off-diagonal pair $(t,r)$,
\begin{equation}\label{eq:offdiagonal-dilution}
 \frac1h\iint_{I_h(t)\times I_h(r)}|K(u,s)|\dd s\dd u
 \longrightarrow0.
\end{equation}

\Needspace{14\baselineskip}
If $K=H\one_{\DeltaT}$ is nonzero, there are an interval
$P\Subset(0,T)$, points $r<\inf P<\sup P<\tau$, and a subsequence of
$(h_n)$ such that the equal-length ordered intervals
\[
 I_{-,n}=I_{h_n}(r)<P<I_{+,n}=I_{h_n}(\tau)
\]
satisfy
\begin{align}
 \frac1{h_n}\int_{I_{+,n}}\int_P H(t,s)\dd s\dd t
   &\longrightarrow d\ne0, \label{eq:target-interaction}\\
 \frac1{h_n}\iint_{I_{\pm,n}^2\cap\DeltaT}|H(t,s)|\dd s\dd t
   &\longrightarrow0, \label{eq:self-dilution}\\
 \frac1{h_n}\int_{I_{+,n}}\int_{I_{-,n}}|H(t,s)|\dd s\dd t
   &\longrightarrow0. \label{eq:cross-dilution}
\end{align}
\end{lemma}

\begin{proof}
Set
\[
 d_h(z)=\frac1h\iint_{I_h(z)^2}|K(t,s)|\dd s\dd t.
\]
The set of centers $z$ whose length-$h$ interval contains both $t$ and $s$
has length at most $h\one_{\{|t-s|<h\}}$.  Fubini's theorem therefore
gives
\begin{equation}\label{eq:dh-fubini}
 \int_0^T d_h(z)\dd z
 \le\iint_{|t-s|<h}|K(t,s)|\dd s\dd t\longrightarrow0.
\end{equation}
The last limit follows from absolute continuity of the $L^1$ integral.
Choose $h_n\downarrow0$ so rapidly that
$\sum_n\|d_{h_n}\|_1<\infty$.  Then
\eqref{eq:diagonal-dilution} holds for almost every $z$.  At an
off-diagonal Lebesgue point of $K$, the numerator in
\eqref{eq:offdiagonal-dilution} is
$h^2(|K(t,r)|+o(1))$, proving that limit.

Because $H$ is nonzero, it has an interior Lebesgue point
$(\tau_0,s_0)\in\DeltaT$ at which it is finite and nonzero.  Sufficiently
small, strictly separated intervals $P$ about $s_0$ and $J$ about $\tau_0$
obey
\[
 \int_J\int_P H(t,s)\dd s\dd t\ne0.
\]
Thus $g_P(t)=\int_P H(t,s)\dd s$ is nonzero on a positive-measure subset of
$J$.  By \eqref{eq:diagonal-dilution}, Lebesgue differentiation, and
Fubini applied to the full-measure set of off-diagonal Lebesgue points,
choose $\tau\in J$ such that $g_P(\tau)=d\ne0$, $\tau$ is a Lebesgue
point of $g_P$, \eqref{eq:diagonal-dilution} holds at $\tau$, and
\eqref{eq:offdiagonal-dilution} holds for almost every earlier $r$.  Choose
such an $r$ strictly before $P$ at which
\eqref{eq:diagonal-dilution} also holds.  Discarding finitely many $h_n$
orders the intervals.  The selected properties give
\eqref{eq:target-interaction}--\eqref{eq:cross-dilution}.
\end{proof}

The normalization by $1/h_n$ compares the baseline-only interactions with
the order-$h_n$ cross-interaction; the statement that the diagonal has
two-dimensional measure zero would not do this.

\subsection{Proof of affine rigidity}\label{sec:rigidity}

We first prove the normalized necessity statement.  The proof also makes
the finite nature of the manipulation explicit.

\begin{proposition}[Finite manipulation from nonlinearity]\label{prop:normalized}
Let $H\in L^1(\DeltaT)$ be nonzero and let $g:\R\to\R$ satisfy
$g(0)=0$.  If $g$ is not proportional to the identity, there is a finite,
bounded $v\in\rt$ such that $C_{H,g}[v]<0$.  If $g$ is nonlinear on
$[-B,B]$, the witness can be chosen with $\|v\|_\infty\le B$.
\end{proposition}

\begin{proof}
Choose $x,y,\theta,n_0$ from \Cref{lem:two-rate}.  Use the source interval
$P$ from \Cref{lem:thin} to build the pump $u_M$.  Fix any $a>0$ and put
\[
 \sigma=\operatorname{sgn}(n_0d).
\]
For the thin intervals selected in \Cref{lem:thin}, define
\begin{equation}\label{eq:finite-witness}
 v_{M,n}
 =\sigma a\one_{I_{-,n}}+u_M\one_P-\sigma a\one_{I_{+,n}}.
\end{equation}
The baseline intervals have equal length and every pump cell has zero
integral, so $v_{M,n}$ is a round trip.

For intervals $A,B\subset(0,T)$, write
\[
 J(A,B)=\int_{t\in B}\int_{\substack{s\in A\\s<t}}H(t,s)\dd s\dd t.
\]
For fixed $n$, \Cref{lem:pump-average,lem:tensor} give
\begin{equation}\label{eq:finite-limit}
 C_{H,g}[v_{M,n}]\longrightarrow L_n\qquad(M\to\infty),
\end{equation}
where
\begin{align}
 L_n={}&\sigma a g(\sigma a)J(I_{-,n},I_{-,n})
 -\sigma a g(-\sigma a)J(I_{+,n},I_{+,n})\notag\\
 &-\sigma a g(\sigma a)J(I_{-,n},I_{+,n})
 -\sigma a n_0J(P,I_{+,n}). \label{eq:Ln}
\end{align}
The first three terms are $o(h_n)$ by
\eqref{eq:self-dilution}--\eqref{eq:cross-dilution}; their coefficients are
fixed finite numbers.  By \eqref{eq:target-interaction}, the last term is
\[
 -a|n_0d|h_n+o(h_n).
\]
Thus $L_n<0$ for some finite $n$.  Keeping that $n$ fixed, the convergence
in \eqref{eq:finite-limit} yields a finite $M$ for which
$C_{H,g}[v_{M,n}]<0$.

For the capped assertion, use $x,y\in[-B,B]$ and fix $0<a<B$.  The
construction then remains in the cap.
\end{proof}

Note the order of choices: the pump rates, then a nonzero kernel
interaction, then a nonzero baseline amplitude, then finite baseline
intervals thin enough, then a finite pump-cell count large enough.  The
witness is an ordinary control.

We now determine the intercept.

\begin{lemma}[Intercept condition]\label{lem:intercept}
Suppose $C_{H,f}\ge0$ on $\rt$ and put $c=f(0)$.  Then $cA_H$ is
constant almost everywhere.
\end{lemma}

\begin{proof}
Write $g=f-c$, so $g(0)=0$ and
\begin{equation}\label{eq:intercept-decomposition}
 C_{H,f}[v]
 =C_{H,g}[v]+c\int_0^T v(t)A_H(t)\dd t.
\end{equation}
For almost every ordered pair $r<\tau$, both points are Lebesgue points of
$A_H$, diagonal dilution holds at both, and off-diagonal dilution holds at
$(\tau,r)$.  Fix such a pair.  For fixed $a>0$ and
$\sigma\in\{-1,1\}$, set
\[
 w_{n,\sigma}
 =\sigma a\one_{I_{h_n}(r)}-\sigma a\one_{I_{h_n}(\tau)}.
\]
The $g$-cost is $o(h_n)$ by the same three dilution estimates used in
\eqref{eq:Ln}.  Lebesgue differentiation in
\eqref{eq:intercept-decomposition} gives
\[
 \frac{C_{H,f}[w_{n,\sigma}]}{h_n}
 \longrightarrow c\sigma a\bigl(A_H(r)-A_H(\tau)\bigr).
\]
Nonnegativity for both values of $\sigma$ forces
$cA_H(r)=cA_H(\tau)$ for almost every ordered pair, so $cA_H$ is constant
almost everywhere.
\end{proof}

\begin{proof}[Proof of \Cref{thm:exact}]
Necessity of the intercept condition follows from \Cref{lem:intercept}.  If
$cA_H=\kappa$ almost everywhere, the second term in
\eqref{eq:intercept-decomposition} equals
$\kappa\int_0^Tv=0$ on round trips.  Hence $C_{H,g}\ge0$ on $\rt$.
\Cref{prop:normalized} forces $g(x)=\lambda x$ for all $x$.  The remaining
cost is $\lambda Q_H[v]$, which must have the sign in
\eqref{eq:slope-sign}.  Conversely, the three conditions in
\eqref{eq:affine-classification}--\eqref{eq:slope-sign} give
\[
 C_{H,f}[v]=\kappa\int_0^Tv(t)\dd t+\lambda Q_H[v]\ge0
\]
for every round trip.
\end{proof}

\begin{proof}[Proof of \Cref{thm:cap}]
The two-interval argument and \Cref{prop:normalized} use only rates in
$[-B,B]$, which gives necessity.  Conversely, capped controls see only the
affine restriction of $f$.  Finally, the sign of $Q_H$ on all round trips is
determined by its sign on capped round trips: every finite
piecewise-constant round trip can be scaled into the cap, and $Q_H$ is
homogeneous of degree two.
\end{proof}

\subsection{Convolution and fractional kernels}\label{sec:fractional}

\Cref{cor:convolution} already eliminates the affine intercept from every
nonzero convolution model.  What remains is the sign of the linear
quadratic form.  It is useful to relate this condition to the conventional
positive-definiteness formulation.  Define the symmetric extension
\begin{equation}\label{eq:symmetrized-kernel}
 \overline H(t,s)
 =H(t,s)\one_{\{s<t\}}+H(s,t)\one_{\{t<s\}}.
\end{equation}
Then, with the diagonal irrelevant,
\begin{equation}\label{eq:symmetric-Q}
 Q_H[v]=\frac12\int_0^T\int_0^T
 \overline H(t,s)v(t)v(s)\dd s\dd t.
\end{equation}
Thus the sufficiency side of the linear model is positivity of the
quadratic form on the zero-integral subspace---conditional positive
semidefiniteness for round-trip test functions.  Positive semidefiniteness on
the full function space is sufficient but, for this restricted axiom,
stronger than necessary
\citep{GatheralSchiedSlynko2012,AlfonsiKlockSchied2016,AbiJaberNeuman2022}.
The nonlinear necessity is the new part: the quadratic problem on the
round-trip subspace is the only one that survives the rate-inside
no-manipulation axiom.

For the fractional kernel, positivity is strict and can be seen directly.

\begin{proposition}[Strict positivity of fractional impact]\label{prop:fractional-positive}
For $0<\gamma<1$ and every nonzero bounded piecewise-constant $v$,
\begin{equation}\label{eq:strict-positive-Q}
 Q_{H_\gamma}[v]>0.
\end{equation}
\end{proposition}

\begin{proof}
Use the Laplace-mixture identity
\begin{equation}\label{eq:laplace-mixture}
 r^{-\gamma}=\frac1{\Gamma(\gamma)}
 \int_0^\infty \mu^{\gamma-1}e^{-\mu r}\dd\mu,
 \qquad r>0.
\end{equation}
For $\mu>0$, define
\begin{equation}\label{eq:ymu}
 y_\mu(t)=\int_0^t e^{-\mu(t-s)}v(s)\dd s.
\end{equation}
The function $y_\mu$ is absolutely continuous and satisfies
$y_\mu'=v-\mu y_\mu$.  Hence
\begin{equation}\label{eq:exp-energy}
 \int_0^T v(t)y_\mu(t)\dd t
 =\frac12y_\mu(T)^2+\mu\int_0^T y_\mu(t)^2\dd t.
\end{equation}

For completeness, truncate the integral in \eqref{eq:laplace-mixture} at
$R<\infty$.  Fubini is absolute because
\[
 \int_0^R\!\mu^{\gamma-1}
 \iint_{\DeltaT}|v(t)v(s)|e^{-\mu(t-s)}\dd s\dd t\dd\mu
 \le \|v\|_\infty^2|\DeltaT|\frac{R^\gamma}{\gamma}<\infty.
\]
The truncated kernels increase pointwise to
$r^{-\gamma}$ and are bounded by it.  Since
$|v(t)v(s)|(t-s)^{-\gamma}$ is integrable on $\DeltaT$, dominated
convergence yields
\begin{equation}\label{eq:fractional-energy}
 Q_{H_\gamma}[v]
 =\frac1{\Gamma(\gamma)}\int_0^\infty\mu^{\gamma-1}
 \left[\frac12y_\mu(T)^2
 +\mu\int_0^T y_\mu(t)^2\dd t\right]\dd\mu.
\end{equation}
The bracket is nonnegative.  If it vanished for some $\mu>0$, then
$y_\mu=0$ almost everywhere and therefore
$v=y_\mu'+\mu y_\mu=0$ almost everywhere.  For nonzero $v$ the bracket is
strictly positive for every $\mu>0$, proving \eqref{eq:strict-positive-Q}.
\end{proof}

\begin{proof}[Proof of \Cref{thm:fractional-main}]
The kernel belongs to $L^1(\DeltaT)$ because $0<\gamma<1$.  Its row
integral is
\begin{equation}\label{eq:fractional-row}
 A_{H_\gamma}(t)=\frac{t^{1-\gamma}}{1-\gamma},
\end{equation}
which is not almost everywhere constant, so \Cref{thm:exact} forces the
intercept to vanish.  \Cref{prop:fractional-positive} then makes
$\lambda Q_{H_\gamma}$ nonnegative for every round trip exactly when
$\lambda\ge0$.  The capped statement follows from \Cref{thm:cap}.
\end{proof}

\begin{proof}[Proof of \Cref{cor:all-powers}]
The power law equals the identity exactly at $\delta=1$, where
\Cref{prop:fractional-positive} gives safety.  For $\delta\ne1$ it is
nonlinear on every interval $[-B,B]$ with $B>0$, so
\Cref{thm:fractional-main,thm:cap} give a finite capped manipulation.
\end{proof}

\FloatBarrier

\subsection{Explicit witnesses at the calibrated exponents}\label{sec:comb}

The proof of \Cref{thm:exact} is existential in the switching count.  At
the calibrated exponents the count is small enough to display.  On $[0,1]$
fix a final span $[1-s,1]$ and split it into $M$ equal cells.  In the
rightmost $1/\ell$ fraction of every cell set $v=\ell/s-1$, and set $v=-1$
elsewhere on $[0,1]$.  Each spike has width $s/(M\ell)$, so $v$ is a round
trip with $2M+1$ blocks: slow selling at unit rate throughout, and $M$ brief
intense buying bursts concentrated in the last span.  This is the burst
pattern that \citet{CuratoGatheralLillo2017} observed in their numerical
optima.  For $f=f_\delta$ and $H=H_\gamma$ every block interaction is a
finite combination of
\begin{equation}\label{eq:comb-kernel-primitive}
 K(r)=\frac{r^{2-\gamma}}{(1-\gamma)(2-\gamma)},
\end{equation}
so the cost is a finite sum of fractional powers of rationals.  Evaluating
that sum in directed interval arithmetic gives the certified negative upper
bounds in \Cref{tab:certificates}; the enclosures have width below
$10^{-30}$ and the supplement reproduces them
(\Cref{app:verification}).  By homogeneity,
$C[av(\cdot/T)]=a^{1+\delta}T^{2-\gamma}C[v]$, so each witness scales into
every rate cap and every horizon without changing sign.

\begin{table}[ht]
\centering
\caption{Sparse-comb manipulations at calibrated exponents.  All
endpoints, rates, and exponents are rational; the last column is the
certified upper endpoint of the interval enclosure of the cost, rounded
toward zero.}
\label{tab:certificates}
\begin{tabular}{@{}ccrrcrr@{}}
\toprule
$\delta$ & $\gamma$ & $M$ & $\ell$ & $s$ & blocks & certified cost bound \\
\midrule
$1/2$ & $1/2$ & 64  & 16 & $3/10$ & 129 & $-0.10963$ \\
$3/5$ & $1/2$ & 256 & 64 & $3/10$ & 513 & $-0.17409$ \\
$1/2$ & $3/5$ & 256 & 64 & $3/10$ & 513 & $-0.044262$ \\
\bottomrule
\end{tabular}
\end{table}

\subsection{Robustness and the admissible-class boundary}\label{sec:robustness}

The main theorem uses piecewise-constant controls because that is the
smallest class on which arbitrary pointwise $f$ is always meaningful.  Under
mild regularity, the manipulation is not an artifact of jumps.

\begin{proposition}[Smooth capped witnesses]\label{prop:smooth}
Assume the hypotheses of \Cref{prop:normalized}.  Suppose additionally that
$g$ is bounded and Borel measurable on $[-B,B]$.  If $g$ is nonlinear on
that interval, there exists $v\in C_c^\infty(0,T)$ such that
\begin{equation}\label{eq:smooth-witness}
 \int_0^T v(t)\dd t=0,
 \qquad \|v\|_\infty\le B,
 \qquad C_{H,g}[v]<0.
\end{equation}
\end{proposition}

\begin{proof}
Choose the piecewise-constant witness $v_0$ from
\Cref{prop:normalized} with support strictly inside $(0,T)$.  Extend it by
zero and convolve with a nonnegative smooth mollifier.  Convolution preserves
the integral and does not increase the $L^\infty$ norm.  Outside shrinking
neighborhoods $E_\varepsilon$ of the finitely many jumps, the mollification
$v_\varepsilon$ equals $v_0$ exactly.  The difference between the two cost
integrands is supported where $t\in E_\varepsilon$ or
$s\in E_\varepsilon$ and is bounded by a constant multiple of $|H(t,s)|$.
Absolute continuity of the $L^1$ integral gives
$C_{H,g}[v_\varepsilon]\to C_{H,g}[v_0]<0$.  Borel measurability makes the
mollified composition measurable.
\end{proof}

For homogeneous powers, scaling makes the instability local in much stronger
topologies.

\begin{corollary}[Smooth local instability]\label{cor:local}
Let $0<\gamma<1$ and $\delta>0$, $\delta\ne1$.  There is
$w\in C_c^\infty(0,T)$ with
$C_{H_\gamma,f_\delta}[w]<0$.  Moreover,
\begin{equation}\label{eq:homogeneous-scaling}
 C_{H_\gamma,f_\delta}[aw]
 =a^{1+\delta}C_{H_\gamma,f_\delta}[w]<0
 \qquad(a>0).
\end{equation}
Consequently every neighborhood of zero in the standard locally convex
topology of $C_c^\infty(0,T)$ contains a manipulation.  Positive bounds on
speed, on finitely many homogeneous Sobolev or $C^k$ seminorms, on total
variation, or on $L^p$ norms do not by themselves restore safety.
\end{corollary}

The switching complexity, by contrast, is essential.  Each witness selected
by the proof has finitely many pieces, but the admissible class contains no
uniform bound on that number.

For $H(t,s)=(t-s)^{-1/2}$ and $f(x)=\operatorname{sgn}(x)|x|^2$, for
instance, every nonzero two-piece round trip has strictly positive cost:
this is the case $\gamma=1/2\ge2-\log_23$ of
\Cref{thm:fractional-two-block-phase}, in which no superlinear power admits
a two-block manipulation, although \Cref{cor:all-powers} supplies a finite
one.

This rules out a tempting but invalid inference: checking
all two-rate schedules, or observing positivity on a coarse numerical grid,
does not establish no manipulation for the unrestricted finite class.  The
same warning applies to a common minimum dwell time.
\Cref{sec:friction-complexity} quantifies the block count for the
power-law kernel; for a general $L^1$ kernel the switching level is
finite but not uniform.

\section{Scalar state-outside and mixed classifications}\label{sec:outside}

We now turn from a nonlinearity applied to the rate before propagation to a
nonlinearity applied to the state after propagation.  The change of order
replaces affine rigidity by a kernel classification.

\subsection{Universal state-outside passivity}

For a fixed horizon $T$, write
\begin{equation}\label{eq:U-property}
 \mathsf U_T(G):\qquad
 \cost^{\mathrm{out}}_{G,h,T}[v]
 =\int_0^Tv(t)h((G*v)(t))\dd t\ge0
\end{equation}
for every $v\in\pc(T)$ and every $h\in\readouts$.

Define two dead-zone readouts
\begin{align}
 h_*(x)&=(x\wedge0)+(x-1)_+,
 \label{eq:fixed-deadzone}\\
 h_\dagger(x)&=
 \begin{cases}
 -e^{-1/x^2},&x<0,\\
 0,&0\le x\le1,\\
 e^{-1/(x-1)^2},&x>1.
 \end{cases}
 \label{eq:smooth-deadzone}
\end{align}
The second function is bounded and $C^\infty$ on $\R$.

\begin{theorem}[Scalar kernel converse]\label{thm:outside-cp}
Let $G\in L^1(0,T)$ be real.  The following are equivalent:
\begin{enumerate}[label=\textup{(\roman*)}]
\item $\mathsf U_T(G)$ holds;
\item the same inequality holds for every
  $h\in\readouts_{\mathrm{an}}$;
\item $G$ is completely positive on $[0,T]$.
\end{enumerate}
Without a prescribed rate cap, these conditions are also equivalent to
safety for the single fixed readout $h_*$ in
\eqref{eq:fixed-deadzone}, and separately to safety for the single fixed
bounded smooth readout $h_\dagger$ in \eqref{eq:smooth-deadzone}.

The equivalences \textup{(i)}$\Leftrightarrow$\textup{(iii)} and
\textup{(ii)}$\Leftrightarrow$\textup{(iii)} remain valid under every fixed
positive symmetric rate cap.  The two fixed-readout equivalences are stated
without a cap.

If $G\ne0$, these conditions are further equivalent to the existence of a
complementary measure of the form \eqref{eq:complement-form}.  If complete
positivity fails, an ordinary finite piecewise-constant negative-cost input
exists; it may be chosen inside every positive cap when the readout is
allowed to vary within the universal or analytic class.
\end{theorem}

The zero kernel is included in the resolvent statement: $r_a=0$ and $s_a=1$.
It is split off only from the complementary-measure formulation, because no
$L$ can satisfy $L*0=1$.

\subsection{The first resolvent is forced to be nonnegative}

\begin{lemma}[Negative-hinge resolvent probe]\label{lem:r-resolvent-probe}
If \eqref{eq:U-property} holds for the readout $h_-(x)=x\wedge0$, then
$r_a\ge0$ almost everywhere for every $a>0$.
\end{lemma}

\begin{proof}
Fix $a>0$, abbreviate $r=r_a$, and let $w\ge0$ be a bounded finite step
function.  Set
\begin{equation}\label{eq:r-probe-control}
 D=\frac1a r*w,\qquad v=w-r*w=w-aD.
\end{equation}
The resolvent identity gives
\begin{equation}\label{eq:r-resolvent-algebra}
 \frac1a r+G*r=G,\qquad G*v=D.
\end{equation}
The input $v$ is bounded and is admissible by
\Cref{lem:finite-step-closure}.  Since $h_-(D)=D$ on $\{D<0\}$ and
vanishes elsewhere,
\begin{align}
 \cost^{\mathrm{out}}_{G,h_-,T}[v]
 &=\int_{\{D<0\}}(wD-aD^2)\dd t
 \le-a\int_0^T(D_-)^2\dd t\le0.
 \label{eq:r-probe-cost}
\end{align}
Safety forces equality, hence $r_a*w\ge0$ for every nonnegative finite step
$w$.  With the causal approximate identity
$w_n=n\one_{(0,1/n)}$, one has $r_a*w_n\to r_a$ in $L^1(0,T)$.
Therefore $r_a\ge0$ almost everywhere.
\end{proof}

\subsection{The second resolvent is forced to be nonnegative}

\begin{lemma}[Positive-hinge resolvent probe]\label{lem:s-resolvent-probe}
Suppose $r_a\ge0$ and \eqref{eq:U-property} holds for
$h_{1/a}(x)=(x-1/a)_+$.  Then $s_a\ge0$ on $[0,T]$.
\end{lemma}

\begin{proof}
Resolvent algebra gives
\begin{equation}\label{eq:s-from-r}
 s_a=1-r_a*1,\qquad s_a+aG*s_a=1.
\end{equation}
Thus $s_a$ is absolutely continuous, $s_a(0)=1$, and $r_a\ge0$ makes it
nonincreasing.  For the bounded input $v=s_a$, its state is
\begin{equation}\label{eq:s-probe-state}
 D=G*s_a=\frac{1-s_a}{a},\qquad
 h_{1/a}(D)=\frac{(-s_a)_+}{a}.
\end{equation}
Consequently,
\begin{equation}\label{eq:s-probe-cost}
 \cost^{\mathrm{out}}_{G,h_{1/a},T}[s_a]
 =-\frac1a\int_0^T(s_a^-)^2\dd t.
\end{equation}
Safety makes the right side zero.  Continuity of $s_a$ then yields
$s_a\ge0$ everywhere.
\end{proof}

\begin{proof}[Proof of necessity in \Cref{thm:outside-cp}]
Universal safety contains the two readout probes above, so both resolvents
are nonnegative for every $a>0$.  This is complete positivity.

For the analytic formulation, use
\begin{align}
 n_\varepsilon(x)
 &=x-\varepsilon\log(1+e^{x/\varepsilon})+\varepsilon\log2,
 \label{eq:analytic-negative-hinge}\\
 p_{\varepsilon,c}(x)
 &=\varepsilon\log(1+e^{(x-c)/\varepsilon})
   -\varepsilon\log(1+e^{-c/\varepsilon}).
 \label{eq:analytic-positive-hinge}
\end{align}
Both are real analytic, strictly increasing, and vanish at zero.  On compact
sets they converge uniformly as $\varepsilon\downarrow0$ to $x\wedge0$ and
$(x-c)_+$, respectively.  A fixed bounded probe has compact state range, so
its cost is continuous under this uniform readout approximation.  The hinge
inequalities, and hence complete positivity, follow.

Under a cap, scale the first probe and its negative hinge jointly.  For the
second probe use the scaled hinge $(x-\lambda/a)_+$ after replacing $s_a$ by
$\lambda s_a$.  Analytic approximation preserves the strict negative
margin.  This proves cap-local necessity for the variable readout classes.
\end{proof}

\subsection{One fixed dead zone is enough}

For completeness, the universal readout quantifier can be collapsed without
a cap.  In the first-resolvent construction, multiply $w,D,v$ by a small
positive scalar so that $D_+\le1$.  Then $h_*(D)=D$ on $D<0$ and vanishes
on the attained nonnegative range; \eqref{eq:r-probe-cost} again forces
$r_a\ge0$.  Hence $s_a\le1$.  Now take
\begin{equation}\label{eq:fixed-deadzone-s-probe}
 v=a s_a,\qquad D=1-s_a.
\end{equation}
The fixed readout vanishes when $0\le s_a\le1$ and equals $-s_a$ when
$s_a<0$.  Therefore
\begin{equation}\label{eq:fixed-deadzone-s-cost}
 \int_0^Ta s_a(t)h_*(1-s_a(t))\dd t
 =-a\int_0^T(s_a^-)^2\dd t,
\end{equation}
which forces $s_a\ge0$.

The smooth dead-zone law $h_\dagger$ has the same sign geometry.  On the
first probe, $v>0$ and $h_\dagger(D)<0$ wherever $D<0$; on the second,
$v<0$ and $h_\dagger(D)>0$ wherever $s_a<0$.  Either negative set would
produce strict negative cost.  This proves both fixed-readout converses.
The threshold at one explains why no arbitrary fixed-cap statement follows
for these particular fixed functions.

\subsection{The storage--dissipation identity}

The following is the convolution chain rule of \citet[identity~(9)]{Zacher2008},
in the form of \citet[Lemma~2.2]{VergaraZacher2015}, written with its
endpoint storage and extended below to bounded-variation and singular
complementary densities.

\begin{proposition}[Complementary-kernel identity]\label{prop:cp-identity}
Let $\beta\in\R$, $\ell\in BV(0,T)$, and $D\in W^{1,1}(0,T)$ with
$D(0)=0$.  Put $Q=\beta D+\ell*D$, $v=Q'$, and, for $h\in\readouts$,
\begin{equation}\label{eq:phi-def}
 \Phi_h(x)=xh(x)-F_h(x)\ge0.
\end{equation}
Then
\begin{align}
 \int_0^Tv(t)h(D(t))\dd t
={}&\beta F_h(D(T))+\bigl(\ell*F_h(D)\bigr)(T)
 +\int_0^T\ell(t)\Phi_h(D(t))\dd t
 \notag\\
&+\int_{(0,T)}(-\dd\ell)(r)
  \int_r^T\mathcal B_h(D(t-r),D(t))\dd t.
\label{eq:cp-storage-identity}
\end{align}
If $\ell$ is nonnegative and nonincreasing, the last two terms are
nonnegative, independently of the sign of $\beta$.  The identity also holds
for a nonnegative nonincreasing $\ell\in L^1(0,T)$ that is unbounded at
zero, with the final integral interpreted against its Stieltjes measure.

In particular, if $G\ne0$ is completely positive,
$L=\beta\delta_0+\ell(t)\dd t$ satisfies \eqref{eq:complement-form},
$v\in\pc(T)$, and $D=G*v$, then \eqref{eq:cp-storage-identity} applies
and every term on its right is nonnegative.
\end{proposition}

\begin{proof}
First suppose $\ell\in BV(0,T)$ and use its right-continuous
representative.  Since $D(0)=F_h(D(0))=0$, Stieltjes differentiation gives,
for almost every $t$,
\[
 (\ell*D)'(t)=\ell(0+)D(t)
      +\int_{(0,t)}D(t-r)\dd\ell(r),
\]
and the analogous formula with $D$ replaced by $F_h(D)$.  The elementary
identity
\[
 h(x)y-F_h(y)=\Phi_h(x)-\mathcal B_h(y,x)
\]
therefore yields
\begin{align*}
 h(D(t))(\ell*D)'(t)-(\ell*F_h(D))'(t)
 ={}&\ell(t)\Phi_h(D(t))\\
 &+\int_{(0,t)}\mathcal B_h(D(t-r),D(t))(-\dd\ell)(r).
\end{align*}
Integrating, using Fubini for the finite signed measure $\dd\ell$, and
adding $\beta\int_0^Th(D)D'=\beta F_h(D(T))$ proves
\eqref{eq:cp-storage-identity}.  No sign assumption on $\beta$ was used.
The stated signs follow from convexity and $h(0)=0$.

For a nonnegative nonincreasing density with a possible singularity at
zero, write $\ell_n=\min(\ell,n)$.  Young's inequality gives
\begin{equation}\label{eq:singular-complement-closure}
 \|((\ell_n-\ell)*D)'\|_1
 \le\|\ell_n-\ell\|_1\|D'\|_1\longrightarrow0.
\end{equation}
The left sides of the truncated identities consequently converge.  The
endpoint term with $\beta$ is unchanged, while each term containing
$\ell_n$ or $-\dd\ell_n$ increases to its counterpart for $\ell$ by
monotone convergence.  This proves the identity also for the full
Stieltjes measure, including atoms and singular-continuous parts.

Finally, a finite-step input makes $D=G*v$ a finite linear combination of
translates of the primitive of $G$, so $D\in W^{1,1}(0,T)$ and $D(0)=0$.
Convolving $L*G=1$ with $v$ gives $Q=\beta D+\ell*D$ and $Q'=v$.
For the completely positive complement, $\beta\ge0$ as well, making its
endpoint term nonnegative.
\end{proof}

\begin{proof}[Completion of the proof of \Cref{thm:outside-cp}]
If $G=0$, the cost vanishes.  If $G\ne0$ is completely positive,
\Cref{prop:cp-identity} proves safety for every $h\in\readouts$ and every
input.  This also proves sufficiency for each smaller readout or capped
class.  Necessity was proved above.

If complete positivity fails, choose $a$ for which one resolvent has a
negative part.  The corresponding probe has strict negative cost.  By
\Cref{lem:finite-step-closure}, an ordinary finite piecewise-constant probe
retains the strict sign.  The cap and analytic refinements were already
constructed in the necessity proof.
\end{proof}

\subsection{Consequences and sharp boundaries}

Complete positivity forces $G\ge0$ almost everywhere because
$r_a/a\to G$ in $L^1(0,T)$ as $a\downarrow0$.  It rules out an initial
delay for a nonzero universally safe kernel.  It does not force
complete monotonicity or pointwise decrease.  For example,
\begin{equation}\label{eq:cp-not-cm-example}
 \beta=1,\qquad \ell(t)=(1+t)e^{-t},\qquad
 G(t)=\frac13+\frac23e^{-3t/2}
 \cos\left(\frac{\sqrt3}{2}t\right)
\end{equation}
satisfies $\ell'=-te^{-t}\le0$ and $L*G=1$, hence is completely
positive, but $G^{(4)}(0+)=-3<0$ and $G$ is not completely monotone.

Round-trip safety alone does not imply complete positivity without a tail
condition.  If $G\equiv g\ne0$ and $q(t)=\int_0^tv(s)\dd s$, then
\begin{equation}\label{eq:pure-permanent-path}
 \cost^{\mathrm{out}}_{g,h,T}[v]
 =\frac{F_h(gq(T))}{g}.
\end{equation}
Every round trip has zero cost, even for $g<0$, although a negative constant
kernel is not completely positive.  The correct replacement for this failed
implication is developed in \Cref{sec:roundtrip-storage}.

\subsection{The mixed architecture}

Let the accessible rate set be $X=\R$ or $X=[-B,B]$ with $B>0$.  Let
$f:X\to\R$ be arbitrary and finite-valued.

\begin{theorem}[Universal mixed classification]\label{thm:mixed-classification}
Fix $T>0$ and a nonzero real $G\in L^1(0,T)$.  The following are equivalent:
\begin{enumerate}[label=\textup{(\roman*)}]
\item $\cost^{\mathrm{mix}}_{G,f,h,T}[v]\ge0$ for every
  $v\in\pc(T;X)$ and every $h\in\readouts$;
\item there is $\lambda\in\R$ such that
  \begin{equation}\label{eq:mixed-classification}
   f(x)=\lambda x\quad(x\in X),\qquad
   \lambda G\text{ is completely positive on }[0,T].
  \end{equation}
\end{enumerate}
The same equivalence holds with $\readouts$ replaced by
$\readouts_{\mathrm{an}}$, and under every fixed positive symmetric rate
cap.  The zero effective kernel $\lambda G=0$ is included.
\end{theorem}

\begin{proof}
Condition \textup{(i)} contains round trips and the identity readout.
For that readout, \eqref{eq:mixed-cost} is the rate-inside convolution cost
of \Cref{thm:exact}.  Since $G$ is nonzero, the convolution corollary forces
$f(x)=\lambda x$ on the entire accessible set, including $f(0)=0$.
After substitution,
\begin{equation}\label{eq:mixed-effective-kernel}
 G*f(v)=(\lambda G)*v.
\end{equation}
The remaining inequality is precisely universal state-outside passivity for
$K=\lambda G$.  \Cref{thm:outside-cp} forces $K$ to be completely positive.
The identity belongs to $\readouts_{\mathrm{an}}$, and both component
converses are cap-local, proving all necessity variants.

Conversely, if $K=0$, the state and cost vanish.  If $K\ne0$ is completely
positive, \Cref{prop:cp-identity} gives nonnegative cost for every input and
every $h\in\readouts$.  Restricting the input or readout class preserves the
conclusion.
\end{proof}

\begin{corollary}[All-horizon mixed classification]\label{cor:mixed-global}
Let $G\in L^1_{\mathrm{loc}}(0,\infty)$ be nonzero.  Universal all-input
safety on every finite horizon is equivalent to
\begin{equation}\label{eq:mixed-global}
 f=\lambda\id,\qquad
 \lambda G\text{ completely positive on every finite horizon}.
\end{equation}
\end{corollary}

\begin{proof}
Choose one horizon on which $G$ is nonzero to fix the single slope
$\lambda$; the observation at the end of \Cref{sec:framework} carries it
to every horizon, where the fixed-horizon theorem applies.
\end{proof}

\section{Round-trip reduction and permanent storage}\label{sec:roundtrip-storage}

The all-input classification in \Cref{thm:outside-cp} does not classify
round trips.  This section identifies the missing quotient.

\subsection{Uniformly vanishing signed tails}

\begin{definition}[Uniformly vanishing tail]\label{def:vanishing-tail}
A real $G\in L^1_{\mathrm{loc}}(0,\infty)$ has a uniformly vanishing tail if
\begin{equation}\label{eq:uniform-tail}
 \eta_G(R):=\esssup_{t\ge R}|G(t)|\longrightarrow0
 \qquad(R\to\infty).
\end{equation}
No sign or monotonicity is imposed.
\end{definition}

\begin{theorem}[Remote compensation]\label{thm:remote-compensation}
Let $G$ satisfy \eqref{eq:uniform-tail}, and fix a continuous readout
$h:\R\to\R$ with $h(0)=0$.  The following are equivalent:
\begin{enumerate}[label=\textup{(\roman*)}]
\item $\cost^{\mathrm{out}}_{G,h,T}[v]\ge0$ for every finite horizon and
  every $v\in\rt(T)$;
\item $\cost^{\mathrm{out}}_{G,h,T}[v]\ge0$ for every finite horizon and
  every $v\in\pc(T)$.
\end{enumerate}
The equivalence preserves any fixed positive symmetric rate cap.
\end{theorem}

\begin{proof}
Only \textup{(i)}$\Rightarrow$\textup{(ii)} requires proof.  Let $u$ be a
finite-step input supported on $[0,T_0]$ and put
\begin{equation}\label{eq:input-volume}
 m=\int_0^{T_0}u(t)\dd t.
\end{equation}
For $R>T_0$ and $L>0$, append the constant block
\begin{equation}\label{eq:remote-block}
 w_{R,L}(t)=-\frac mL\one_{[R,R+L]}(t).
\end{equation}
Then $u+w_{R,L}$, with an idle gap inserted, is an ordinary finite-step
round trip.  On the compensating block,
\begin{align}
 |(G*u)(t)|
 &\le\|u\|_1\eta_G(R-T_0),
 \label{eq:remote-inherited-bound}\\
 \esssup_{t\in[R,R+L]}|(G*w_{R,L})(t)|
 &\le\frac{|m|}{L}\int_0^L|G(r)|\dd r.
 \label{eq:remote-self-bound}
\end{align}
The first bound vanishes as $R\to\infty$.  The second vanishes as
$L\to\infty$ because uniform tail decay makes the Ces\`aro mean of $|G|$
converge to zero.  Thus the state on the remote block converges essentially
uniformly to zero.  Its cost satisfies
\begin{equation}\label{eq:remote-cost-bound}
 |\cost_{\mathrm{remote}}|
 \le |m|\sup_{|z|\le\varepsilon_{R,L}}|h(z)|\longrightarrow0.
\end{equation}
Causality leaves the original cost unchanged and the idle gap costs zero.
Taking the scalar cost limit in the round-trip inequality gives
$\cost^{\mathrm{out}}_{G,h,T_0}[u]\ge0$.

Every approximating strategy is itself finite piecewise constant; no limiting
control is added to the admissible class.  Under a cap $B$, choose
$L\ge|m|/B$.
\end{proof}

\begin{corollary}[Decaying scalar classifications]\label{cor:decaying-cp}
For every nonzero scalar kernel with a uniformly vanishing tail,
\begin{equation}\label{eq:decaying-cp-equivalence}
 \begin{aligned}
 &\text{universal round-trip safety on every horizon}\\
 &\quad\Longleftrightarrow
 \text{universal all-input safety on every horizon}\\
 &\quad\Longleftrightarrow
 G\text{ is completely positive on every horizon}.
 \end{aligned}
\end{equation}
For the mixed model, the corresponding classification is
\begin{equation}\label{eq:decaying-mixed-equivalence}
 f=\lambda\id,\qquad \lambda G\text{ globally completely positive}.
\end{equation}
Both statements are cap-local and remain valid with the universal class
replaced by all analytic strictly increasing normalized readouts.
\end{corollary}

\begin{proof}
Combine \Cref{thm:remote-compensation} with
\Cref{thm:outside-cp,thm:mixed-classification}.  The zero effective kernel
is again split off before a complementary measure is invoked.
\end{proof}

\subsection{Asymptotically permanent memory}

Let
\begin{equation}\label{eq:permanent-decomposition}
 G(t)=g+H(t),\qquad
 \esssup_{t\ge R}|H(t)|\longrightarrow0.
\end{equation}
For $h\in\readouts$, define
\begin{equation}\label{eq:permanent-storage}
 \storage_{g,h}(m)=
 \begin{cases}
 \displaystyle\frac{F_h(gm)}g,&g\ne0,\\[0.6em]
 0,&g=0.
 \end{cases}
\end{equation}
The quantity can be negative when $g<0$.  It is the recoverable scalar
storage associated with the permanent state direction, not a positive energy
by definition.

\begin{theorem}[Permanent-storage quotient]\label{thm:permanent-quotient}
Fix one continuous normalized readout $h$ and a kernel satisfying
\eqref{eq:permanent-decomposition}.  The following are equivalent:
\begin{enumerate}[label=\textup{(\roman*)}]
\item every finite-step round trip on every finite horizon has
  nonnegative state-outside cost;
\item every finite-step input $u$ on every finite horizon satisfies
  \begin{equation}\label{eq:storage-domination}
   \cost^{\mathrm{out}}_{G,h,T}[u]
   \ge \storage_{g,h}\!\left(\int_0^Tu(t)\dd t\right).
  \end{equation}
\end{enumerate}
The equivalence preserves any fixed positive symmetric rate cap.  For
$G\equiv g$, \eqref{eq:storage-domination} is an equality for every input.
\end{theorem}

\begin{proof}
Take a finite-step $u$ supported on $[0,T_0]$ with volume $m$ and append the
remote block \eqref{eq:remote-block}.  At $t=R+L\tau$ on that block, the
permanent state is exactly
\begin{equation}\label{eq:permanent-path}
 gm+g\left(-\frac mL\right)L\tau=gm(1-\tau).
\end{equation}
The transient error obeys
\begin{equation}\label{eq:permanent-transient-error}
 \|E_{R,L}\|_{L^\infty_{\mathrm{ess}}(0,1)}
 \le\|u\|_1\eta_H(R-T_0)
 +\frac{|m|}{L}\int_0^L|H(r)|\dd r\longrightarrow0,
\end{equation}
first as the gap and then the block length tend to infinity.  Uniform
continuity of $h$ on a compact neighborhood of the segment from $gm$ to
zero gives
\begin{align}
 \cost_{\mathrm{remote}}
 &=-m\int_0^1h(gm(1-\tau)+E_{R,L}(\tau))\dd\tau
 \notag\\
 &\longrightarrow -m\int_0^1h(gm(1-\tau))\dd\tau
 =-\frac{F_h(gm)}g
 \label{eq:remote-permanent-limit}
\end{align}
when $gm\ne0$.  If $m=0$, both sides vanish; if $g=0$, the decaying-tail
argument applies.  Round-trip safety of the augmented strategies therefore
gives \eqref{eq:storage-domination}.  Conversely, a round trip has $m=0$ and
$\storage_{g,h}(0)=0$.

If $G\equiv g\ne0$, write $q(t)=\int_0^tu(s)\dd s$ and $D=gq$.
Then
\begin{equation}\label{eq:pure-permanent-equality}
 \cost^{\mathrm{out}}_{g,h,T}[u]
 =\frac1g\int_0^TD'(t)h(D(t))\dd t
 =\frac{F_h(gm)}g.
\end{equation}
The cap statement follows as in \Cref{thm:remote-compensation}.
\end{proof}

\begin{remark}[The sign of the quotient]\label{rem:storage-sign}
For $h(x)=x$, $\storage_{g,h}(m)=gm^2/2$; for $h(x)=x^3$,
$\storage_{g,h}(m)=g^3m^4/4$.  The quotient is negative for $g<0$ and
odd $h$, as the pure-permanent equality
\eqref{eq:pure-permanent-equality} requires.
\end{remark}

\Cref{thm:permanent-quotient} changes how permanent memory is compared with
decaying memory.  The decaying theory tests nonnegative all-input supply;
the permanent theory tests supply after subtracting the recoverable
endpoint quantity $\storage_{g,h}(m)$.  Round trips live on the quotient of
the input space by total volume, and \eqref{eq:permanent-storage} is the
endpoint term that remote closure exposes.  The permanent-shift and Prony
classifications of the next section are storage-domination theorems in this
sense.

\section{Permanent shifts and finite-Prony frontiers}\label{sec:prony}

This section turns the permanent quotient into safe classes and
separators.  We begin with a broad infinite-dimensional sufficient cone,
then identify a different stable-inverse cone that is maximal within its
class, and finally close the entire two-mode family.

\subsection{Every real shift of completely monotone decay}

\begin{definition}[LICM transient]\label{def:licm}
A function $H:(0,\infty)\to[0,\infty)$ is a decaying locally integrable
completely monotone (LICM) transient if $H\in L^1(0,T)$ for every $T<\infty$,
$H(t)\to0$, and
\begin{equation}\label{eq:bernstein-representation}
 H(t)=\int_{(0,\infty)}e^{-\lambda t}\rho(\dd\lambda)
\end{equation}
for a positive measure $\rho$.
\end{definition}

\begin{theorem}[Arbitrary real permanent shifts]\label{thm:licm-shifts}
Let $H$ be a decaying LICM transient, let $g\in\R$, and put $G=g+H$.
For every finite horizon, every $v\in\pc(T)$ of volume $m$, and every
$h\in\readouts$,
\begin{equation}\label{eq:licm-storage-bound}
 \cost^{\mathrm{out}}_{G,h,T}[v]\ge\storage_{g,h}(m).
\end{equation}
Hence every real permanent shift of $H$ is universally safe on round
trips, including shifts for which $G$ changes sign or tends to a negative
level.  The conclusion is valid under every rate cap.
\end{theorem}

\begin{proof}
For $H=0$, the chain rule gives equality in
\eqref{eq:licm-storage-bound}.  Otherwise first take a positive Prony transient
\begin{equation}\label{eq:positive-prony}
 H(t)=\sum_{i=1}^n a_i e^{-\lambda_i t},\qquad
 a_i>0,\quad 0<\lambda_1<\cdots<\lambda_n.
\end{equation}
Let
\begin{equation}\label{eq:prony-R}
 R(p)=p\laplace G(p)
 =g+\sum_{i=1}^na_i\frac{p}{p+\lambda_i},\qquad
 b=G(0+)=g+\sum_i a_i.
\end{equation}
Assume first $g\ne0$ and $b\ne0$.  The zeros $-\mu_j$ of $R$ are real and
simple and
\begin{equation}\label{eq:positive-prony-inverse}
 \frac1{R(p)}=\frac1b+\sum_{j=1}^n\frac{c_j}{p+\mu_j},\qquad
 c_j=\frac1{R'(-\mu_j)}>0.
\end{equation}
Indeed,
$R'(p)=\sum_i a_i\lambda_i/(p+\lambda_i)^2>0$ between the real poles, so
the zeros interlace the poles and all residues are positive.  Evaluation at
$p=0$ gives
\begin{equation}\label{eq:prony-weight-sum}
 \frac1b+\sum_{j=1}^n\frac{c_j}{\mu_j}=\frac1g.
\end{equation}

Let $D=G*v$ and define inverse modes
\begin{equation}\label{eq:inverse-modes}
 Z_j(t)=\mu_j\int_0^te^{-\mu_j(t-s)}D(s)\dd s,\qquad
 Z_j'=\mu_j(D-Z_j).
\end{equation}
The inverse transfer yields
\begin{equation}\label{eq:inverse-volume-decomp}
 Q=1*v=\frac1bD+\sum_{j=1}^n\frac{c_j}{\mu_j}Z_j.
\end{equation}
Using $F_h'=h$ and adding and subtracting $h(Z_j)$ gives the identity
\begin{align}
 \cost^{\mathrm{out}}_{G,h,T}[v]
={}&\frac1bF_h(D(T))
 +\sum_{j=1}^n\frac{c_j}{\mu_j}F_h(Z_j(T))
 \notag\\
&+\sum_{j=1}^nc_j\int_0^T
 (D-Z_j)(h(D)-h(Z_j))\dd t.
\label{eq:star-bregman}
\end{align}
Every integral on the second line is nonnegative.

Put $w_0=1/b$, $Y_0=D(T)$ and
$w_j=c_j/\mu_j$, $Y_j=Z_j(T)$.  Then
\begin{equation}\label{eq:endpoint-barycenter}
 \sum_{j=0}^nw_j=\frac1g,\qquad
 \sum_{j=0}^nw_jY_j=m.
\end{equation}
If $g>0$, all weights are positive and ordinary Jensen gives
\begin{equation}\label{eq:positive-jensen}
 \sum_{j=0}^nw_jF_h(Y_j)\ge\frac1gF_h(gm).
\end{equation}
If $g<0$, exactly one weight is negative.  Write it as $-A$, let the
positive weights be $p_i$, and set
$P=\sum p_i$, $B=A-P=-1/g>0$.  The barycenter equation implies
\begin{equation}\label{eq:signed-jensen-barycenter}
 Y_- =\frac PA\left(\sum_i\frac{p_i}{P}Y_i\right)+\frac BA(gm).
\end{equation}
Applying Jensen first to this convex combination and then to the positive
weights gives
\begin{equation}\label{eq:signed-jensen}
 \sum_{j=0}^nw_jF_h(Y_j)\ge-BF_h(gm)=\frac1gF_h(gm).
\end{equation}
Combining \eqref{eq:star-bregman} with the endpoint inequality proves the
claim when $g,b\ne0$.  For either boundary, choose $g_n\to g$ with
$g_n\ne0$ and $g_n+\sum_i a_i\ne0$.  On a fixed horizon,
$\|(G_n-G)*v\|_\infty\le|g_n-g|\|v\|_1$, so the costs converge.
The storage term is continuous also at $g=0$, since
\[
 \frac{F_h(g_nm)}{g_n}
 =m\int_0^1h(g_nm u)\dd u\longrightarrow0.
\]
Thus both $b=0$ and $g=0$ retain the same inequality.

For a general LICM $H$, approximate the Bernstein measure in
\eqref{eq:bernstein-representation} by finite positive atomic measures so
that the resulting $H_N$ converge to $H$ in $L^1(0,T)$.  This remains true
when $H(0+)=\infty$: truncate the measure to $[1/N,N]$ and discretize it
there.  For a fixed bounded step input,
\begin{equation}\label{eq:licm-state-convergence}
 \|(H_N-H)*v\|_\infty
 \le\|v\|_\infty\|H_N-H\|_1\longrightarrow0.
\end{equation}
Continuity of $h$ on the common compact state range gives convergence of
costs, while the right side of \eqref{eq:licm-storage-bound} is unchanged.
Passing to the limit completes the proof.
\end{proof}

\begin{corollary}[Monotone readouts after completely monotone memory]
\label{cor:concave-after-powerlaw}
Let $H$ be a decaying LICM transient, for instance $H(t)=t^{-\gamma}$ with
$0<\gamma<1$, and let $h\in\readouts$ be arbitrary, for instance
$h(x)=\operatorname{sgn}(x)|x|^{\delta}$ with $\delta>0$.  Then, for every
$g\ge0$, every finite horizon, and every $v\in\pc(T)$,
\begin{equation}\label{eq:concave-after-powerlaw}
 \cost^{\mathrm{out}}_{g+H,h,T}[v]\ge\storage_{g,h}\Bigl(\int_0^Tv\Bigr)\ge0.
\end{equation}
For $g<0$ the first inequality still holds, so every round trip is safe.
\end{corollary}

\begin{proof}
This is \Cref{thm:licm-shifts}; for $g\ge0$ the quotient is nonnegative
because $F_h\ge0$.
\end{proof}

Square-root impact applied after power-law memory is therefore free of
price manipulation for every input, in contrast with the same law applied
before the memory (\Cref{cor:all-powers}).  The kernel condition behind
this is complete positivity (\Cref{thm:outside-cp}), of which complete
monotonicity is the classical sufficient case.

\begin{corollary}[One-mode classification]\label{cor:one-mode}
For every $\lambda>0$ and every $g,a\in\R$,
\begin{equation}\label{eq:one-mode-classification}
 g+ae^{-\lambda t}\text{ is universally round-trip safe}
 \quad\Longleftrightarrow\quad a\ge0.
\end{equation}
The permanent coefficient is unrestricted.
\end{corollary}

\begin{proof}
Sufficiency is \Cref{thm:licm-shifts}.  For necessity, use $h(x)=x$ and
write $w'=v-\lambda w$ for the exponential state.  On a nonzero round trip,
\begin{equation}\label{eq:one-mode-identity-cost}
 \int_0^Tv(t)w(t)\dd t
 =\frac12w(T)^2+\lambda\int_0^Tw(t)^2\dd t>0,
\end{equation}
whereas the permanent quadratic contribution vanishes.  Thus $a<0$ gives
negative cost.
\end{proof}

\begin{corollary}[Mixed LICM classification]\label{cor:mixed-licm}
Suppose $H\not\equiv0$ is a decaying LICM transient and $G=g+H$.
Universal round-trip safety for the mixed model on every horizon holds
if and only if
\begin{equation}\label{eq:mixed-licm}
 f=\lambda\id,\qquad \lambda\ge0.
\end{equation}
The equivalence is cap-local.  If $H=0$ and $g\ne0$, linearity is still
forced but either sign of $\lambda$ is lossless on round trips.
\end{corollary}

\begin{proof}
The identity readout and rate-inside rigidity force $f=\lambda\id$.  The
permanent part contributes zero to the identity-readout cost of a round trip,
whereas the LICM transient has the strictly positive quadratic representation
obtained by integrating \eqref{eq:one-mode-identity-cost} against the
Bernstein measure.  Hence $\lambda<0$ fails.  If $\lambda\ge0$, the effective
kernel is another real permanent shift of a LICM transient and
\Cref{thm:licm-shifts} applies.
\end{proof}

\subsection{Stable complementary inverses}

The LICM-shift cone is not maximal.  A second mechanism is visible from
the inverse map between observed state and cumulative volume.

\begin{theorem}[Stable complementary-inverse classification]
\label{thm:stable-complement}
Assume
\begin{equation}\label{eq:stable-complement-hyp}
 G(t)=g+H(t),\quad g\ne0,\quad b=G(0+)\ne0,\quad
 \esssup_{t\ge R}|H(t)|\to0,
\end{equation}
and suppose
\begin{equation}\label{eq:stable-complement-transfer}
 \frac1{p\laplace G(p)}=\beta+\laplace\ell(p),\qquad
 \beta=\frac1b,\qquad \ell\in L^1(0,\infty)\cap BV(0,\infty).
\end{equation}
Then the following are equivalent:
\begin{enumerate}[label=\textup{(\roman*)}]
\item every round trip on every horizon is safe for every
  $h\in\readouts$;
\item $\ell$ has a nonnegative, nonincreasing representative;
\item every input satisfies
  $\cost^{\mathrm{out}}_{G,h,T}[v]\ge\storage_{g,h}(m)$ for every
  $h\in\readouts$.
\end{enumerate}
The endpoint gains automatically obey $gb>0$, and the equivalence is
cap-local.
\end{theorem}

\begin{proof}
\emph{Endpoint gains.}
For $L(p)=\beta+\laplace\ell(p)$, analytic continuation gives
$p\laplace G(p)L(p)=1$ on $\re p>0$; in particular $L(p)\ne0$ for real
$p>0$.  Local integrability and the uniform tail bound imply
$p\laplace H(p)\to0$ as $p\downarrow0$, by splitting the integral at a
fixed large threshold.  Hence
\begin{equation}\label{eq:stable-endpoint-masses}
 L(+\infty)=\beta=\frac1b,\qquad
 L(0+)=\beta+\int_0^\infty\ell(t)\dd t=\frac1g.
\end{equation}
These nonzero endpoint values have the same sign, so $gb>0$.
Convolution inversion gives
\begin{equation}\label{eq:stable-inverse-relation}
 Q=\beta D+\ell*D,\qquad Q(t)=\int_0^tv(s)\dd s.
\end{equation}

\emph{Sufficiency.}
Under \textup{(ii)}, \Cref{prop:cp-identity}, with arbitrary real $\beta$,
gives
\begin{align}
 \cost^{\mathrm{out}}_{G,h,T}[v]
 &\ge\beta F_h(D(T))
       +\int_0^T\ell(T-t)F_h(D(t))\dd t,
 \label{eq:delay-mixture-bound}\\
 m&=\beta D(T)+\int_0^T\ell(T-t)D(t)\dd t.
 \label{eq:delay-mixture-volume}
\end{align}
Set $P=\int_0^T\ell$ and $w=\beta+P$.  For $\beta>0$, Jensen bounds
the right side below by $wF_h(m/w)$.  For $\beta=-A<0$, one has
$P<A$ by \eqref{eq:stable-endpoint-masses}.  If $P>0$, put
$y=P^{-1}\int_0^T\ell(T-t)D(t)\dd t$ and $B=A-P=-w$.
The volume identity reads $D(T)=(P/A)y+(B/A)(m/w)$; Jensen applied to
this convex combination and to $y$ gives
\[
 -AF_h(D(T))+\int_0^T\ell(T-t)F_h(D(t))\dd t
 \ge -AF_h(D(T))+PF_h(y)\ge wF_h(m/w).
\]
For $P=0$ this is equality.  In either case $w\le w_\infty=1/g$ on
the same sign ray, and
\begin{equation}\label{eq:stable-perspective-derivative}
 \frac{\dd}{\dd w}\bigl[wF_h(m/w)\bigr]
 =-\Phi_h(m/w)\le0.
\end{equation}
Thus the cost is at least $w_\infty F_h(m/w_\infty)$, proving
\textup{(iii)}.  Taking $m=0$ gives \textup{(i)}.

\emph{Necessity: compact loops and their cell limits.}
Assume \textup{(i)}.  \Cref{thm:permanent-quotient} gives storage
domination, also for bounded measurable inputs by
\Cref{lem:finite-step-closure} and continuity of the storage term.
For $D\in C_c^\infty((0,\infty))$, set
$Q=\beta D+\ell*D$ and $v=\beta D'+\ell*D'$.
Then $v\in L^1\cap L^\infty$ and $G*v=D$.  If
$\supp D\subset[0,S]$ and $R>S$, then
\begin{equation}\label{eq:stable-compact-volume-tail}
 |Q(R)|\le\|D\|_\infty
              \int_{R-S}^R|\ell(r)|\dd r\longrightarrow0.
\end{equation}
Truncation at $R$ leaves the cost
$E_h(D)=\int_0^\infty vh(D)$ unchanged because $h(D)=0$ after $S$.
Storage domination gives $E_h(D)\ge\storage_{g,h}(Q(R))\to0$.

Now extend $D_0=\sum_{i=1}^nx_i\one_{[(i-1)\Delta,i\Delta)}$ by zero
and write $u_0=\ell*D_0'\in L^1(\R)$, extending $\ell$ causally.
For a smooth approximate identity $\rho_\varepsilon$ supported in
$(0,\varepsilon)$, the compact smooth loops
$D_\varepsilon=\rho_\varepsilon*D_0$ have a common bound, converge
almost everywhere to $D_0$, and satisfy
\begin{equation}\label{eq:stable-ramp-memory-convergence}
 (\ell*D_\varepsilon)'=\rho_\varepsilon*u_0\longrightarrow u_0
 \quad\text{in }L^1(\R).
\end{equation}
The feedthrough cost is exactly
$\beta\int h(D_\varepsilon)D_\varepsilon'=0$.
Boundedness and continuity of $h$, followed by dominated convergence
against $|u_0|$, imply $\int h(D_0)u_0\ge0$.
With $k_j=\int_{j\Delta}^{(j+1)\Delta}\ell$, the endpoint identity
$(\ell*D_0)(i\Delta)=\sum_{j\le i}x_jk_{i-j}$ gives
\begin{equation}\label{eq:toeplitz-accretivity}
 \begin{gathered}
 \int_\R h(D_0)u_0=\sum_{i=1}^n(A_nx)_ih(x_i),\\
 (A_n)_{ij}=\begin{cases}
 k_0,&i=j,\\
 k_{i-j}-k_{i-j-1},&i>j,\\
 0,&i<j.
 \end{cases}
 \end{gathered}
\end{equation}
By \Cref{cor:toeplitz-complete-accretivity}, testing every $x,h,n$ forces
\begin{equation}\label{eq:toeplitz-criterion}
 k_0\ge k_1\ge k_2\ge\cdots\ge0.
\end{equation}
This holds for every $\Delta>0$.  At Lebesgue points $0<s<t$, the cells
containing $s$ and $t$ are ordered for small $\Delta$.  Divide their
inequality by $\Delta$ and let $\Delta\downarrow0$; the averages converge
to $\ell(s)$ and $\ell(t)$.  Hence $\ell(s)\ge\ell(t)\ge0$, proving
\textup{(ii)}.

If \textup{(ii)} fails, this argument yields a finite negative cell form.
First fix a smooth loop retaining its strict negative cost; then choose
$R$ so that $E_h(D)<\storage_{g,h}(Q(R))$.
Finite-step approximation preserves this storage deficit, and
\Cref{thm:permanent-quotient} closes it by one remote block.
Apply \Cref{lem:cap-scaling} to the resulting finite witness.
\end{proof}

\subsection{The complete two-mode theorem}

Let
\begin{equation}\label{eq:two-mode-def}
 G(t)=g+a_1e^{-\lambda_1t}+a_2e^{-\lambda_2t},\qquad
 0<\lambda_1<\lambda_2,
\end{equation}
and define
\begin{equation}\label{eq:two-mode-moments}
 A=a_1+a_2,\quad b=g+A,\quad
 M_0=a_1\lambda_2+a_2\lambda_1,\quad
 M_\infty=a_1\lambda_1+a_2\lambda_2.
\end{equation}

\begin{theorem}[Two-mode classification]\label{thm:two-mode}
The kernel \eqref{eq:two-mode-def} is universally safe for every finite
piecewise-constant round trip on every horizon if and only if
\begin{equation}\label{eq:two-mode-iff}
 a_1\ge0\ \text{and }a_2\ge0,
 \qquad\text{or}\qquad
 gb>0\ \text{and }\frac1{p\laplace G(p)}=\frac1b+\laplace\ell(p)
 \text{ with }\ell\ge0\text{ nonincreasing}.
\end{equation}
In the second branch the inverse is automatically stable: under the
identity-readout gate \eqref{eq:two-mode-identity-gate} the polynomial $P$
in \eqref{eq:two-mode-P} has both roots in the open left half-plane, and
the residue test \eqref{eq:two-mode-nu-test} is the whole condition.  Both
alternatives are cap-local.  If $\lambda_1=\lambda_2$, the coefficients are
first combined; the condition is $a_1+a_2\ge0$.
\end{theorem}

\begin{proof}
The identity readout forces
\begin{equation}\label{eq:two-mode-identity-gate}
 M_0\ge0,\qquad M_\infty\ge0,
\end{equation}
because
\begin{equation}\label{eq:two-mode-spectrum}
 \re\laplace H(i\omega)
 =\frac{\lambda_1\lambda_2M_0+M_\infty\omega^2}
 {\left(\lambda_1^2+\omega^2\right)
  \left(\lambda_2^2+\omega^2\right)},\qquad H=G-g.
\end{equation}
If $a_1,a_2\ge0$, \Cref{thm:licm-shifts} proves safety for every $g$.
Suppose instead that the spectrum is signed and the identity gate holds.
Then $a_1a_2<0$ and $A>0$.

If $g,b\ne0$ and $gb>0$, put
\begin{equation}\label{eq:two-mode-P}
 P(p)=bp^2+\left[g(\lambda_1+\lambda_2)+M_0\right]p
      +g\lambda_1\lambda_2.
\end{equation}
The identity gate makes both zeros of $P/b$ lie in the open left half-plane.
Moreover,
\begin{equation}\label{eq:two-mode-inverse-density}
 \frac1{p\laplace G(p)}=\frac1b+\laplace\ell(p),\qquad
 \laplace\ell(p)=
 \frac{M_\infty p+A\lambda_1\lambda_2}{bP(p)}.
\end{equation}
The stable-complement theorem is therefore necessary and sufficient.

If $gb<0$, the gate forces $g<0<b$.  The inverse has one unstable pole
$\kappa>0$ and one stable pole $-\mu<0$,
\begin{equation}\label{eq:opposite-inverse}
 \frac1{p\laplace G(p)}
 =\frac1b+\frac{u}{p-\kappa}+\frac{s}{p+\mu},\qquad u>0>s.
\end{equation}
The four-cell construction of \Cref{prop:four-cell-separator} cancels the
unstable exponential moment exactly and isolates the negative stable edge,
giving a strict finite round-trip manipulation.

If $b=0$, the inverse has a polynomial part.  When $M_\infty>0$, the
identity in \Cref{prop:singular-first-order-separator} is the sum of a
positive derivative energy and a negative Bregman edge; a flat-readout loop
makes the latter dominate.  When $M_\infty=0$, the second-order identity in
\Cref{prop:singular-second-order-separator} is made negative by an asymmetric
smooth time-rescaling loop.  Both constructions are cap-local and close to
finite round trips.

Finally, if $g=0$, uniform decay reduces the problem to complete positivity.
For a signed spectrum and $M_0>0$, direct partial fractions give
\begin{equation}\label{eq:g0-negative-complement}
 \begin{gathered}
 \frac1{p\laplace G(p)}
 =\frac1A+\frac{\lambda_1\lambda_2}{M_0p}
       +\frac{\gamma_0}{p+M_0/A},\\
 \gamma_0=\frac{a_1a_2(\lambda_2-\lambda_1)^2}{A^2M_0}<0.
 \end{gathered}
\end{equation}
Thus the complementary density is strictly increasing.  If $M_0=0$,
it is $(\lambda_1+\lambda_2)/A+(\lambda_1\lambda_2/A)t$, again strictly
increasing.  Complete positivity fails in both cases.  These cases exhaust the parameter
space under \eqref{eq:two-mode-identity-gate}; failure of the gate was already
detected by the identity readout.
\end{proof}

For evaluation of the second branch, note first that $M_\infty>0$ there:
$\ell(0+)=M_\infty/b^2$, and a nonnegative nonincreasing density with
$\ell(0+)=0$ vanishes identically, which is impossible for a nonconstant
transfer.  Let
\begin{equation}\label{eq:two-mode-nu}
 \nu=\frac{A\lambda_1\lambda_2}{M_\infty}
\end{equation}
be the zero of the numerator in \eqref{eq:two-mode-inverse-density}.  If
$P$ has distinct roots $-\mu_1,-\mu_2$ with $0<\mu_1<\mu_2$, then
\begin{equation}\label{eq:two-mode-residues}
 \ell(t)=c_1e^{-\mu_1t}+c_2e^{-\mu_2t},\qquad
 c_1=\frac{M_\infty(\nu-\mu_1)}{b^2(\mu_2-\mu_1)},\qquad
 c_2=\frac{M_\infty(\mu_2-\nu)}{b^2(\mu_2-\mu_1)},
\end{equation}
so that $\ell(0+)=M_\infty/b^2$ and
$-\ell'(0+)=c_1\mu_1+c_2\mu_2=M_\infty(\mu_1+\mu_2-\nu)/b^2$.  Since
$e^{\mu_2t}(-\ell'(t))=c_1\mu_1e^{(\mu_2-\mu_1)t}+c_2\mu_2$ is monotone in
$t$, the density is nonincreasing exactly when $c_1\ge0$ and
$-\ell'(0+)\ge0$, and it is then nonnegative because it tends to zero.
The test is therefore
\begin{equation}\label{eq:two-mode-nu-test}
 \mu_1\le\nu\le\mu_1+\mu_2:
\end{equation}
the zero of the inverse lies between its slower pole and the sum of its
poles.  A nonzero complex-pole density oscillates and fails.  At a repeated
root $\mu$, $\ell(t)=(M_\infty/b^2)(1+(\nu-\mu)t)e^{-\mu t}$ and the test
is $\mu\le\nu\le2\mu$, which is \eqref{eq:two-mode-nu-test} again.

\begin{figure}[t]
\centering
\begin{minipage}[b]{0.56\textwidth}\centering
\begin{tikzpicture}[x=0.92cm,y=0.80cm,>=Latex]
  \begin{scope}
  \clip (-1,-1.4) rectangle (4,2.9);
  \fill[blue!13] (0,3.2) -- (0.0000,-0.0000) -- (0.0500,-0.0123) -- (0.1000,-0.0244) -- (0.1500,-0.0361) -- (0.2000,-0.0476) -- (0.2500,-0.0588) -- (0.3000,-0.0698) -- (0.3500,-0.0805) -- (0.4000,-0.0909) -- (0.4500,-0.1011) -- (0.5000,-0.1111) -- (0.5500,-0.1209) -- (0.6000,-0.1304) -- (0.6500,-0.1398) -- (0.7000,-0.1489) -- (0.7500,-0.1579) -- (0.8000,-0.1667) -- (0.8500,-0.1753) -- (0.9000,-0.1837) -- (0.9500,-0.1919) -- (1.0000,-0.2000) -- (1.0500,-0.2079) -- (1.1000,-0.2157) -- (1.1500,-0.2233) -- (1.2000,-0.2308) -- (1.2500,-0.2381) -- (1.3000,-0.2453) -- (1.3500,-0.2523) -- (1.4000,-0.2593) -- (1.4500,-0.2661) -- (1.5000,-0.2727) -- (1.5500,-0.2793) -- (1.6000,-0.2857) -- (1.6500,-0.2920) -- (1.7000,-0.2982) -- (1.7500,-0.3043) -- (1.8000,-0.3103) -- (1.8500,-0.3162) -- (1.9000,-0.3220) -- (1.9500,-0.3277) -- (2.0000,-0.3333) -- (2.0500,-0.3388) -- (2.1000,-0.3443) -- (2.1500,-0.3496) -- (2.2000,-0.3548) -- (2.2500,-0.3600) -- (2.3000,-0.3651) -- (2.3500,-0.3701) -- (2.4000,-0.3750) -- (2.4500,-0.3798) -- (2.5000,-0.3846) -- (2.5500,-0.3893) -- (2.6000,-0.3939) -- (2.6500,-0.3985) -- (2.7000,-0.4030) -- (2.7500,-0.4074) -- (2.8000,-0.4118) -- (2.8500,-0.4161) -- (2.9000,-0.4203) -- (2.9500,-0.4245) -- (3.0000,-0.4286) -- (3.0500,-0.4326) -- (3.1000,-0.4366) -- (3.1500,-0.4406) -- (3.2000,-0.4444) -- (3.2500,-0.4483) -- (3.3000,-0.4521) -- (3.3500,-0.4558) -- (3.4000,-0.4595) -- (3.4500,-0.4631) -- (3.5000,-0.4667) -- (3.5500,-0.4702) -- (3.6000,-0.4737) -- (3.6500,-0.4771) -- (3.7000,-0.4805) -- (3.7500,-0.4839) -- (3.8000,-0.4872) -- (3.8500,-0.4904) -- (3.9000,-0.4937) -- (3.9500,-0.4969) -- (4.0000,-0.5000) -- (4,3.2) -- cycle;
  \fill[blue!13] (-0.3112,3.2) -- (-0.3112,3.2002) -- (-0.2982,3.1411) -- (-0.2853,3.0812) -- (-0.2723,3.0205) -- (-0.2593,2.9590) -- (-0.2464,2.8966) -- (-0.2334,2.8333) -- (-0.2204,2.7688) -- (-0.2075,2.7032) -- (-0.1945,2.6364) -- (-0.1815,2.5682) -- (-0.1686,2.4984) -- (-0.1556,2.4269) -- (-0.1426,2.3535) -- (-0.1297,2.2778) -- (-0.1167,2.1996) -- (-0.1037,2.1184) -- (-0.0908,2.0337) -- (-0.0778,1.9445) -- (-0.0648,1.8499) -- (-0.0519,1.7479) -- (-0.0389,1.6357) -- (-0.0259,1.5074) -- (-0.0130,1.3480) -- (0.0000,1.0000) -- (0,3.2) -- cycle;
  \fill[blue!24] (0,0) rectangle (4,3.2);
  \draw[thick,blue!60!black] (0.0000,-0.0000) -- (0.0500,-0.0123) -- (0.1000,-0.0244) -- (0.1500,-0.0361) -- (0.2000,-0.0476) -- (0.2500,-0.0588) -- (0.3000,-0.0698) -- (0.3500,-0.0805) -- (0.4000,-0.0909) -- (0.4500,-0.1011) -- (0.5000,-0.1111) -- (0.5500,-0.1209) -- (0.6000,-0.1304) -- (0.6500,-0.1398) -- (0.7000,-0.1489) -- (0.7500,-0.1579) -- (0.8000,-0.1667) -- (0.8500,-0.1753) -- (0.9000,-0.1837) -- (0.9500,-0.1919) -- (1.0000,-0.2000) -- (1.0500,-0.2079) -- (1.1000,-0.2157) -- (1.1500,-0.2233) -- (1.2000,-0.2308) -- (1.2500,-0.2381) -- (1.3000,-0.2453) -- (1.3500,-0.2523) -- (1.4000,-0.2593) -- (1.4500,-0.2661) -- (1.5000,-0.2727) -- (1.5500,-0.2793) -- (1.6000,-0.2857) -- (1.6500,-0.2920) -- (1.7000,-0.2982) -- (1.7500,-0.3043) -- (1.8000,-0.3103) -- (1.8500,-0.3162) -- (1.9000,-0.3220) -- (1.9500,-0.3277) -- (2.0000,-0.3333) -- (2.0500,-0.3388) -- (2.1000,-0.3443) -- (2.1500,-0.3496) -- (2.2000,-0.3548) -- (2.2500,-0.3600) -- (2.3000,-0.3651) -- (2.3500,-0.3701) -- (2.4000,-0.3750) -- (2.4500,-0.3798) -- (2.5000,-0.3846) -- (2.5500,-0.3893) -- (2.6000,-0.3939) -- (2.6500,-0.3985) -- (2.7000,-0.4030) -- (2.7500,-0.4074) -- (2.8000,-0.4118) -- (2.8500,-0.4161) -- (2.9000,-0.4203) -- (2.9500,-0.4245) -- (3.0000,-0.4286) -- (3.0500,-0.4326) -- (3.1000,-0.4366) -- (3.1500,-0.4406) -- (3.2000,-0.4444) -- (3.2500,-0.4483) -- (3.3000,-0.4521) -- (3.3500,-0.4558) -- (3.4000,-0.4595) -- (3.4500,-0.4631) -- (3.5000,-0.4667) -- (3.5500,-0.4702) -- (3.6000,-0.4737) -- (3.6500,-0.4771) -- (3.7000,-0.4805) -- (3.7500,-0.4839) -- (3.8000,-0.4872) -- (3.8500,-0.4904) -- (3.9000,-0.4937) -- (3.9500,-0.4969) -- (4.0000,-0.5000);
  \draw[thick,blue!60!black] (-0.3112,3.2002) -- (-0.2982,3.1411) -- (-0.2853,3.0812) -- (-0.2723,3.0205) -- (-0.2593,2.9590) -- (-0.2464,2.8966) -- (-0.2334,2.8333) -- (-0.2204,2.7688) -- (-0.2075,2.7032) -- (-0.1945,2.6364) -- (-0.1815,2.5682) -- (-0.1686,2.4984) -- (-0.1556,2.4269) -- (-0.1426,2.3535) -- (-0.1297,2.2778) -- (-0.1167,2.1996) -- (-0.1037,2.1184) -- (-0.0908,2.0337) -- (-0.0778,1.9445) -- (-0.0648,1.8499) -- (-0.0519,1.7479) -- (-0.0389,1.6357) -- (-0.0259,1.5074) -- (-0.0130,1.3480) -- (0.0000,1.0000);
  \draw[very thick,blue!60!black] (0,0)--(0,3.2);
  \end{scope}
  \draw[dotted,gray] (0,-1)--(4,-1);
  \draw[->] (-1.1,0)--(4.25,0) node[right] {$a_1$};
  \draw[->] (0,-1.45)--(0,3.1) node[above] {$a_2$};
  \foreach \x in {-1,1,2,3,4} \draw (\x,0.05)--(\x,-0.05) node[below,font=\scriptsize] {$\x$};
  \foreach \y in {-1,1,2} \draw (0.05,\y)--(-0.05,\y) node[left,font=\scriptsize] {$\y$};
  \fill[red!70!black] (1,-0.125) circle (1.7pt);
  \draw[red!70!black] (1.05,-0.16) -- (1.55,-0.67) node[anchor=west,font=\scriptsize,inner sep=1pt] {$(1,-\tfrac18)$};
  \node[blue!60!black,anchor=west,font=\scriptsize] at (4.05,-0.5) {$\ell'(0+)=0$};
  \draw[->,gray!70!black] (-0.62,2.35) -- (-0.17,2.55);
  \node[gray!30!black,anchor=north east,font=\scriptsize,align=right] at (-0.55,2.4) {poles\\coincide};
  \node[font=\small] at (2.2,1.8) {safe};
  \node[font=\small] at (3.2,-0.82) {unsafe};
  \node[font=\small] at (-0.62,0.42) {unsafe};
  \node[anchor=south,font=\small] at (1.5,3.15) {(a) $g=1$};
\end{tikzpicture}
\end{minipage}\hfill
\begin{minipage}[b]{0.42\textwidth}\centering
\begin{tikzpicture}[x=6.4cm,y=27cm,>=Latex]
  \begin{scope}
  \clip (0,-0.09) rectangle (0.65,0.06);
  \fill[blue!13] (0,0.06) -- (0.0000,0.00000) -- (0.0050,-0.00125) -- (0.0100,-0.00251) -- (0.0150,-0.00376) -- (0.0200,-0.00503) -- (0.0250,-0.00629) -- (0.0300,-0.00756) -- (0.0350,-0.00883) -- (0.0400,-0.01010) -- (0.0450,-0.01138) -- (0.0500,-0.01266) -- (0.0550,-0.01394) -- (0.0600,-0.01523) -- (0.0650,-0.01652) -- (0.0700,-0.01781) -- (0.0750,-0.01911) -- (0.0800,-0.02041) -- (0.0850,-0.02171) -- (0.0900,-0.02302) -- (0.0950,-0.02433) -- (0.1000,-0.02564) -- (0.1050,-0.02696) -- (0.1100,-0.02828) -- (0.1150,-0.02960) -- (0.1200,-0.03093) -- (0.1250,-0.03226) -- (0.1300,-0.03359) -- (0.1350,-0.03493) -- (0.1400,-0.03627) -- (0.1450,-0.03761) -- (0.1500,-0.03896) -- (0.1550,-0.04031) -- (0.1600,-0.04167) -- (0.1650,-0.04302) -- (0.1700,-0.04439) -- (0.1750,-0.04575) -- (0.1800,-0.04712) -- (0.1850,-0.04849) -- (0.1900,-0.04987) -- (0.1950,-0.05125) -- (0.2000,-0.05263) -- (0.2050,-0.05402) -- (0.2100,-0.05541) -- (0.2150,-0.05680) -- (0.2200,-0.05820) -- (0.2250,-0.05960) -- (0.2300,-0.06101) -- (0.2350,-0.06242) -- (0.2400,-0.06383) -- (0.2450,-0.06525) -- (0.2500,-0.06667) -- (0.2550,-0.06809) -- (0.2600,-0.06952) -- (0.2650,-0.07095) -- (0.2700,-0.07031) -- (0.2750,-0.06676) -- (0.2800,-0.06334) -- (0.2850,-0.06003) -- (0.2900,-0.05685) -- (0.2950,-0.05377) -- (0.3000,-0.05081) -- (0.3050,-0.04795) -- (0.3100,-0.04520) -- (0.3150,-0.04255) -- (0.3200,-0.04000) -- (0.3250,-0.03755) -- (0.3300,-0.03519) -- (0.3350,-0.03293) -- (0.3400,-0.03076) -- (0.3450,-0.02868) -- (0.3500,-0.02668) -- (0.3550,-0.02477) -- (0.3600,-0.02294) -- (0.3650,-0.02120) -- (0.3700,-0.01953) -- (0.3750,-0.01795) -- (0.3800,-0.01644) -- (0.3850,-0.01501) -- (0.3900,-0.01365) -- (0.3950,-0.01236) -- (0.4000,-0.01115) -- (0.4050,-0.01000) -- (0.4100,-0.00892) -- (0.4150,-0.00791) -- (0.4200,-0.00697) -- (0.4250,-0.00609) -- (0.4300,-0.00528) -- (0.4350,-0.00452) -- (0.4400,-0.00383) -- (0.4450,-0.00320) -- (0.4500,-0.00263) -- (0.4550,-0.00212) -- (0.4600,-0.00167) -- (0.4650,-0.00127) -- (0.4700,-0.00093) -- (0.4750,-0.00064) -- (0.4800,-0.00041) -- (0.4850,-0.00023) -- (0.4900,-0.00010) -- (0.4950,-0.00003) -- (0.5000,-0.00000) -- (0.5,0) -- (0.65,0) -- (0.65,0.06) -- cycle;
  \fill[blue!24] (0,0) rectangle (0.65,0.06);
  \end{scope}
  \draw[thick,blue!60!black] (0.0000,0.00000) -- (0.0050,-0.00125) -- (0.0100,-0.00251) -- (0.0150,-0.00376) -- (0.0200,-0.00503) -- (0.0250,-0.00629) -- (0.0300,-0.00756) -- (0.0350,-0.00883) -- (0.0400,-0.01010) -- (0.0450,-0.01138) -- (0.0500,-0.01266) -- (0.0550,-0.01394) -- (0.0600,-0.01523) -- (0.0650,-0.01652) -- (0.0700,-0.01781) -- (0.0750,-0.01911) -- (0.0800,-0.02041) -- (0.0850,-0.02171) -- (0.0900,-0.02302) -- (0.0950,-0.02433) -- (0.1000,-0.02564) -- (0.1050,-0.02696) -- (0.1100,-0.02828) -- (0.1150,-0.02960) -- (0.1200,-0.03093) -- (0.1250,-0.03226) -- (0.1300,-0.03359) -- (0.1350,-0.03493) -- (0.1400,-0.03627) -- (0.1450,-0.03761) -- (0.1500,-0.03896) -- (0.1550,-0.04031) -- (0.1600,-0.04167) -- (0.1650,-0.04302) -- (0.1700,-0.04439) -- (0.1750,-0.04575) -- (0.1800,-0.04712) -- (0.1850,-0.04849) -- (0.1900,-0.04987) -- (0.1950,-0.05125) -- (0.2000,-0.05263) -- (0.2050,-0.05402) -- (0.2100,-0.05541) -- (0.2150,-0.05680) -- (0.2200,-0.05820) -- (0.2250,-0.05960) -- (0.2300,-0.06101) -- (0.2350,-0.06242) -- (0.2400,-0.06383) -- (0.2450,-0.06525) -- (0.2500,-0.06667) -- (0.2550,-0.06809) -- (0.2600,-0.06952) -- (0.2650,-0.07095) -- (0.2700,-0.07031) -- (0.2750,-0.06676) -- (0.2800,-0.06334) -- (0.2850,-0.06003) -- (0.2900,-0.05685) -- (0.2950,-0.05377) -- (0.3000,-0.05081) -- (0.3050,-0.04795) -- (0.3100,-0.04520) -- (0.3150,-0.04255) -- (0.3200,-0.04000) -- (0.3250,-0.03755) -- (0.3300,-0.03519) -- (0.3350,-0.03293) -- (0.3400,-0.03076) -- (0.3450,-0.02868) -- (0.3500,-0.02668) -- (0.3550,-0.02477) -- (0.3600,-0.02294) -- (0.3650,-0.02120) -- (0.3700,-0.01953) -- (0.3750,-0.01795) -- (0.3800,-0.01644) -- (0.3850,-0.01501) -- (0.3900,-0.01365) -- (0.3950,-0.01236) -- (0.4000,-0.01115) -- (0.4050,-0.01000) -- (0.4100,-0.00892) -- (0.4150,-0.00791) -- (0.4200,-0.00697) -- (0.4250,-0.00609) -- (0.4300,-0.00528) -- (0.4350,-0.00452) -- (0.4400,-0.00383) -- (0.4450,-0.00320) -- (0.4500,-0.00263) -- (0.4550,-0.00212) -- (0.4600,-0.00167) -- (0.4650,-0.00127) -- (0.4700,-0.00093) -- (0.4750,-0.00064) -- (0.4800,-0.00041) -- (0.4850,-0.00023) -- (0.4900,-0.00010) -- (0.4950,-0.00003) -- (0.5000,-0.00000);
  \draw[very thick,blue!60!black] (0,0)--(0,0.06);
  \draw[->] (0,0)--(0.68,0) node[right] {$a_1$};
  \draw[->] (0,-0.093)--(0,0.07) node[above] {$a_2$};
  \foreach \x in {0.25,0.5} \draw (\x,0.002)--(\x,-0.002) node[below,font=\scriptsize] {$\x$};
  \foreach \y/\lab in {-0.05/{-0.05},0.05/{0.05}} \draw (0.004,\y)--(-0.004,\y) node[left,font=\scriptsize] {$\lab$};
  \node[font=\small] at (0.33,0.032) {safe};
  \node[font=\small] at (0.47,-0.06) {unsafe};
  \node[blue!60!black,rotate=-48,anchor=north west,font=\scriptsize,inner sep=1pt] at (0.04,-0.0155) {$\ell'(0+)=0$};
  \node[gray!30!black,anchor=west,font=\scriptsize,align=left] at (0.40,-0.028) {poles\\coincide};
  \draw[->,gray!70!black] (0.40,-0.028) -- (0.385,-0.01501);
  \node[anchor=south,font=\small] at (0.33,0.072) {(b) $g=-1$};
\end{tikzpicture}
\end{minipage}
\caption{Safe set of $G(t)=g+a_1e^{-t}+a_2e^{-2t}$ for every readout in
$\readouts$ and every round trip (\Cref{thm:two-mode}).  Dark shading is
the coefficientwise positive quadrant, light shading the signed safe set.
(a) $g=1$.  For $a_1>0$ the lower boundary is the hyperbola
$a_2=-a_1/(a_1+4)$ of \eqref{eq:two-mode-hyperbola}, on which
$\ell'(0+)=0$; the dotted line is its asymptote $a_2=-1$, and the dot is
\Cref{ex:safe-signed-two-mode}.  For $a_1<0$ the boundary, computed from
\eqref{eq:two-mode-nu-test}, is the curve $a_2=(1+\sqrt{2|a_1|})^2$ on
which the two inverse poles coincide; below it they are complex and the
density oscillates.  (b) $g=-1$, magnified.  The signed safe set is the
thin wedge $0<a_1<1/2$,
$a_2\ge\max\{a_1/(a_1-4),-(1-\sqrt{2a_1})^2\}$, bounded by the curve
$\ell'(0+)=0$ for $a_1<2-\sqrt3$ and by the coincident-pole curve for
$a_1>2-\sqrt3$; the two curves meet at
$a_1=2-\sqrt3\approx0.267949$.  The wedge requires $b<0$ and closes at
$a_1=1/2$.}
\label{fig:two-mode}
\end{figure}
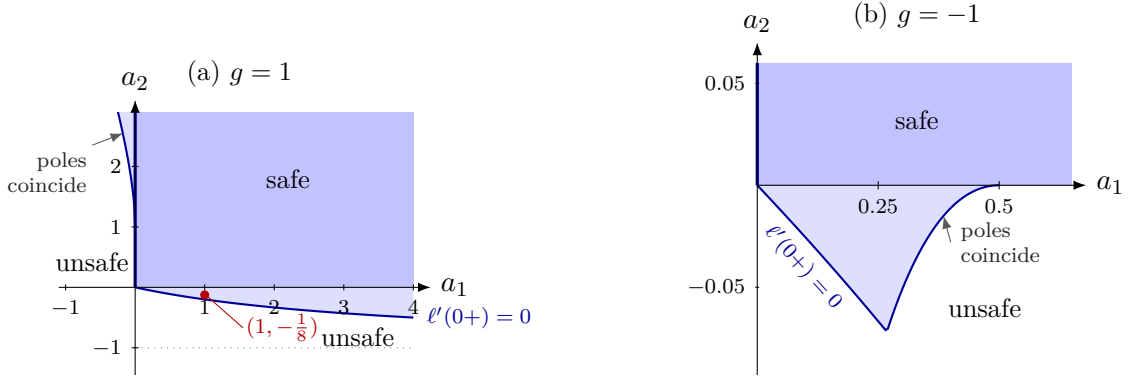

\begin{example}[A safe signed two-mode kernel]\label{ex:safe-signed-two-mode}
Let
\begin{equation}\label{eq:safe-signed-kernel}
 G(t)=1+e^{-t}-\tfrac18e^{-2t}.
\end{equation}
Here $g=1$, $b=15/8$, $A=7/8$, $M_0=15/8$, and $M_\infty=3/4$: the
permanent and instantaneous impacts are positive, the identity gate holds,
and the fast mode is negative.  The inverse is
\begin{equation}\label{eq:safe-signed-inverse}
 \frac1{p\laplace G(p)}=\frac{8(p+1)(p+2)}{15p^2+39p+16}
 =\frac8{15}+\frac{16(3p+7)}{15(15p^2+39p+16)},
\end{equation}
with poles $-\mu_{1,2}=-(39\mp\sqrt{561})/30\approx-0.5105,\,-2.0895$ and
zero $-\nu=-7/3$.  Since $\mu_1\le7/3\le\mu_1+\mu_2=13/5$, the kernel is
universally round-trip safe, and every input of volume $m$ has cost at least
$F_h(m)$.  Explicitly,
$c_1=8/75+248\sqrt{561}/42075\approx0.2463$,
$c_2=8/75-248\sqrt{561}/42075\approx-0.0329$, $\ell(0+)=16/75$, and
$-\ell'(0+)=64/1125$.  The kernel is positive, decreasing, and convex, but
not completely monotone.

The same computation describes the whole family $g=1$,
$(\lambda_1,\lambda_2)=(1,2)$, $a_1>0$: it is safe if and only if
\begin{equation}\label{eq:two-mode-hyperbola}
 a_2\ge-\frac{a_1}{a_1+4}.
\end{equation}
Indeed, for $a_1>0>a_2$ the discriminant of $P$ is
$a_2^2+(4a_1-2)a_2+(2a_1+1)^2>0$.  If
$a_2<-a_1/2$, then $M_\infty<0$ and the identity gate fails.  At
$a_2=-a_1/2$, a nonconstant nonnegative nonincreasing inverse density cannot
start from $\ell(0+)=M_\infty/b^2=0$.  It remains to consider the connected
strip $-a_1/2<a_2<0$, where $b$, $M_0$, and $M_\infty$ are positive and
the two inverse poles are stable and distinct.  At $a_2=0$ the unreduced
denominator roots are $-2$ and $-1/(1+a_1)$, but the factor $p+2$ cancels
from the inverse; $\nu=2$ meets this removable root.  As $a_2\uparrow0$
from the signed strip, the faster pole tends to $\mu_2=2=\nu$.  A direct
expansion gives
$M_\infty^2P(-\nu)=-2a_1a_2b>0$.  For $a_2<0$ sufficiently close to zero,
continuity and the separation from $\mu_1$ therefore place $\nu$ above
$\mu_2$.  The same strict identity prevents a crossing of either pole, so
$\nu>\mu_2$ throughout the strip.  The test
\eqref{eq:two-mode-nu-test} therefore reduces to
$\nu\le\mu_1+\mu_2$, which is \eqref{eq:two-mode-hyperbola}.  On that set
the identity gate holds automatically, while below it either the identity
gate fails or $-\ell'(0+)<0$, so \Cref{thm:two-mode} gives manipulation.
At $a_1=1$ the boundary is $a_2=-1/5$.  \Cref{fig:two-mode}
draws the safe set; coefficientwise positivity is sufficient but not
necessary at two modes, and the safe set is neither convex nor a cone.
\end{example}

\begin{proposition}[Linear passivity is strictly weaker]\label{prop:linear-not-nonlinear}
The transient $H_*(t)=2e^{-t}-e^{-2t}$ has strictly positive
identity-readout spectral density
\begin{equation}\label{eq:linear-positive-spectrum}
 \re\laplace H_*(i\omega)=
 \frac6{(1+\omega^2)(4+\omega^2)},
\end{equation}
but $g+H_*$ fails universal nonlinear round-trip safety for every
$g\in\R$.
\end{proposition}

The proof uses stable-inverse necessity for $g<-1$, an explicit five-pulse
family for $g>-1$, and a separate singular pulse train for $g=-1$; the
formulas and finite-step closure appear in \Cref{app:prony-separators}.

\subsection{Three modes and the cut hierarchy}

\begin{proposition}[Minimal signed safe three-mode tail]
\label{prop:signed-three-mode}
Let
\begin{equation}\label{eq:three-mode-complement}
 \laplace L(p)=1+\frac2p+\frac1{p+1}-\frac1{10(p+2)},\qquad
 \ell(t)=2+e^{-t}-\frac1{10}e^{-2t}.
\end{equation}
Then $\ell>0$ and $\ell'<0$, and the reciprocal tail
\begin{equation}\label{eq:three-mode-tail}
 \laplace H(p)=
 \frac{10(p+1)(p+2)}{10p^3+59p^2+99p+40}
\end{equation}
is completely positive.  It is a three-mode Prony kernel with residue signs
$+,-,+$.  No signed safe completely positive tail exists with only two
distinct modes, so three is minimal.
\end{proposition}

\begin{proof}
For $x=e^{-t}\in(0,1]$,
$\ell=2+x-x^2/10>0$ and $\ell'=-x(1-x/5)<0$, so the complementary-kernel
theorem makes $H$ completely positive.  Set
\begin{equation}\label{eq:three-mode-polynomial}
 F(\lambda)=10\lambda^3-59\lambda^2+99\lambda-40.
\end{equation}
Sign changes at
$1/2,3/5,2,11/5,3,13/4$, together with discriminant $179881>0$, isolate
three distinct positive rates
\begin{equation}\label{eq:three-mode-rates}
 \lambda_1\in(1/2,3/5),\quad
 \lambda_2\in(2,11/5),\quad
 \lambda_3\in(3,13/4).
\end{equation}
The residue at $p=-\lambda_j$ is
\begin{equation}\label{eq:three-mode-residue}
 A_j=\frac{10(1-\lambda_j)(2-\lambda_j)}
 {30\lambda_j^2-118\lambda_j+99}.
\end{equation}
The numerator is positive at all three roots and the denominator alternates
$+,-,+$ for an ordered cubic, proving the sign pattern.  Minimality follows
from the $g=0$ branch of \Cref{thm:two-mode}.
\end{proof}

For the general hierarchy, let
\begin{equation}\label{eq:finite-prony}
 G(t)=g+\sum_{j=1}^na_je^{-\lambda_jt},\qquad
 0<\lambda_1<\cdots<\lambda_n,\qquad b=G(0+)\ne0.
\end{equation}
Write
\begin{equation}\label{eq:reduced-transfer}
 R(p)=p\laplace G(p)=\frac{\widetilde P(p)}{\widetilde D_0(p)}
\end{equation}
after canceling the polynomial greatest common divisor.  Actual inverse
poles are zeros of $\widetilde P$; removable roots impose no constraint.
Assume no actual inverse pole lies on the imaginary axis.  In particular,
$g=R(0)\ne0$.  Put $\beta=1/b$ and write
$1/R=\beta+\laplace\ell$, where $\ell$ is a real exponential polynomial.

Choose the following independent real moment basis
$\phi_1,\ldots,\phi_r$.  A real unstable pole $\zeta>0$ of multiplicity
$m$ contributes $s^ke^{-\zeta s}$ for $0\le k<m$.  From each nonreal
conjugate pair with positive real part, choose just one pole $\zeta$ and
include the real and imaginary parts of the corresponding functions.  The
resulting real exponential-monomial functions are linearly independent on
every open interval: no selected imaginary part is the zero function, and
including both poles of a conjugate pair would add only redundant functions.
For a finite partition $\Pi:0=t_0<\cdots<t_N=T$, let
\begin{equation}\label{eq:generalized-moments}
 (W_\Pi)_{ki}=\int_{t_i}^{t_{i+1}}\phi_k(s)\dd s,
 \qquad 1\le k\le r,\quad 0\le i<N.
\end{equation}
When there are no unstable poles, $W_\Pi$ has no rows and its kernel is
all of $\R^N$.  The compact cell state associated with $x\in\R^N$ is
$D_x=x_i$ on $[t_i,t_{i+1})$ and zero after $T$.
Define the lower-triangular inverse-increment matrix
\begin{equation}\label{eq:inverse-increment-matrix}
 (A_\Pi)_{ij}=\begin{cases}
 0,&j>i,\\
 \displaystyle\int_{t_i}^{t_{i+1}}\ell(t_{i+1}-s)\dd s,&j=i,\\
 \displaystyle\int_{t_j}^{t_{j+1}}
 [\ell(t_{i+1}-s)-\ell(t_i-s)]\dd s,&j<i.
 \end{cases}
\end{equation}
Its $i$th component is the increment of $\ell*D_x$ on the $i$th cell.
The feedthrough contribution telescopes when the state jumps are smoothed
into a compact loop, leaving the limiting cost
\begin{equation}\label{eq:finite-grid-cost}
 \mathcal C_\Pi(x,h)=h(x)^\transpose A_\Pi x,
 \qquad x\in\ker W_\Pi.
\end{equation}

For $S\subseteq\{0,\ldots,N-1\}$, set
$c_S=A_\Pi^\transpose\one_S$.  The positive-cut program seeks $x,z$ with
\begin{equation}\label{eq:positive-cut-lp}
 W_\Pi x=0,\quad z\ge0,\quad
 x_i-z\ge1\ (i\in S),\quad z-x_i\ge1\ (i\notin S),\quad
 c_S^\transpose x\le-1,
\end{equation}
and the negative-cut program seeks
\begin{equation}\label{eq:negative-cut-lp}
 W_\Pi x=0,\quad z\le0,\quad
 z-x_i\ge1\ (i\in S),\quad x_i-z\ge1\ (i\notin S),\quad
 c_S^\transpose x\ge1.
\end{equation}

\begin{theorem}[Finite-Prony cut hierarchy]\label{thm:finite-prony-hierarchy}
Under the nonzero-gain and off-axis hypotheses above, the following are
equivalent:
\begin{enumerate}[label=\textup{(\roman*)}]
\item every finite piecewise-constant round trip is safe for every
  $h\in\readouts$;
\item for every finite partition $\Pi$ and every subset $S$, both cut
  programs \eqref{eq:positive-cut-lp}--\eqref{eq:negative-cut-lp} are
  infeasible.
\end{enumerate}
If safety fails, one rational partition and one finite cut program detect it
with a strict margin.  Every feasible cut closes to a finite bounded
round-trip manipulation inside every positive cap.  The theorem includes
repeated inverse poles through \eqref{eq:generalized-moments}; it makes no
claim about imaginary-axis poles or the arbitrary-$n$ surface $b=0$.
\end{theorem}

\begin{proof}
At the finitely many state values, the vectors $(h(x_i))_i$ form the
isotone cone pinned by $h(0)=0$.  Its generators are positive upper cuts
above nonnegative thresholds and negative lower cuts below nonpositive
thresholds; this is the finite isotone cone of \citet{Ubhaya2001}.  Hence \eqref{eq:finite-grid-cost} is negative for some
$h\in\readouts$ precisely when one cut has the wrong sign.  The threshold
can be chosen strictly between the adjacent distinct state values and on
the prescribed side of zero.  Scaling $x,z$ by a common positive factor
normalizes the separation and objective margins to one, giving
\eqref{eq:positive-cut-lp} or \eqref{eq:negative-cut-lp}.  Conversely a
continuous monotone ramp realizes a feasible cut at all the state values.

\Cref{prop:prony-cell-realization} proves both passages between a
negative cell form and an ordinary finite round trip, including exact
unstable-moment correction, convergence of the full nonlinear cost,
stable-tail compensation, and cap scaling.  Its support-independent
estimate \eqref{eq:prony-global-memory-bound} is essential in the
converse, where the state support grows.  This proves the equivalence.

A feasible cut can first be refined to full moment rank without changing
its state or objective.  After a small rational endpoint perturbation,
projection onto the perturbed moment kernel preserves its strict margins;
\Cref{app:prony-rational-refinement} supplies the construction.
Only partition endpoints, not projected state coordinates, are asserted
to be rational.
\end{proof}

\begin{corollary}[Effective detection of unsafety]
\label{cor:prony-unsafety-semialgorithm}
For computably presented coefficients and rates in
\eqref{eq:finite-prony}, unsafety is semidecidable: a procedure halts with a
finite round-trip witness and a certified strictly negative cost
whenever universal safety fails.  No termination on safe inputs is
asserted.
\end{corollary}

\begin{proof}
Enumerate rational finite-step round trips and normalized monotone
piecewise-linear readouts with rational vertices.  This class contains a
strict negative witness whenever one exists, and its costs admit certified
rational enclosures; both facts are proved in
\Cref{app:prony-effective-detection}.  Dovetail the enclosure refinements
and halt when an upper endpoint is negative.  This avoids tests for exact
inverse-pole multiplicity or real-coefficient LP feasibility.  Without
computable parameter presentations, the hierarchy is a countable
mathematical separator, not an effective algorithm.
\end{proof}

\subsection{Why complete positivity of the tail is not enough}

The LICM condition in \Cref{thm:licm-shifts} cannot be weakened to arbitrary
complete positivity of the decaying tail.

\begin{proposition}[A completely positive tail whose negative shift fails]\label{prop:cp-tail-shift-fails}
Let $H$ be determined by
\begin{equation}\label{eq:cp-tail-counterexample-complement}
 L=\delta_0+\left(\frac15+5\,\one_{(0,1)}(t)\right)\dd t,\qquad L*H=1.
\end{equation}
Then $H$ is positive, exponentially decaying, and completely positive on
every finite horizon.  Nevertheless, for $g=-1/2$ there is an analytic
strictly increasing normalized readout and the four-block round trip
$v=(-1,-4,4,1)$ on equal half-unit cells such that
\begin{equation}\label{eq:cp-tail-counterexample-margin}
 \cost^{\mathrm{out}}_{g+H,h}[v]<-\frac{11}{1000}.
\end{equation}
The violation is cap-local.
\end{proposition}

The proof is an interval estimate recorded in
\Cref{app:prony-separators}.  The shifted inverse has a right-half-plane pole
and opposite endpoint masses $2$ and $-2=1/g$.  It therefore lies outside
both the LICM interlacing mechanism and the stable-complement class.

\section{Time-inhomogeneous first-order memory}
\label{sec:time-inhomogeneous}

Convolution and asymptotic assumptions are unnecessary for the first-order
nonstationary family, which is the time-varying resilience model studied by
\citet{FruthSchoenebornUrusov2014} and \citet{AlfonsiInfante2014} with a
general readout.  Let
\begin{equation}\label{eq:time-varying-model}
 b\in C^1([0,\infty)),\quad b(t)>0,\quad
 \rho\in C([0,\infty)),\qquad
 D'=-\rho(t)D+b(t)v,\quad D(0)=0.
\end{equation}
For $h\in\readouts$ set
\begin{equation}\label{eq:time-F-Phi}
 F(x)=\int_0^xh(z)\dd z,\qquad \Phi(x)=xh(x)-F(x).
\end{equation}
Both functions are nonnegative on $\R$.

\begin{theorem}[Time-inhomogeneous classification]
\label{thm:time-inhomogeneous}
The following are equivalent:
\begin{enumerate}[label=\textup{(\roman*)}]
\item for every finite horizon, every finite piecewise-constant round
  trip, and every $h\in\readouts$, the cost
  $\int_0^Tv(t)h(D(t))\dd t$ is nonnegative;
\item the same inequality holds for every finite piecewise-constant input;
\item for every $t\ge0$,
  \begin{equation}\label{eq:time-inhomogeneous-conditions}
   \rho(t)\ge0,\qquad b'(t)+\rho(t)b(t)\ge0.
  \end{equation}
\end{enumerate}
The equivalence is cap-local.  In the necessity direction,
$\readouts$ may be replaced by all globally real-analytic strictly
increasing normalized readouts.  No forgetting, bounded-gain, convolution,
or asymptotic-homogeneity assumption is used.  For a fixed linear or
power readout, \citet{FruthSchoenebornUrusov2014} and
\citet{AlfonsiInfante2014} obtain single combined coefficient conditions;
quantifying over $\readouts$ separates them into the two pointwise
inequalities.
\end{theorem}

\subsection{Two state-coordinate identities}

Define
\begin{equation}\label{eq:time-coefficients}
 c_1(t)=\frac{b'(t)+\rho(t)b(t)}{b(t)^2},\qquad
 c_2(t)=\frac{\rho(t)}{b(t)}.
\end{equation}

\begin{lemma}[Storage and volume identities]\label{lem:time-identities}
Every finite-step input satisfies
\begin{equation}\label{eq:time-storage-identity}
 \int_0^Tv(t)h(D(t))\dd t
 =\frac{F(D(T))}{b(T)}
 +\int_0^T\bigl[c_1(t)F(D(t))+c_2(t)\Phi(D(t))\bigr]\dd t.
\end{equation}
If $D\in C_c^1((0,T))$ and
\begin{equation}\label{eq:time-realizing-rate}
 v_D(t)=\frac{D'(t)+\rho(t)D(t)}{b(t)},
\end{equation}
then
\begin{equation}\label{eq:time-volume-identity}
 \int_0^Tv_D(t)\dd t=\int_0^Tc_1(t)D(t)\dd t.
\end{equation}
\end{lemma}

\begin{proof}
Since $v=(D'+\rho D)/b$,
\begin{align*}
 vh(D)
 &=\frac1b\frac{\dd}{\dd t}F(D)+\frac\rho bDh(D)\\
 &=\frac{\dd}{\dd t}\left(\frac{F(D)}b\right)
   +\frac{b'}{b^2}F(D)+\frac\rho b\bigl(F(D)+\Phi(D)\bigr).
\end{align*}
Integration proves \eqref{eq:time-storage-identity}.  For a compact interior
state path, integration by parts gives
\begin{align*}
 \int_0^T\frac{D'}b\dd t
 &=\left[\frac Db\right]_0^T+\int_0^T\frac{b'}{b^2}D\dd t,
\end{align*}
and adding $\rho D/b$ gives \eqref{eq:time-volume-identity}.
\end{proof}

If \eqref{eq:time-inhomogeneous-conditions} holds, then $c_1,c_2\ge0$ and
\eqref{eq:time-storage-identity} proves all-input safety.  The rest of the
section proves that the single volume constraint
\begin{equation}\label{eq:time-state-constraint}
 \int c_1D=0
\end{equation}
cannot hide a negative point of either coefficient.

\subsection{A negative convex-storage coefficient}

\begin{lemma}[Two-bump separator for $c_1<0$]
\label{lem:negative-c1}
If $c_1(t_0)<0$, there is a smooth compact state path whose realizing rate is
a round trip and has negative cost for a globally analytic strictly
increasing normalized readout.
\end{lemma}

\begin{proof}
Choose an interval $I\Subset(0,\infty)$ on which $c_1<0$ and disjoint
nonzero $\varphi_+,\varphi_-\in C_c^\infty(I)$ with
$\varphi_\pm\ge0$.  Put
\begin{equation}\label{eq:c1-bump-masses}
 A_\pm=\int_Ic_1(t)\varphi_\pm(t)\dd t<0.
\end{equation}
For $x>0$, set $y=xA_+/A_->0$ and
\begin{equation}\label{eq:c1-bump-state}
 D=x\varphi_+-y\varphi_-.
\end{equation}
Then \eqref{eq:time-volume-identity} gives exact zero volume.  Let
\begin{equation}\label{eq:F-isolating-readout}
 \sigma(r)=\frac1{1+e^{-r}},\qquad
 h_n^F(z)=\sigma(nz-\sqrt n)-\sigma(-\sqrt n).
\end{equation}
These readouts are globally analytic, strictly increasing, and normalized.
On the fixed state range their primitives and remainders satisfy, pointwise
and boundedly,
\begin{equation}\label{eq:F-isolating-limit}
 F_n^F(z)\to z_+,\qquad \Phi_n^F(z)\to0.
\end{equation}
The terminal state is zero.  Dominated convergence in
\eqref{eq:time-storage-identity} yields
\begin{equation}\label{eq:c1-negative-limit}
 \int v_Dh_n^F(D)\longrightarrow
 x\int_Ic_1(t)\varphi_+(t)\dd t=xA_+<0.
\end{equation}
A sufficiently large finite $n$ gives the required strict separator.
\end{proof}

\subsection{A negative Bregman-remainder coefficient}

\begin{lemma}[Plateau separator for $c_2<0$]
\label{lem:negative-c2}
If $c_2(t_0)<0$, there is a smooth compact state path whose realizing rate is
a round trip and has negative cost for a globally analytic strictly
increasing normalized readout, regardless of the local value or sign of
$c_1$.
\end{lemma}

\begin{proof}
Choose $I\Subset(0,\infty)$ on which $c_2\le-\eta<0$.  If $c_1\equiv0$
there, take $D=x\varphi_+$, where $0\le\varphi_+\le1$ is compactly supported
in $I$ and equals one on a positive-measure compact set $E$.  Then
\eqref{eq:time-state-constraint} holds automatically.

Otherwise shrink to $J\Subset I$ where $c_1$ has one strict sign.  Choose
disjoint nonnegative $\varphi_+,\varphi_-\in C_c^\infty(J)$, with
$\varphi_+=1$ on $E$, and put
\begin{equation}\label{eq:c2-balancing-state}
 A_\pm=\int_Jc_1\varphi_\pm,\qquad
 D=x\varphi_+-x\frac{A_+}{A_-}\varphi_-.
\end{equation}
Thus $\int c_1D=0$ exactly.  The two $A_\pm$ have the same nonzero sign, so the
balancing coefficient is positive.

Use the analytic normalized readout
\begin{equation}\label{eq:Phi-isolating-readout}
 h_{n,x}^\Phi(z)=e^{n(z-x)}-e^{-nx}.
\end{equation}
On the attained state range $z\le x$,
\begin{equation}\label{eq:Phi-isolating-limit}
 F_{n,x}^\Phi(z)\to0,\qquad
 \Phi_{n,x}^\Phi(z)\to
 \begin{cases}0,&z<x,\\x,&z=x.
 \end{cases}
\end{equation}
Dominated convergence gives
\begin{equation}\label{eq:c2-negative-limit}
 \int v_Dh_{n,x}^\Phi(D)\longrightarrow
 x\int_{\{D=x\}}c_2(t)\dd t
 \le x\int_Ec_2(t)\dd t<0.
\end{equation}
This proves the claim.
\end{proof}

\begin{proof}[Proof of \Cref{thm:time-inhomogeneous}]
Sufficiency was proved after \Cref{lem:time-identities}.  If universal
round-trip safety holds, \Cref{lem:negative-c1,lem:negative-c2} force
$c_1,c_2\ge0$, which is
\eqref{eq:time-inhomogeneous-conditions} because $b>0$.

The smooth realizing rates are bounded and have exact zero volume.  Replace
each by its cell averages on a refining partition.  Cell averaging preserves
the integral and the cap.  Variation of constants gives uniform state
convergence, and continuity of the fixed readout on the common compact range
gives cost convergence.  Hence a finite piecewise-constant round-trip
separator exists whenever either coefficient is negative.

For cap locality, scale $D,v$ by $\alpha$ and replace $h(z)$ by
$h(z/\alpha)$.  The cost scales by the positive factor $\alpha$ and remains
negative, while the rate enters any prescribed positive cap.  Analyticity
and strict increase are preserved.
\end{proof}

\begin{example}[Two diagnostics]\label{ex:time-diagnostics}
If $\rho=0$, $b(t)=(1+t)^{-1}$, and
$D(t)=\sin(2\pi t)$ on $[0,1]$, then $c_1=-1$, $c_2=0$, the state realizes
a round trip, and the identity-readout cost is $-1/4$.

If $\rho=-1$, $b=e^t$, and $D(t)=\sin(\pi t)$ on $[0,1]$, then
$c_1=0$, $c_2=-e^{-t}$ and
\begin{equation}\label{eq:negative-resilience-diagnostic}
 \cost=-\frac12\int_0^1e^{-t}\sin^2(\pi t)\dd t
 =-\frac{\pi^2(1-e^{-1})}{1+4\pi^2}<0.
\end{equation}
The latter dynamics need not forget at infinity; this is why the direct
volume-neutral proof is stronger than remote compensation.  In both
examples, smoothing the endpoint derivatives and averaging the resulting
bounded rate over a refining partition preserves the volume and the strict
negative sign, and the state/readout scaling of the theorem fits the
witness into every positive cap.
\end{example}

\section{Vector and matrix memory}\label{sec:matrix}

Let $v,D,h$ take values in $\R^d$ and use the power pairing
\begin{equation}\label{eq:matrix-cost}
 D_v(t)=\int_0^tG(t-s)v(s)\dd s,\qquad
 \cost_{G,h,T}[v]=\int_0^Tv(t)^\transpose h(D_v(t))\dd t.
\end{equation}
The vector round-trip constraint is $\int_0^Tv=0\in\R^d$.
No coordinatewise order on $h$ is assumed in the first-order theorem;
the readout is fixed and its storage geometry is part of the classification.

\subsection{Vector remote compensation}

Fix any norm on $\R^d$, its dual norm, and the induced matrix norm.

\begin{theorem}[Vector round-trip reduction]\label{thm:vector-remote}
Let
$G\in L^1_{\mathrm{loc}}(0,\infty;\R^{d\times d})$ satisfy
\begin{equation}\label{eq:matrix-tail-decay}
 \varepsilon(R)=\esssup_{t\ge R}\|G(t)\|\longrightarrow0.
\end{equation}
For any fixed continuous $h:\R^d\to\R^d$ with $h(0)=0$, safety on every
finite horizon for all finite piecewise-constant inputs is equivalent to
safety for all finite piecewise-constant vector round trips.  The equivalence
holds inside every fixed positive symmetric norm cap.
\end{theorem}

\begin{proof}
Let $v$ be supported on $[0,T]$, put
$m=\int_0^Tv$ and $K_v=\int_0^T\|v(t)\|\dd t$, and append after a remote gap
the constant vector block $w=-m/L$ on $[R,R+L]$.  On that block,
\begin{align}
 \left\|\int_0^TG(R+s-u)v(u)\dd u\right\|
 &\le K_v\varepsilon(R-T),
 \label{eq:matrix-old-state-bound}\\
 \left\|\int_0^sG(s-r)w\dd r\right\|
 &\le\|m\|a(L),\qquad
 a(L)=\frac1L\int_0^L\|G(r)\|\dd r\longrightarrow0.
 \label{eq:matrix-self-state-bound}
\end{align}
Thus the compensator state is uniformly bounded by a number
$\delta_{R,L}\to0$.  Duality gives the orientation-free estimate
\begin{equation}\label{eq:matrix-comp-cost-bound}
 |\cost_{\mathrm{comp}}|
 \le\|m\|\sup_{\|x\|\le\delta_{R,L}}\|h(x)\|_*\longrightarrow0.
\end{equation}
The enlarged control has exact zero vector volume, the old cost is unchanged
by causality, and the gap costs zero.  Round-trip safety therefore implies
the original all-input inequality.  If the cap is $c$, choose
$L\ge\|m\|/c$.
\end{proof}

Matrix rotation, nonnormality, and cross-time noncommutation create no extra
term in this argument.  Volume is closed before the kernel is applied, and
operator/dual norms control the resulting state and supply rate.

\subsection{Stable fully actuated first-order dynamics}

Consider
\begin{equation}\label{eq:matrix-first-order}
 G(t)=e^{-At}B,\qquad B\in \mathrm{GL}_d(\R),\qquad
 \|e^{-At}\|\le Me^{-\alpha t}\quad(t\ge0),
\end{equation}
or equivalently
\begin{equation}\label{eq:matrix-state-equation}
 D'=-AD+Bv,\qquad D(0)=0.
\end{equation}
For $h\in C^1(\R^d;\R^d)$ with $h(0)=0$, set
\begin{equation}\label{eq:matrix-one-form}
 \omega(x)=B^{-\transpose}h(x).
\end{equation}

The all-input conditions in part \textup{(iii)} below are the
zero-feedthrough, fully actuated passive specialization of
\citet[Theorem~17]{HillMoylan1980}, in the storage--supply framework of
\citet{Willems1972}; see also \citet{Moylan1974,HillMoylan1976}.
Classical cycle-dissipativity returns the internal state to its initial value
\citep[Definition~8]{HillMoylan1980}.  The vector round trips here instead
close input volume and generally leave the terminal memory state open.
What the theorem adds is that this weaker closure is nevertheless as strong
as all-input safety once every horizon is allowed.

\begin{theorem}[Stable fully actuated matrix classification]
\label{thm:stable-matrix}
The following are equivalent:
\begin{enumerate}[label=\textup{(\roman*)}]
\item every finite piecewise-constant vector round trip is safe on every
  finite horizon;
\item every finite piecewise-constant input is safe on every finite horizon;
\item $\omega$ is conservative and
  \begin{equation}\label{eq:matrix-drift-condition}
   \omega(x)^\transpose Ax\ge0\qquad(x\in\R^d).
  \end{equation}
\end{enumerate}
Equivalently,
\begin{equation}\label{eq:matrix-differential-conditions}
 B^{-\transpose}Dh(x)\text{ is symmetric},\qquad
 h(x)^\transpose B^{-1}Ax\ge0\quad(x\in\R^d).
\end{equation}
If $\nabla V=\omega$ and $V(0)=0$, then stability and
\eqref{eq:matrix-drift-condition} imply $V\ge0$ automatically, and every
input satisfies
\begin{equation}\label{eq:matrix-storage-identity}
 \cost_{G,h,T}[v]
 =V(D(T))+\int_0^T\omega(D(t))^\transpose AD(t)\dd t\ge0.
\end{equation}
\end{theorem}

\begin{proof}
By \Cref{thm:vector-remote}, round-trip and all-input safety are equivalent.
Since $B$ is invertible, every piecewise $C^1$ state path $x$ is generated by
\begin{equation}\label{eq:matrix-path-control}
 v=B^{-1}(x'+Ax).
\end{equation}
Compressed closed loops and their reversals force the circulation of
$\omega$ to vanish, hence $\omega=\nabla V$.  Holding an arbitrary state
$x$ for a time $L$ with the constant rate $B^{-1}Ax$ contributes
$L\omega(x)^\transpose Ax$, so arbitrary horizons force
\eqref{eq:matrix-drift-condition}.  These path controls are limits of finite
piecewise-constant controls; state and cost converge uniformly on compact
horizons.

Conversely,
\begin{align}
 v^\transpose h(D)
 &=D'^\transpose B^{-\transpose}h(D)
   +D^\transpose A^\transpose B^{-\transpose}h(D)
 \notag\\
 &=\frac{\dd}{\dd t}V(D)+\omega(D)^\transpose AD,
\end{align}
which proves \eqref{eq:matrix-storage-identity} once $V\ge0$.  To obtain that
sign without an additional hypothesis, let $y(t)=e^{-At}x$.  Stability gives
$y(t)\to0$, and
\begin{equation}\label{eq:stable-potential-formula}
 V(x)=\int_0^\infty\omega(e^{-At}x)^\transpose Ae^{-At}x\dd t\ge0.
\end{equation}
Finally, on simply connected $\R^d$, a $C^1$ vector field is conservative
exactly when its Jacobian is symmetric.
\end{proof}

\begin{remark}[A fixed cap does not identify an unreachable readout]
\label{rem:matrix-cap}
\Cref{thm:vector-remote} gives a round-trip/all-input equivalence
inside one cap.  Safety under that cap for one fixed $h$ need not force
\eqref{eq:matrix-differential-conditions} outside the cap-reachable state
set.  For $A=B=I$ and $\|v\|_2\le1$, every state satisfies $\|D(t)\|_2<1$;
a $C^1$ radial readout can equal $D$ on the unit ball and have negative
radial pairing beyond radius two without affecting any capped trajectory.
\end{remark}

\begin{example}[Safe nonlinear noncommuting memory]
\label{ex:matrix-safe-noncommuting}
Let
\begin{equation}\label{eq:matrix-safe-data}
 A=\begin{pmatrix}1&0\\0&2\end{pmatrix},\quad
 B=\begin{pmatrix}1&1\\0&1\end{pmatrix},\quad
 V(x)=\frac12\|x\|_2^2+\frac14\|x\|_2^4,\quad
 h(x)=B^\transpose(1+\|x\|_2^2)x.
\end{equation}
Then $B^{-\transpose}h=\nabla V$ and
\begin{equation}\label{eq:matrix-safe-dissipation}
 \nabla V(x)^\transpose Ax
 =(1+\|x\|_2^2)(x_1^2+2x_2^2)>0\quad(x\ne0).
\end{equation}
Thus every nonzero input has strictly positive cost.  Yet
\begin{equation}\label{eq:matrix-safe-kernel}
 G(t)=\begin{pmatrix}e^{-t}&e^{-t}\\0&e^{-2t}\end{pmatrix},\qquad
 [G(t),G(s)]
 =\begin{pmatrix}0&e^{-t-2s}-e^{-s-2t}\\0&0\end{pmatrix}\ne0
\end{equation}
for $s\ne t$.  Universal safety does not force symmetry, normality, or
cross-time commutation.
\end{example}

\begin{example}[Orientation separator]\label{ex:matrix-orientation}
Take $A=B=I_2$, let
$J=\left(\begin{smallmatrix}0&-1\\1&0\end{smallmatrix}\right)$, and set
$h(x)=-Jx$.  On four consecutive blocks of common length $\tau$, use rates
$e_1,e_2,-e_1,-e_2$.  Their vector volume is zero and direct integration
gives
\begin{equation}\label{eq:matrix-four-block-cost}
 \cost=-(1-e^{-\tau})^2(3-e^{-2\tau})<0.
\end{equation}
This is already a finite round trip and scales into every positive
cap.  It exposes nonconservativity without any asymptotic argument.
\end{example}

\begin{example}[Stable drift is not enough]\label{ex:matrix-drift}
Let
$A=\left(\begin{smallmatrix}1&3\\0&1\end{smallmatrix}\right)$,
$B=I_2$, and $h(x)=x$.  The matrix exponential is stable, but at
$x=(1,-1)^\transpose$ one has $x^\transpose Ax=-1$.  A long hold at this
state violates all-input safety, and \Cref{thm:vector-remote} closes the
witness to a vector round trip.
\end{example}

\subsection{Permanent matrices and the instantaneous metric}

\begin{proposition}[Invertible permanent matrix]\label{thm:permanent-matrix}
Let $G(t)\equiv P$ with $P\in \mathrm{GL}_d(\R)$ and let
$h\in C^1(\R^d;\R^d)$.  Safety for every vector round trip is equivalent to
$P^{-\transpose}h$ being conservative.  If safety is required for every
input, the normalized potential $V$ must additionally satisfy $V\ge0$.
\end{proposition}

This is the classical statement that a path-independent line integral is a
gradient, in the cyclo-dissipativity language of
\citet[Definition~8 and Theorem~5]{HillMoylan1980}; we record it as the
permanent baseline for the transient results.

\begin{proof}
With $q(t)=\int_0^tv$, one has $D=Pq$ and
\begin{equation}\label{eq:permanent-matrix-line-integral}
 \cost=\int_0^TD'(t)^\transpose P^{-\transpose}h(D(t))\dd t.
\end{equation}
Every polygonal state loop based at zero is generated by a finite-step vector
round trip.  Testing both orientations forces zero circulation and hence
conservativity.  Conversely a conservative field makes every loop cost zero;
on an open path the cost is $V(D(T))$.
\end{proof}

\begin{theorem}[Instantaneous metric obstruction]\label{thm:instant-metric}
Assume $G(t)\to B\in \mathrm{GL}_d(\R)$ as $t\downarrow0$.  If every bounded
finite-step round trip is safe on one fixed horizon, with no common rate cap,
then $B^{-\transpose}h$ is conservative.
\end{theorem}

\begin{proof}
For a polygonal state loop $x:[0,1]\to\R^d$, put $q=B^{-1}x$ and use the
compressed rate $v_\varepsilon(t)=\varepsilon^{-1}q'(t/\varepsilon)$ on
$[0,\varepsilon]$.  Uniformly in loop time,
\begin{equation}\label{eq:matrix-fast-loop-limit}
 D_\varepsilon(\varepsilon\tau)
 =\int_0^\tau G(\varepsilon(\tau-s))q'(s)\dd s\longrightarrow x(\tau).
\end{equation}
The costs converge to the circulation of $B^{-\transpose}h$.  Safety in both
orientations makes every such circulation zero.  Loops away from zero are
joined to zero by a stem and its reverse.
\end{proof}

If $h=\nabla\Psi$ and $B=B^\transpose\succ0$, the condition implies
$[B^{-1},\nabla^2\Psi(x)]=0$ for every $x$.  If the common commutant of the
Hessians is scalar, then $B$ itself must be scalar.  The no-common-cap
hypothesis is essential.

\subsection{Conic left complements}

Matrix multiplication order matters.  A useful complement acts on the left:
\begin{equation}\label{eq:matrix-left-complement}
 (R*G)(t)=\int_{[0,t]}R(\dd r)G(t-r)=1(t)I_d.
\end{equation}
Then $R*D=1*v$ and differentiation recovers the input.

\begin{theorem}[Finite conic left-complement certificate]
\label{thm:matrix-conic}
Assume
\begin{equation}\label{eq:matrix-conic-complement}
 R=\sum_{k=1}^mM_k\bigl(\beta_k\delta_0+\ell_k(t)\dd t\bigr),\qquad
 R*G=1I_d,
\end{equation}
where $\beta_k\ge0$ and each $\ell_k\ge0$ is integrable and
nonincreasing.  Suppose
\begin{equation}\label{eq:matrix-conic-compatibility}
 M_k^\transpose h=\nabla\Psi_k,\qquad
 \Psi_k\text{ convex},\quad
 \Psi_k(0)=0,\quad\nabla\Psi_k(0)=0.
\end{equation}
Then every input has nonnegative cost.  More precisely,
\begin{align}
 \cost_{G,h,T}[v]
 =\sum_{k=1}^m\Bigg\{&
 \beta_k\Psi_k(D(T))+(\ell_k*\Psi_k(D))(T)
 +\int_0^T\ell_k(t)\Phi_k(D(t))\dd t
 \notag\\
 &+\int_{(0,T)}(-\dd\ell_k)(r)
   \int_r^T\mathcal B_k(D(t-r),D(t))\dd t\Bigg\}\ge0,
 \label{eq:matrix-conic-identity}
\end{align}
where
$\Phi_k(x)=x^\transpose\nabla\Psi_k(x)-\Psi_k(x)$ and $\mathcal B_k$ is
the Bregman divergence of $\Psi_k$.
\end{theorem}

\begin{proof}
The complement relation gives
$v=\sum_kM_k[\beta_kD'+(\ell_k*D)']$.  Pair with $h(D)$ and use
\eqref{eq:matrix-conic-compatibility}.  Each summand is the vector convex
version of \Cref{prop:cp-identity}.  Truncation of an unbounded monotone
density and the $W^{1,1}$ convolution closure prove the singular case.
\end{proof}

The matrices $M_k$ need not commute.  What matters is compatibility of each
transformed one-form with a convex storage.  Loewner positivity of the
matrices alone is insufficient; an explicit positive-definite permanent
square-loop counterexample is given in \Cref{app:matrix-details}.

\subsection{General state-space storage and the linear boundary}

For a finite-dimensional realization
\begin{equation}\label{eq:general-realization}
 x'=-Ax+Bv,\qquad D=Cx,\qquad x(0)=0,
\end{equation}
the classical dissipation inequality of \citet{Willems1972},
\citet{HillMoylan1976}, and \citet[Theorem~14]{HillMoylan1980} is a
sufficient certificate that retains the
orientation information: if $W\in C^1(\R^n)$ satisfies
\begin{equation}\label{eq:nonlinear-kyp-conditions}
 W(0)=0,\quad W\ge0,\quad
 B^\transpose\nabla W(x)=h(Cx),\quad
 \nabla W(x)^\transpose Ax\ge0,
\end{equation}
then differentiating $W(x(t))$ along \eqref{eq:general-realization} gives
\begin{equation}\label{eq:nonlinear-kyp-identity}
 \cost_{G,h,T}[v]
 =W(x(T))+\int_0^T\nabla W(x(t))^\transpose Ax(t)\dd t\ge0.
\end{equation}
\Cref{thm:stable-matrix} is the case $C=I$, $B$ invertible, in which these
conditions are also necessary.

For a linear readout $h(D)=JD$ and an integrable matrix kernel, apply the
matrix-valued positive-type criterion of
\citet[Chapter~16, Theorem~2.4]{GripenbergLondenStaffans1990} to
$\mu(\dd s)=JG(s)\dd s$.  Equivalently, Plancherel gives the all-input
condition
\begin{equation}\label{eq:matrix-frequency-test}
 \operatorname{Herm}\!\left(J\laplace G(i\omega)\right)\succeq0
 \quad\text{for almost every }\omega.
\end{equation}
If a nonlinear readout is differentiable at zero, this condition with
$J=Dh(0)$ is necessary by amplitude scaling.  It is not sufficient for the
nonlinear problem.  The conic theorem and the storage certificate
\eqref{eq:nonlinear-kyp-conditions} are broad sufficient mechanisms; only
the permanent and stable fully actuated branches above are classifications.

\section{Fixed readouts, friction, and strategy complexity}
\label{sec:friction-complexity}

The classifications above quantify over the full monotone readout class.
This section records what survives when the readout is fixed, which
frictions repair the rate-inside model, and how many blocks a manipulation
needs.  The first result identifies the state interval that a nonlinear
pre-memory law can generate while carrying zero volume, and shows that a
bounded dead zone cannot hide such a law under power-law memory.  The
second characterizes the power penalties that repair the model with a
finite coefficient.  The third determines the full two-block phase of the
fixed power-law family, relates its two boundaries to Gatheral's rate
regimes, and bounds the number of blocks near $\delta=1$.  The last computes
the critical cubic friction for permanent
memory in closed form.  Proofs that are long but routine are collected in
\Cref{app:friction-details}.

\subsection{Vertical reachability behind a fixed readout}
\label{subsec:vertical-reachability}

Fix a cap $B>0$, a bounded finite-valued map
$f\colon[-B,B]\to\R$, and a continuous readout $h$.  No regularity of $f$
is assumed.  Define the closed graph hull and its vertical section by
\begin{equation}
 K_f:=\overline{\operatorname{co}}
 \{(x,f(x)):-B\le x\le B\},
 \qquad
 V_f:=\{q:(0,q)\in K_f\}=[\nu_-,\nu_+].
 \label{eq:graph-vertical-section}
\end{equation}
For $x<0<y$, let
\begin{equation}
 m_f(x,y):=\frac{y f(x)-x f(y)}{y-x}.
 \label{eq:two-point-vertical-moment}
\end{equation}
The weights $y/(y-x)$ and $-x/(y-x)$ are the unique zero-mean
probability weights on $\{x,y\}$, so $m_f(x,y)$ is their mean impact.

\begin{lemma}[Static vertical section]
\label{lem:vertical-section}
One has
\begin{equation}
 V_f=\overline{\operatorname{co}}
 \left(\{f(0)\}\cup
 \{m_f(x,y):-B\le x<0<y\le B\}\right).
 \label{eq:vertical-section-formula}
\end{equation}
Moreover, every $q\in V_f$ is a limit of finite graph mixtures whose rate
barycenter is exactly zero.
\end{lemma}

\begin{proof}
Every expression on the right of \eqref{eq:vertical-section-formula} is a
zero-rate graph barycenter, hence belongs to $V_f$.  Conversely, take finite
graph barycenters $(m_n,q_n)\to(0,q)$.  If $m_n>0$, mix the barycenter with
one fixed graph point $(\xi_-,f(\xi_-))$, $\xi_-<0$, using weight
$m_n/(m_n-\xi_-)$.  The new first coordinate is zero and the added weight
tends to zero.  Use a fixed $\xi_+>0$ when $m_n<0$.  Boundedness of $f$
shows that the corrected second coordinate still converges to $q$.

It remains to decompose a finite exact-zero-mean law.  Separate its mass at
zero, its negative support, and its positive support.  The total negative
and positive first moments agree.  Coupling the corresponding finite moment
measures decomposes the law into a convex combination of zero-mean
two-point laws and a possible atom at zero.  This proves the formula and the
last assertion.
\end{proof}

For a signed kernel $G\in L^1(0,T)$, put
\begin{equation}
 P(t):=\int_0^t G(r)_+\dd r,
 \qquad
 N(t):=\int_0^t G(r)_-\dd r\le0,
 \label{eq:signed-cumulative-kernel}
\end{equation}
where $G_-:=\min(G,0)$.  Center only the readout,
\begin{equation}
 k(z):=h(z)-h(0),\qquad Z_h:=k^{-1}(0).
 \label{eq:centered-readout-zero-set}
\end{equation}
The source itself is not centered: a zero rate continues to generate
$f(0)$.

\begin{theorem}[Vertical-reachability obstruction]
\label{thm:vertical-reachability}
Suppose the mixed model
\begin{equation*}
 D(t)=\int_0^tG(t-s)f(v(s))\dd s,
 \qquad
 \cost[v]=\int_0^T v(t)h(D(t))\dd t,
\end{equation*}
is safe for every finite piecewise-constant round trip with
$|v|\le B$.  Then, for every $0<t<T$,
\begin{equation}
 [\nu_-P(t)+\nu_+N(t),\;
  \nu_+P(t)+\nu_-N(t)]\subset Z_h.
 \label{eq:vertical-reachable-interval}
\end{equation}
The interval is the limiting state range of ordinary finite pumps
whose volume is zero on every cell of a partition with vanishing mesh.
\end{theorem}

The proof is given in \Cref{app:friction-details}.

\begin{corollary}[Divergent memory removes bounded dead zones]
\label{cor:divergent-memory-deadzone}
Suppose $G\ge0$ and
$A(t):=\int_0^tG(r)\dd r\to\infty$.  If safety holds on every horizon and
$Z_h$ is bounded, then $V_f=\{0\}$ and
\begin{equation}
 f(x)=\lambda x,\qquad -B\le x\le B,
 \label{eq:capped-local-linearity}
\end{equation}
for one $\lambda\in\R$.  In particular, a bounded dead zone cannot hide a
nonlinear $f$ under fractional memory $G(t)=t^{-\gamma}$,
$0<\gamma<1$.
\end{corollary}

\begin{proof}
\Cref{thm:vertical-reachability} gives $A(t)V_f\subset Z_h$.  A
nonzero point of $V_f$ would generate an unbounded ray, contradicting the
boundedness of $Z_h$.  Hence $V_f=\{0\}$.  Formula
\eqref{eq:vertical-section-formula} then gives $f(0)=0$ and
$yf(x)=xf(y)$ for every $x<0<y$.  Fixing one rate of each sign shows that
$f(x)/x$ is the same constant at every nonzero rate.
\end{proof}

The bounded-zero-set hypothesis is sharp in kind.  If $G\ge0$, $f\ge0$,
and $[0,\infty)\subset Z_h$, then every attainable state lies in $Z_h$ and
every round-trip cost is exactly zero.  Universal readout quantifiers in the
earlier rigidity theorems rule out precisely this kind of fixed-readout
blindness.

\begin{remark}[Dead zones and minimum price increments]\label{rem:deadzone-spread}
A readout that vanishes on an interval around zero is the natural model of
a minimum price increment seen from the impact state: small accumulated
impact produces no observable move.
\Cref{cor:divergent-memory-deadzone} says that under power-law memory such
a band, however wide, cannot shelter a concave pre-memory law: the pump
accumulates state without bound and eventually reaches the responsive
region.  Only a readout that is flat on a whole half-line, which is no
longer a response, can do so.
\end{remark}

\subsection{Which collections of power penalties can repair the model?}
\label{subsec:multipower-repair}

Let
\begin{equation}
 \Psi(x)=\sum_{j=1}^k a_j|x|^{p_j},
 \qquad a_j>0,\quad p_j>0,
 \label{eq:multipower-penalty}
\end{equation}
and multiply the entire penalty by one coefficient $\kappa$.  Write
$p_{\min}:=\min_jp_j$ and $p_{\max}:=\max_jp_j$.  For a kernel $G$ and
the power law $f_\delta(x)=\operatorname{sgn}(x)|x|^\delta$ of
\Cref{cor:all-powers}, $\cost_\delta[v]$ denotes the rate-inside cost
\eqref{eq:rate-inside-cost} with $H(t,s)=G(t-s)$ and $f=f_\delta$; the
same convention defines $\cost_1$ and $\cost_2$.

\begin{theorem}[Exponent test for finite repair]
\label{thm:multipower-repair}
Fix $T>0$, a nonzero real kernel $G\in L^1(0,T)$, and
$\delta>0$, $\delta\ne1$.  Under a common cap $|v|\le B$, a finite
coefficient $\kappa\ge0$, allowed to depend on $G,T,B,\delta,\Psi$, can make
\begin{equation*}
 \cost_\delta[v]+\kappa\int_0^T\Psi(v(t))\dd t
\end{equation*}
nonnegative on every finite round trip if and only if
\begin{equation}
 p_{\min}\le q=1+\delta.
 \label{eq:capped-repair-phase}
\end{equation}
Without a common cap, such a finite coefficient exists if and only if
\begin{equation}
 p_{\min}\le q\le p_{\max}.
 \label{eq:uncapped-repair-phase}
\end{equation}
If instead $0\ne G\in L^1(0,\infty)$, the same criteria characterize a
single finite coefficient valid on every finite horizon.
\end{theorem}

\begin{proof}
On the cap, the ratio $x^q/\Psi(x)$ is bounded on $(0,B]$ exactly when
$p_{\min}\le q$.  Then the generic convolution estimate
\begin{equation}
 |\cost_\delta[v]|
 \le\|G\|_{L^1(0,T)}\int_0^T|v|^q\dd t
\end{equation}
gives a finite repair with
$\kappa=\|G\|_{L^1(0,T)}\sup_{0<x\le B}x^q/\Psi(x)$.
The capped rigidity theorem supplies a negative finite round trip because
$G\ne0$ and $f_\delta$ is nonlinear.  If every $p_j>q$, scale this witness
toward zero: its negative cost is of order $a^q$, while the penalty is
$o(a^q)$.

Without a cap, boundedness of $x^q/\Psi(x)$ on $(0,\infty)$ is equivalent to
\eqref{eq:uncapped-repair-phase}.  The same small-amplitude argument proves
necessity at zero and scaling the witness to infinity proves necessity at
infinity.  When the exponents bracket $q$, the ratio is bounded, either
because a $q$-power is present or because it tends to zero at both endpoints.
For the all-horizon version, use $\|G\|_{L^1(0,\infty)}$ in the same bound;
necessity uses any one horizon on which $G$ is nonzero.
\end{proof}

\begin{remark}[Spreads, temporary costs, and concave impact]\label{rem:spread-repair}
A proportional bid--ask spread charges $\kappa\int|v|$, the case
$\Psi(x)=|x|$, $p_{\min}=p_{\max}=1$.  Since $q=1+\delta>1$,
\Cref{thm:multipower-repair} says that a spread repairs the rate-inside
model inside a rate cap $B$, with a coefficient that grows like
$\|G\|_1B^\delta$, but never without a cap: scaling a manipulation up by
the factor $a$ multiplies its negative cost by $a^{1+\delta}$ and the
spread by $a$.  A spread thus converts manipulation into a size threshold
rather than removing it, which is the sense in which the frictionless
theorems are weakened, not overturned, by slippage.  A quadratic temporary
cost alone, $p_{\min}=2$, does not repair concave impact even inside a cap:
a chattering manipulation of amplitude $a$ costs $-ca^{1+\delta}$ and is
penalized by $O(a^2)$, so it survives at small amplitude whenever
$\delta<1$.  A spread together with a quadratic cost, $p_{\min}=1$ and
$p_{\max}=2$, repairs every concave law $0<\delta\le1$ without a cap, and
this pair is the natural friction structure: a spread controls small
trades and a quadratic cost controls large ones.
\end{remark}

\subsection{Fractional memory: the two-block phase and a complexity gap}
\label{subsec:fractional-complexity}

Take $G(t)=t^{-\gamma}$ with $0<\gamma<1$ and let
$q_\gamma:=2-\gamma$.  For two adjacent blocks, normalize the larger rate
magnitude to one and the smaller to $r\in(0,1)$.  Zero volume fixes their
durations to $r/(1+r)$ and $1/(1+r)$.  Put
\begin{equation}
 A_q(r):=(1+r)^q-1-r^q.
 \label{eq:fractional-Aq}
\end{equation}
After removal of a strictly positive factor, the adverse-order cost has sign
equal to
\begin{equation}
 S_{q,\delta}(r)
 =r^q+r^{\delta+1}
 -r^{\min(1,\delta)}A_q(r),
 \qquad q=2-\gamma.
 \label{eq:fractional-two-block-sign}
\end{equation}

\citet[Lemma~5.1]{Gatheral2010} derives the slow-rate condition
$\delta+\gamma\ge1$ for the homogeneous power-law family from the limit
$r\to0$.  His Lemma~5.2 treats a different impact law that blows up at a
maximal rate and obtains $\gamma\ge2-\log_23$.  The same constant enters the
fixed-power-law problem through Gatheral's Appendix~A: in the notation below,
his function $h(r,\gamma)$ equals $-(A_q(r)-r^{q-1})$ and is negative for
some $r\in(0,1)$ exactly when $\gamma<2-\log_23$.  The theorem determines
the full two-block phase for the fixed power-law family: it makes the
slow-rate boundary sharp on the sublinear side and locates the distinct
superlinear threshold.

\begin{theorem}[Fractional two-block phase]
\label{thm:fractional-two-block-phase}
On the sublinear side, a two-block manipulation exists if and only if
\begin{equation}
 \delta<1-\gamma.
 \label{eq:fractional-sublinear-boundary}
\end{equation}
On the superlinear side, define
\begin{equation}
 \delta_+(\gamma):=
 \inf_{\substack{0<r<1\\ A_q(r)>r^{q-1}}}
 \frac{\log(A_q(r)-r^{q-1})}{\log r},
 \qquad q=2-\gamma,
 \label{eq:fractional-superlinear-threshold}
\end{equation}
with the infimum of the empty set equal to $+\infty$.  A two-block
manipulation exists if and only if
\begin{equation}
 \delta>\delta_+(\gamma).
 \label{eq:fractional-superlinear-phase}
\end{equation}
Moreover,
\begin{equation}
 \delta_+(\gamma)=+\infty
 \quad\Longleftrightarrow\quad
 \gamma\ge2-\log_2 3.
 \label{eq:infinite-superlinear-threshold}
\end{equation}
Consequently, throughout
$1-\gamma\le\delta\le\delta_+(\gamma)$, $\delta\ne1$, every manipulating
finite strategy has at least three blocks.
\end{theorem}

\begin{proof}
If $\delta<1$, divide \eqref{eq:fractional-two-block-sign} by $r^\delta$.
The sign is that of
\begin{equation}
 r^{q-\delta}+r-A_q(r).
 \label{eq:sublinear-two-block-sign}
\end{equation}
For $1<q<2$, one has $A_q(r)<2r$: differentiate
$1+r^q+2r-(1+r)^q$ and use subadditivity of $x^{q-1}$.  If
$\delta\ge q-1$, then $r^{q-\delta}\ge r$, so
\eqref{eq:sublinear-two-block-sign} is positive.  If $\delta<q-1$, the
expansion $A_q(r)=qr-r^q+O(r^2)$ makes it negative for small $r$.  This
proves \eqref{eq:fractional-sublinear-boundary}.

If $\delta>1$, division by $r$ shows that negativity is equivalent to
$r^\delta<A_q(r)-r^{q-1}$, which is
\eqref{eq:fractional-superlinear-threshold}--
\eqref{eq:fractional-superlinear-phase}.  Finally,
$H_q(r):=A_q(r)/r^{q-1}$ is strictly increasing: writing
$\psi(u)=(1+u)^{q-1}-u^{q-1}$, which is strictly decreasing, gives
$H_q(r)=q\int_0^1\psi(t/r)\dd t$.  Also $H_q(1)=2^q-2$.  Hence the admissible set in
\eqref{eq:fractional-superlinear-threshold} is nonempty exactly when
$2^q-2>1$, equivalent to the strict reverse of
\eqref{eq:infinite-superlinear-threshold}.  When the set is nonempty, the quotient in
\eqref{eq:fractional-superlinear-threshold} diverges at both endpoints of
its domain, so its minimum is attained.  Strict positivity of the linear
fractional cost (\Cref{prop:fractional-positive}) then gives
$\delta_+(\gamma)>1$.  Rate-inside rigidity supplies some finite
negative witness whenever $\delta\ne1$, so absence of a two-block witness
forces at least three blocks.
\end{proof}

The last statement is a complexity lower bound that requires no dwell-time
assumption.  When the number of blocks and their minimum durations are
both fixed, a quantitative safe neighborhood of $\delta=1$ exists
(\Cref{prop:dwell-safe-neighborhood} in \Cref{app:friction-details}).
Without a positive dwell, block lengths and smaller rate magnitudes can
collapse, and compactness alone gives no cap-only lower bound.  There is,
however, a quantitative upper bound from an explicit pump construction.

\begin{proposition}[Constructive near-linear block bound]
\label{prop:constructive-block-bound}
Fix $0<\gamma<1$.  For all $\delta\ne1$ sufficiently close to one, fractional
memory admits a capped finite round trip with negative cost and at
most
\begin{equation}
 C_\gamma|\delta-1|^{-2/(1-\gamma)}
 \label{eq:constructive-block-bound}
\end{equation}
blocks.  The same block bound holds on any prescribed positive horizon and
under any prescribed positive symmetric rate cap.
\end{proposition}

\begin{proof}
Normalize the horizon and cap to one and put $\alpha=1-\gamma$.
Let $I_-=[0,1/4]$, $P=[1/4,3/4]$, $I_+=[3/4,1]$, and
$b=\one_{I_-}-\one_{I_+}$.  On each of $M$ equal cells in $P$, repeat
rates $1,-1/2$ for fractions $1/3,2/3$ of the cell.  This pump $u_M$ has
zero volume on every cell and mean impact
\begin{equation*}
 \eta_\delta=\frac{1-2^{1-\delta}}3
 =\frac{\log2}{3}(\delta-1)+O((\delta-1)^2).
\end{equation*}
Write $\mathcal B(v,w)=\int_{s<t}(t-s)^{-\gamma}v(t)w(s)\dd s\dd t$.
For $v_M=\sigma a b+u_M$, $\sigma=\operatorname{sgn}\eta_\delta$, the
limiting rate/source pair is
$(\sigma ab,\sigma a^\delta b+\eta_\delta\one_P)$, with cost
\begin{equation*}
 A a^{1+\delta}-D|\eta_\delta|a,
 \qquad A=\mathcal B(b,b)>0,\quad
 D=\int_{I_+}\int_P(t-s)^{-\gamma}\dd s\dd t>0.
\end{equation*}
Choose $a=(D|\eta_\delta|/((1+\delta)A))^{1/\delta}\le1$.
The negative margin is
$\delta A a^{1+\delta}\ge c_\gamma|\delta-1|^2$ near one, since
$|\delta-1|^{1/\delta-1}\to1$.

We need a uniform recovery estimate.  If $w$ is bounded, supported on $P$,
and has zero mean on every pump cell of width $h=1/(2M)$, cancellation on
cells at distance at least $h$ from $t$, and a direct estimate on the at most
two remaining cells, give
\begin{equation}
 \left|\int_0^t(t-s)^{-\gamma}w(s)\dd s\right|
 \le C_\gamma\|w\|_\infty h^\alpha.
 \label{eq:fractional-cell-cancellation}
\end{equation}
Indeed, the near cells contribute at most
$\|w\|_\infty(2h)^\alpha/\alpha$; on the others, subtract a cellwise
constant kernel value and use
$h\int_h^\infty\gamma r^{-\gamma-1}\dd r=h^\alpha$.
Time reversal gives the same estimate for the upper-triangular integral.
Set $e_M=f_\delta(u_M)-\eta_\delta\one_P$, extended by zero off $P$.
Both $u_M,e_M$ have cell mean zero and uniformly bounded magnitudes.
Applying \eqref{eq:fractional-cell-cancellation} to
\begin{equation*}
 \mathcal B(v_M,e_M)
 +\mathcal B(u_M,\sigma a^\delta b+\eta_\delta\one_P)
\end{equation*}
therefore bounds the cost-recovery error by $C_\gamma M^{-\alpha}$,
uniformly for $\delta$ near one.  Taking
$M\ge C_\gamma|\delta-1|^{-2/\alpha}$ preserves half the negative margin.
There are $2M+2$ blocks and volume is exactly zero.  Finally,
$v(t)=Bv_M(t/T)$ multiplies cost by the positive factor
$B^{1+\delta}T^{2-\gamma}$ and leaves the block count unchanged.
\end{proof}

This is an upper bound for the displayed architecture, not a matching lower
bound on every possible strategy.

\subsection{Permanent memory: the critical cubic friction}
\label{subsec:permanent-friction}

Let
\begin{equation}
 f_\delta(x)=\operatorname{sgn}(x)|x|^\delta,
 \qquad q:=1+\delta>1,
 \label{eq:signed-power-law}
\end{equation}
and take the permanent kernel $G\equiv1$.  If
$X(t):=\int_0^t v(s)\dd s$ and $X(0)=X(T)=0$, then Fubini gives
\begin{equation}
 \cost_\delta[v]
 =-\int_0^T X(t)f_\delta(v(t))\dd t.
 \label{eq:permanent-inventory-cost}
\end{equation}
At the critical penalty exponent $p=q$, the smallest coefficient that makes
$\cost_\delta[v]+\kappa\int|v|^q$ nonnegative on every round trip is
therefore
\begin{equation}
 \mathcal K_q^{\rm RT}(T)
 :=\sup_{0\ne X\in W^{1,q}_0(0,T)}
 \frac{\displaystyle\int_0^T
 X|X'|^{q-2}X'\dd t}
 {\displaystyle\int_0^T|X'|^q\dd t}.
 \label{eq:critical-permanent-coefficient}
\end{equation}
This is a signed Volterra numerical-radius problem.  Absolute-value Opial
constants do not answer it; for example, the numerator vanishes identically
at $q=2$.

For later bounds define the scalar antisymmetry constant
\begin{equation}
 M_q:=\max_{0\le r\le1}
 \frac{|r^{q-1}-r|}{1+r^q}.
 \label{eq:antisymmetry-constant}
\end{equation}
The best pointwise estimate
\begin{equation}
 |x f_\delta(y)-y f_\delta(x)|
 \le M_q(|x|^q+|y|^q)
 \label{eq:pointwise-antisymmetry}
\end{equation}
follows by homogeneity and reduction to the ratio of the smaller to the
larger magnitude.  Antisymmetrizing the permanent cost yields
\begin{equation}
 [-\cost_\delta[v]]_+
 \le\frac{TM_q}{2}\int_0^T|v|^q\dd t.
 \label{eq:permanent-antisymmetry-bound}
\end{equation}
Thus a finite critical repair always exists.  At $q=3$ the critical
coefficient can be computed in closed form.

\begin{theorem}[Critical cubic repair]
\label{thm:exact-cubic-repair}
Define
\begin{align}
 x_*&:=\frac12-\frac{\sqrt3}{6}
       +\frac{\sqrt2\,3^{1/4}}6,
 &
 y_*&:=-\frac12+\frac{\sqrt3}{6}
       +\frac{\sqrt2\,3^{1/4}}6,\label{eq:xy-star}\\
 C_*&:=\frac23\left[
 \log\frac{x_*}{x_*-1/2}+\frac1{x_*}
 +\log\frac{y_*+1/2}{y_*}+\frac1{y_*}\right],
 &
 \kappa_*&:=C_*^{-1}.\label{eq:kappa-star}
\end{align}
Then
\begin{equation}
 \mathcal K_3^{\rm RT}(T)=T\kappa_*,
 \qquad
 \kappa_*=0.0881331300961391335081086746213\ldots.
 \label{eq:exact-cubic-constant}
\end{equation}
The supremum is attained by a bounded ordinary rate whose inventory has one
positive hump and whose rate changes sign once.  The value is unchanged by
any positive symmetric rate cap after amplitude scaling, and finite
piecewise-constant rates approach it.
\end{theorem}

The proof is given in \Cref{app:friction-details}.  The extremizer is not
a two-block or bang--bang strategy.  For $T=B=1$ the two-block optimum is
$0.07507\ldots$, an explicit four-block strategy reaches $10448/136205
=0.07670\ldots$, and an exact rational eighty-block certificate reaches
$0.088107\ldots$, approaching $\kappa_*$ from below
(\Cref{prop:four-block-friction} in \Cref{app:friction-details}).

\FloatBarrier\section{Synthesis and conclusion}
\label{sec:discussion}

The classifications separate three mechanisms.  A law applied to the rate
before memory is constrained by source rigidity.  A monotone law applied
after memory is governed by complementary storage.  In vector state space,
conservativity of the transformed readout supplies an additional orientation
condition.  \Cref{tab:classification-summary} collects the tests.

\begin{table}[ht]
\centering
\small
\caption{Classifications established in the paper.  The cap and readout
quantifiers are those of the cited theorems; the fixed matrix readout is
tested globally without a cap.}
\label{tab:classification-summary}
\begin{tabular}{@{}L{0.22\textwidth}L{0.33\textwidth}L{0.33\textwidth}@{}}
\toprule
Model & Safety condition & Decisive mechanism \\
\midrule
Scalar rate-inside & $f=c+\lambda x$ with $cA_H$ constant;
$\lambda Q_H\ge0$ on round trips & local chattering plus two baselines \\
Scalar state-outside, all inputs & $G$ completely positive & complementary
storage and hinge separation \\
Mixed scalar, all inputs & $f(x)=\lambda x$, $\lambda G$ completely positive & source
rigidity followed by kernel converse \\
Uniformly vanishing tail, round trips & same as corresponding all-input problem
& remote zero-volume compensation \\
Two-mode Prony & nonnegative modal weights, or $gG(0+)>0$ and
$\mu_1\le\nu\le\mu_1+\mu_2$ for the inverse zero $\nu$ and poles
$-\mu_1,-\mu_2$ & shifted LICM storage, stable inverse, explicit separators \\
Time-inhomogeneous first order & $\rho\ge0$ and $b'+\rho b\ge0$ & storage
and state-volume identities \\
Stable matrix first order & $B^{-\transpose}h$ conservative and
$(B^{-\transpose}h(x))^\transpose Ax\ge0$ & nonlinear dissipativity and
orientation separators \\
\bottomrule
\end{tabular}
\end{table}

Round trips detect every all-input defect when the memory tail
vanishes uniformly.  Permanent memory instead leaves the storage quotient
$F_h(gm)/g$.  This distinction explains both the signed safe Prony sectors
and the failure of all-input complete positivity as a universal permanent
round-trip criterion.  Individual modal weights do not determine safety:
the inverse-memory order does.

The quantifiers remain essential.  A single fixed readout can hide a defect
outside its reachable state set; universal readouts cannot.
\Cref{thm:vertical-reachability} identifies this blindness for
cellwise zero-volume pumps.  Finite controls suffice for every separator,
whereas relaxed controls serve only as intermediate constructions.  In the
opposite direction, the critical cubic-friction constant holds even after
Young-measure relaxation.  Its extremizer, the power-repair
criterion, the two-block phase, and the block-count bounds quantify the
effects of costs and execution complexity beyond the structural
classifications.

The finite-Prony hierarchy assumes nonzero instantaneous gain and no
imaginary-axis inverse poles.  Removing these assumptions at arbitrary mode
count, or allowing marginal stability and singular actuation in the matrix
model, requires additional inverse or reachable-state analysis.  General
nonlinear matrix convolution is a further extension beyond first-order
realizations.  These are distinct mathematical questions; the classifications
here apply with the hypotheses and input classes stated in their theorems.

The central conclusion is that universal no-manipulation determines where
nonlinearity can be placed and which memory geometry it must preserve.
Linear quadratic positivity alone does not capture these requirements.  For
the empirical question that motivated the paper the answer is concrete: a
concave power law applied to the trading rate before power-law memory is
manipulable at every exponent pair, and the same law applied to the
impact state after the same memory is not.  The proofs and the finite
certificates in \Cref{app:verification} make each necessary condition
testable without replacing a universal assertion by finite-grid evidence.

\section{Use of generative AI}
\label{sec:ai-use}

This work used Claude Fable~5 and OpenAI GPT-5.6 Sol.  Claude Fable~5
assisted with problem selection and whole-manuscript review and revision.
GPT-5.6 Sol assisted with exploratory construction, proof development and
checking, counterexample search, assumption and literature auditing,
computational work, and manuscript drafting and revision.

The author selected the research question and acceptance criteria, directed
the work, evaluated conflicting outputs, and decided which arguments, claims,
citations, and formulations to retain.

\appendix
\section{Technical closure for rate-inside localization}
\label{app:rate-localization}

This appendix records the measure-theoretic and approximation details
behind \Cref{thm:exact,thm:cap}: exact volume in the pump, integrable
singularities, the affine intercept, and the role of the accessible plateau
set.

\subsection{Exact volume in the pump}

For two plateau rates $x<0<y$, the duty weights
\begin{equation}
 \alpha=\frac{y}{y-x},\qquad 1-\alpha=\frac{-x}{y-x}
 \label{eq:app-duty-weights}
\end{equation}
satisfy $\alpha x+(1-\alpha)y=0$ exactly.  Dividing every pump cell in these
real proportions therefore gives zero volume on each complete cell; no
rational approximation of the duty cycle is needed.  The general principle
is to preserve the rate and impact moments separately.

\begin{lemma}[Finite-support joint realization]
\label{lem:finite-support-realization}
Let $a_1,\ldots,a_m\in\R$, let $f(a_j)$ be finite, and let measurable
weights $\alpha_j:[0,T]\to[0,1]$ satisfy $\sum_j\alpha_j=1$ almost
everywhere.  Set
\begin{equation}\label{eq:joint-relaxed-moments}
 u=\sum_{j=1}^m\alpha_j a_j,\qquad
 w=\sum_{j=1}^m\alpha_j f(a_j).
\end{equation}
If $\int_0^T u=0$, there are $v_n\in\rt(T)$ taking values only in
$\{a_1,\ldots,a_m\}$ such that
\begin{equation}\label{eq:joint-realization-limits}
 v_n\weakstar u,\qquad f(v_n)\weakstar w
 \quad\text{in }L^\infty(0,T).
\end{equation}
For every $H\in L^1(\DeltaT)$,
\begin{equation}\label{eq:joint-realization-cost}
 \cost^{\mathrm{in}}_{H,f,T}[v_n]
 \longrightarrow\int_{\DeltaT}H(t,s)u(t)w(s)\dd s\dd t.
\end{equation}
In particular, a strictly negative right side has a finite round-trip
representative with cap at most $\max_j|a_j|$.  No regularity of $f$ is
required.
\end{lemma}

\begin{proof}
Take finite interval partitions $\mathcal P_n$ with mesh tending to zero.
On each cell $I$, put
$\bar\alpha_{j,I}=|I|^{-1}\int_I\alpha_j$ and subdivide $I$ into
consecutive intervals of lengths $|I|\bar\alpha_{j,I}$, omitting zero
lengths.  Set $v_n=a_j$ on the corresponding interval.  This gives exactly
\begin{equation}\label{eq:joint-cell-moments}
 \int_Iv_n=\int_Iu,\qquad
 \int_If(v_n)=\int_Iw.
\end{equation}
Thus the total volume is zero, the control is a finite step function, and
the prescribed cap is preserved.  Given $\psi\in L^1(0,T)$, let
$\psi_n=\mathbb E[\psi\mid\mathcal P_n]$.  Interval averaging is an
$L^1$ contraction; approximation of $\psi$ by continuous functions shows
$\|\psi_n-\psi\|_1\to0$ whenever the mesh tends to zero.  The two
pairings against $\psi_n$ vanish by \eqref{eq:joint-cell-moments}.  The
uniform bounds on $v_n,u,f(v_n),w$ therefore give
\eqref{eq:joint-realization-limits}.  Apply \Cref{lem:tensor} to the
zero extension of $H$ from $\DeltaT$ to $(0,T)^2$ to obtain
\eqref{eq:joint-realization-cost}.  A sufficiently large finite $n$
retains any strict negative margin.
\end{proof}

The moment $w$ in \eqref{eq:joint-relaxed-moments} need not equal $f(u)$,
and ordinary conditional averaging of a control, which preserves volume,
need not preserve the limiting cost: for $T=1$, $H\equiv1$, $v(t)=1/2-t$,
and $f=\one_{\R\setminus\mathbb Q}$, the cost is $-1/12$, while the cell
averages of $v$ on uniform partitions are rational and have zero cost.  The
finite-support construction above, which preserves the rate and impact
moments separately, avoids this; the proof of \Cref{prop:normalized} uses
the same two-moment structure through \Cref{lem:pump-average,lem:tensor}.

\subsection{Weak singularities and diagonal dilution}

Only $L^1(\DeltaT)$ integrability is used.  If $I_n$ is an interval with
$|I_n|\to0$, absolute continuity of the Lebesgue integral gives
\begin{equation}
 \int_{I_n\times I_n}|H(t,s)|\dd s\dd t\longrightarrow0.
 \label{eq:diagonal-absolute-continuity}
\end{equation}
This replaces every pointwise boundedness estimate near the diagonal.  The
ordered off-diagonal interaction is preserved by choosing Lebesgue points of
the two-variable kernel and shrinking source and target rectangles around a
point $(t_*,s_*)$ with $s_*<t_*$.  The diagonal and reverse-order rectangles
vanish by \eqref{eq:diagonal-absolute-continuity}; the desired rectangle
converges to the nonzero localized kernel mass.  The proof therefore applies
unchanged to $H(t,s)=(t-s)^{-\gamma}$ for $0<\gamma<1$.

\subsection{Why the general theorem is affine}

The constant term of $f(x)=c+\lambda x$ contributes
\begin{equation}
 c\int_0^Tv(t)A_H(t)\dd t,\qquad
 A_H(t)=\int_0^tH(t,s)\dd s.
 \label{eq:intercept-app}
\end{equation}
It vanishes on every round trip exactly when $cA_H$ is almost everywhere
constant, which for $c\ne0$ means that $A_H$ itself is almost everywhere
constant.  This can occur for a nonstationary nonzero kernel: for example,
$H(t,s)=1/t$ on $0<s<t$ has $A_H(t)=1$.  For a convolution kernel,
$A_H(t)=\int_0^tG(r)\dd r$; if it is constant, its distributional derivative
is $G=0$.  Hence a nonzero convolution kernel forces the intercept to vanish.

\subsection{The accessible plateau set is part of identification}

The rigidity statement identifies $f$ on the set of plateau values that the
admissible controls can use.  If tests were restricted to rational plateaus,
the fractional kernel could not distinguish $f(x)=x$ on $\mathbb Q$,
$f(x)=2x$ off $\mathbb Q$, from the identity.  Under a cap, the same
argument identifies one affine law on the full interval $[-B,B]$ and says
nothing outside it.  Finally, the tensor-density step is qualitative for a
general $L^1$ kernel: it proves existence of a finite switching level but
no switching rate uniform over the $L^1$ class.  Quantitative rates are
given only for the fractional kernel, in \Cref{sec:friction-complexity}.

\section{Discrete complete accretivity and the scalar converse}
\label{app:scalar-converse}

The continuous complementary-kernel identity has a finite-dimensional
counterpart: complete accretivity in the sense of
\citet{BenilanCrandall1991}, whose characterization by monotone truncations
is the mechanism behind the readout probes of \Cref{sec:outside}.  We
record the finite-dimensional statement in the row-sum and column-sum form
used by the Toeplitz corollary, with its short proof.

Let $A=(a_{ij})\in\R^{n\times n}$ and define
\begin{equation}
 \mathcal E_{A,h}(x):=\sum_{i=1}^n(Ax)_i h(x_i),
 \label{eq:discrete-accretive-form}
\end{equation}
where $h\colon\R\to\R$ is continuous, nondecreasing, and $h(0)=0$.

\begin{theorem}[Discrete complete accretivity]
\label{thm:discrete-complete-accretivity}
The inequality
\begin{equation}
 \mathcal E_{A,h}(x)\ge0
 \quad\text{for every }x\in\R^n
 \text{ and every such }h
 \label{eq:discrete-universal-accretivity}
\end{equation}
holds if and only if
\begin{equation}
 a_{ij}\le0\ (i\ne j),\qquad
 A\one\ge0,\qquad A^\transpose\one\ge0.
 \label{eq:doubly-substochastic-M}
\end{equation}
No symmetry or diagonalizability is required.
\end{theorem}

\begin{proof}
For distinct $i,j$, put $x_i=a>0$, $x_j=-R$, and all other coordinates
equal to zero.  Choose a nondecreasing readout that vanishes on
$(-\infty,0]$ and equals one at $a$.  Only the $i$th term contributes, so
$aa_{ii}-Ra_{ij}\ge0$ for every $R>0$, forcing $a_{ij}\le0$.

Next set $x_i=a$ and all other coordinates to $a-\varepsilon$, and choose a
readout that vanishes up to $a-\varepsilon$ and equals one at $a$.  Letting
$\varepsilon\downarrow0$ gives the $i$th row sum nonnegative.  Finally take
an arbitrary strictly positive $x$ and approximate the constant-one readout
on its compact coordinate range by normalized continuous nondecreasing
readouts.  Equation \eqref{eq:discrete-universal-accretivity} yields
$\one^\transpose Ax\ge0$ for every $x>0$.  Approximation of coordinate rays
gives $A^\transpose\one\ge0$.

Conversely, let $P_t=e^{-tA}$.  The off-diagonal signs in
\eqref{eq:doubly-substochastic-M} make $P_t$ entrywise nonnegative.  Also
\begin{equation}
 \one-P_t\one=\int_0^tP_sA\one\dd s\ge0,
 \qquad
 \one-P_t^\transpose\one
 =\int_0^tP_s^\transpose A^\transpose\one\dd s\ge0.
 \label{eq:doubly-substochastic-semigroup}
\end{equation}
Thus every row and column sum of $P_t$ is at most one.  Let $F'=h$ and
$F(0)=0$.  Then $F$ is convex and nonnegative.  Row-wise Jensen, assigning
the missing row mass to zero, gives
\begin{equation}
 F((P_tx)_i)\le\sum_j(P_t)_{ij}F(x_j).
\end{equation}
Summing and using the column bounds shows
$\sum_iF((P_tx)_i)\le\sum_iF(x_i)$.  Its right derivative at $t=0$ is
$-\mathcal E_{A,h}(x)$, proving sufficiency.
\end{proof}

\begin{corollary}[Lower-triangular Toeplitz form]
\label{cor:toeplitz-complete-accretivity}
Suppose
\begin{equation}
 (A_nx)_i=\sum_{j\le i}a_{i-j}x_j.
 \label{eq:lower-toeplitz}
\end{equation}
Then \eqref{eq:discrete-universal-accretivity} holds if and only if
\begin{equation}
 a_k\le0\quad(1\le k<n),\qquad
 \sum_{k=0}^m a_k\ge0\quad(0\le m<n).
 \label{eq:toeplitz-partial-sums}
\end{equation}
\end{corollary}

\begin{proof}
The off-diagonal entries of $A_n$ are $a_k$, $k\ge1$.  Its row sums are the
partial sums in \eqref{eq:toeplitz-partial-sums}, and its column sums are the
same list in reverse order.  Apply
\Cref{thm:discrete-complete-accretivity}.
\end{proof}

For a nonnegative nonincreasing complementary density, the coefficients
$a_k$ in a uniform causal discretization are successive decreases of the
complement, and the partial sums are its nonnegative levels.  Thus
\eqref{eq:toeplitz-partial-sums} is the discrete inverse-memory
pattern behind \Cref{prop:cp-identity}; it is stronger than positive
semidefiniteness of the symmetric part.

\subsection{Closure of the resolvent probes}

The resolvent inputs used in \Cref{lem:r-resolvent-probe,lem:s-resolvent-probe}
need not themselves be finite step functions.  Let $u$ be any bounded
compactly supported input supplied by those constructions.  On a refining
finite partition $\mathcal P_n$, set $u_n=\mathbb E[u\mid\mathcal P_n]$.
Then $u_n\to u$ in $L^1$, $\|u_n\|_\infty\le\|u\|_\infty$, and
$\int u_n=\int u$.  Young's inequality gives
\begin{equation}
 \|G*(u_n-u)\|_{L^1(0,T)}
 \le\|G\|_{L^1(0,T)}\|u_n-u\|_{L^1(0,T)}\to0.
 \label{eq:scalar-state-L1-closure}
\end{equation}
For continuous $h$ on the common compact state range, truncation and uniform
continuity show convergence of the costs.  A strict negative probe therefore
survives at a finite partition level.  Under a cap, the state/readout scaling
in the main proof is performed before this averaging, so the finite
representative stays inside the cap.

For an unbounded complementary density at the origin, set
$\ell_n=\min(\ell,n)$.  A finite-step input makes its state a finite linear
combination of translates of the primitive of $G$, hence
$D\in W^{1,1}$ and $D(0)=0$.  The estimate
\begin{equation}
 \|((\ell_n-\ell)*D)'\|_1
 \le\|\ell_n-\ell\|_1\|D'\|_1\to0
 \label{eq:appendix-singular-complement}
\end{equation}
justifies passage to the full Stieltjes complement in
\eqref{eq:cp-storage-identity}, including its atoms and singular continuous
part.

\section{Prony separators and hierarchy closure}
\label{app:prony-separators}

This appendix supplies the finite constructions used to close the sectors of
\Cref{thm:two-mode} that are not covered by a positive-Prony or stable
complementary-inverse storage representation.  It also proves both closure
directions, rational refinement, and effective witness detection for the
finite-Prony hierarchy of \Cref{thm:finite-prony-hierarchy}.

\subsection{Opposite permanent and instantaneous signs}

Assume the identity gate \eqref{eq:two-mode-identity-gate}, a signed spectrum
$a_1a_2<0$, and $gb<0$.  Then necessarily $g<0<b$.  The transfer polynomial
in \eqref{eq:two-mode-P} factors as
\begin{equation}
 P(p)=b(p-\kappa)(p+\mu),\qquad \kappa,\mu>0,
 \label{eq:app-unstable-factorization}
\end{equation}
and
\begin{equation}
 \frac1{p\laplace G(p)}
 =\frac1b+\frac{u}{p-\kappa}+\frac{s}{p+\mu},
 \qquad u>0>s.
 \label{eq:app-unstable-inverse}
\end{equation}
Indeed,
\begin{equation}
 us=
 \frac{a_1a_2\lambda_1\lambda_2(\lambda_2-\lambda_1)^2}
 {b^4(\kappa+\mu)^2}<0,
 \label{eq:app-residue-product}
\end{equation}
while the residue at $\kappa$ is positive.

\begin{proposition}[Four-cell unstable-mode separator]
\label{prop:four-cell-separator}
Every kernel satisfying the preceding conditions has a continuous
nondecreasing normalized readout and an ordinary finite piecewise-constant
round trip of strictly negative cost.  The witness scales into every
positive symmetric rate cap.
\end{proposition}

\begin{proof}
Fix a cell length $\Delta>0$ and set
\begin{equation}
 R=e^{\kappa\Delta}>1,\quad S=e^{-\mu\Delta}\in(0,1),\quad
 U=\frac{u(R-1)}\kappa>0,\quad
 V=\frac{s(1-S)}\mu<0.
 \label{eq:four-cell-parameters}
\end{equation}
For the inverse density
$\ell(t)=ue^{\kappa t}+se^{-\mu t}$, its mass on the $j$th lag cell is
\begin{equation}
 k_j=UR^j+VS^j.
 \label{eq:inverse-cell-mass}
\end{equation}
Prescribe four consecutive constant state levels
\begin{equation}
 \begin{aligned}
 x_0&=-c\left(\frac{1+\theta}{R}+\frac1{R^2}+\frac1{R^3}\right),\\
 x_1&=c(1+\theta),\qquad x_2=x_3=c,
 \end{aligned}
 \qquad c,\theta>0.
 \label{eq:four-cell-state}
\end{equation}
Then
\begin{equation}
 \sum_{i=0}^3R^{-i}x_i=0,
 \label{eq:unstable-moment-cancellation}
\end{equation}
so the unstable inverse tail cancels exactly.  The inverse-generated
inventory is $Q=b^{-1}D+\ell*D$.  For a cellwise constant state, the
convolution increment on cell $i$ is
\begin{equation}
 y_i=k_0x_i+\sum_{j<i}(k_{i-j}-k_{i-j-1})x_j.
 \label{eq:four-cell-increment}
\end{equation}
Choose $h(x)=(x-c)_+$.  Only cell $1$ is active, and the feedthrough term
telescopes because the compact state path returns to zero.  Hence
\begin{equation}
 \cost=c\theta y_1,
 \end{equation}
where direct substitution gives
\begin{equation}
 \frac{y_1}{c}
 =\frac UR\left(\theta+R^{-2}\right)
 +V\left[(1+\theta)+(1-S)
 \left(\frac{1+\theta}{R}+R^{-2}+R^{-3}\right)\right].
 \label{eq:four-cell-cost-formula}
\end{equation}
As $\Delta\to\infty$, this tends to
\begin{equation}
 \frac u\kappa\theta+\frac s\mu(1+\theta),
\end{equation}
which is negative at $\theta=0$.  Choose first a small positive $\theta$
and then a large finite $\Delta$; the cost is strictly negative.

The state has the exact unstable moment and a strictly negative cell
cost.  \Cref{prop:prony-cell-realization} therefore supplies the bounded
ordinary round trip and its cap-local scaling; its moment correction
and tail estimates apply directly to this one-unstable-pole case.
\end{proof}

For example, the kernel
\begin{equation}
 G(t)=-1-3e^{-t}+10e^{-3t}
 \label{eq:rational-four-cell-kernel}
\end{equation}
has inverse
$1/(p\laplace G(p))=1/6-5/[18(2p+1)]+8/[9(p-1)]$.  Taking
$\Delta=2\log2$, $c=4$, $\theta=1/8$ gives the rational state levels
$(-23/16,9/2,4,4)$ and limiting cost $-259/2304$ for the hinge
$(x-4)_+$.

\subsection{The singular surface \texorpdfstring{$G(0+)=0$}{G(0+)=0}}

When $b=0$, one has $g=-A$ and the transfer polynomial reduces to
\begin{equation}
 P(p)=-M_\infty p-A\lambda_1\lambda_2.
 \label{eq:bzero-polynomial}
\end{equation}

\begin{proposition}[First-order singular separator]
\label{prop:singular-first-order-separator}
Assume \eqref{eq:two-mode-identity-gate}, $b=0$, $M_\infty>0$, and
$a_1a_2<0$.  Thus $A>0$.  Put
\begin{equation}
 \nu=\frac{A\lambda_1\lambda_2}{M_\infty}>0,\qquad
 c_*=-\frac{(\nu-\lambda_1)(\nu-\lambda_2)}{M_\infty}<0.
 \label{eq:singular-first-parameters}
\end{equation}
For every compact $C^2$ state path with $D,D'$ zero at its endpoints and
every $C^1$ monotone readout,
\begin{equation}
 \cost=
 \frac1{M_\infty}\int h'(D)(D')^2\dd t
 +c_*\int(D-Z)(h(D)-h(Z))\dd t,
 \label{eq:singular-first-identity}
\end{equation}
where
\begin{equation}
 Z(t)=\nu\int_0^te^{-\nu(t-s)}D(s)\dd s,\qquad
 Z'=\nu(D-Z).
 \label{eq:singular-filter}
\end{equation}
There exists such a path and a continuous monotone readout for which
$\cost<0$; it closes to a finite cap-local round trip.
\end{proposition}

\begin{proof}
Polynomial division gives
\begin{equation}
 \frac1{p\laplace G(p)}
 =-\frac p{M_\infty}
 +\frac{\nu-\lambda_1-\lambda_2}{M_\infty}
 +\frac{c_*}{p+\nu}.
 \label{eq:singular-first-inverse}
\end{equation}
The sign $c_*<0$ follows from
\begin{equation}
 (\nu-\lambda_1)(\nu-\lambda_2)
 =-\frac{a_1a_2\lambda_1\lambda_2
 (\lambda_2-\lambda_1)^2}{M_\infty^2}>0.
\end{equation}
Integrate the derivative term by parts, let the constant multiple of $D'$
 telescope, and use
\begin{equation*}
 \int h(Z)(D-Z)\dd t
 =\nu^{-1}\int h(Z)Z'\dd t=0
\end{equation*}
to obtain \eqref{eq:singular-first-identity}.

Choose smooth monotone profiles $\eta,\psi$ on $[0,1]$, increasing
from $0$ to $1$ and flat at both endpoints, and extend $\eta$ constantly
outside $[0,1]$.  Set $h_\varepsilon(x)=\eta(x/\varepsilon)$.
First prepare $Z=-1$ with a compact smooth negative state segment ending
at $D=0$; rescaling any nonzero negative bump achieves this exactly.
Both terms of \eqref{eq:singular-first-identity} vanish during preparation.
Reset time at the end of preparation.  Run $D=\varepsilon\psi(t/\tau)$,
hold $D=\varepsilon$ for
$(\log2)/\nu$, and return along the reversed ramp in time $\tau$.
Take $\varepsilon=\tau^2$.

During this excursion, $0\le D\le\varepsilon$ and comparison gives
$Z(t)\le\varepsilon-(1+\varepsilon)e^{-\nu t}$, so $Z<0$ for small
$\tau$.  On the hold, $h_\varepsilon(D)-h_\varepsilon(Z)=1$ and
\[
 \int_{\rm hold}(D-Z)\dd t
 \ge\frac{e^{-\nu\tau}}{2\nu}>\frac1{3\nu}.
\]
The two smooth ramps satisfy
$\int h_\varepsilon'(D)(D')^2\le C\varepsilon/\tau=C\tau$, where
$C=2\|\eta'\|_\infty\int_0^1(\psi')^2$ is fixed.  All other Bregman
contributions are nonnegative.  Hence
$\cost\le C\tau/M_\infty+c_*/(3\nu)<0$ for sufficiently small $\tau$.
This constructs the smooth state and readout directly.

The inverse generates a bounded integrable rate with $Q(\infty)=Q(0)=0$
and only a stable tail.  The truncation estimates
\eqref{eq:prony-finite-rt-compensation}--\eqref{eq:prony-compensation-cost-bound}
use only these properties, compactness of $D$, and boundedness of $G$;
they do not require $b\ne0$.  They preserve the strict margin, after
which \Cref{lem:finite-step-closure,lem:cap-scaling} finish the construction.
\end{proof}

\begin{proposition}[Second-order singular separator]
\label{prop:singular-second-order-separator}
Assume \eqref{eq:two-mode-identity-gate}, $a_1a_2<0$, $b=0$, and
$M_\infty=0$, so $A>0$.  For every compact $C^3$ state path whose
first two derivatives vanish at its endpoints and every $C^2$ monotone
readout,
\begin{equation}
 \cost=\frac1{A\lambda_1\lambda_2}
 \left[(\lambda_1+\lambda_2)\int h'(D)(D')^2\dd t
 -\frac12\int h''(D)(D')^3\dd t\right].
 \label{eq:singular-second-identity}
\end{equation}
There is a compact smooth round-trip state path with strictly negative
cost for the analytic readout $h(x)=e^x-1$.
\end{proposition}

\begin{proof}
Here
\begin{equation}
 \frac1{p\laplace G(p)}
 =-\frac{p^2+(\lambda_1+\lambda_2)p+\lambda_1\lambda_2}
 {A\lambda_1\lambda_2}.
 \label{eq:singular-second-inverse}
\end{equation}
Two integrations by parts give \eqref{eq:singular-second-identity}.  Let
$\phi\in C^\infty([0,1])$ increase from $0$ to $X>0$ with every derivative
zero at both endpoints, and put
\begin{equation}
 J_1=\int_0^1e^{\phi}(\phi')^2\dd s>0,\qquad
 J_2=\int_0^1e^{\phi}(\phi')^3\dd s>0.
\end{equation}
Run the rising path in time $\tau$ and the reversed falling path in time
$\sigma$.  Then
\begin{align}
 \int h'(D)(D')^2\dd t
 &=J_1(\tau^{-1}+\sigma^{-1}),\\
 \int h''(D)(D')^3\dd t
 &=J_2(\tau^{-2}-\sigma^{-2}).
\end{align}
Consequently the cost equals the positive factor
$(\tau^{-1}+\sigma^{-1})/(A\lambda_1\lambda_2)$ times
\begin{equation}
 (\lambda_1+\lambda_2)J_1
 -\frac{J_2}{2}(\tau^{-1}-\sigma^{-1}).
\end{equation}
Fix $\sigma$ and choose $\tau$ sufficiently small.  The polynomial inverse
in \eqref{eq:singular-second-inverse} generates a compact smooth rate whose
total volume is already zero.  Cell averaging gives a finite-step negative
round trip, and joint state/readout scaling makes it cap-local.
\end{proof}

\subsection{A linearly positive transient that always fails nonlinearly}

For $H_*(t)=2e^{-t}-e^{-2t}$, the identity-readout density is strictly
positive as stated in \eqref{eq:linear-positive-spectrum}.  We prove the
second claim of \Cref{prop:linear-not-nonlinear}.

When $g<-1$ or $g>0$, the endpoints $g$ and $b=g+1$ have the same sign.  A
safe stable inverse would have a nonnegative nonincreasing density, but
\begin{equation}
 \ell(0+)=\frac{2\cdot1-1\cdot2}{(g+1)^2}=0.
\end{equation}
Such a density would be identically zero, contradicting the nonconstant
transfer.

For every $g>-1$, use ideal pulse volumes
\begin{equation}
 (5,-6,1,1,-1)
 \label{eq:five-pulse-volumes}
\end{equation}
with intervening waits
$(\log2,\log2,3\log2,4\log2)$.  The pre- and post-pulse states are
\begin{equation}
\begin{array}{c|cc}
k&D_k^-&D_k^+\\ \hline
1&0&5g+5\\
2&5g+15/4&-g-9/4\\
3&-g-37/16&-21/16\\
4&-189/1024&g+835/1024\\
5&g+28675/262144&-233469/262144.
\end{array}
\label{eq:five-pulse-state-table}
\end{equation}
For $-1<g<-15/16$, choose
$-21/16<z<-g-9/4$; for $g>-15/16$, reverse these inequalities.  The
nondecreasing step $-\one_{(-\infty,z)}$ has respective costs
$-1$ and $-1/[16(g+1)]$.  At $g=-15/16$, $z=-43/32$ gives $-1/2$.
Replace the step in a small endpoint-free band by a linear ramp; the segment
integrals remain unchanged.  Rectangular pulses of small positive width have
the same exact total volume and converge to the strict negative limit.

At $g=-1$, use pulse volumes $(-3,1,7,-6,1)$ and waits
$(\log2,2\log2,\log2,\log2)$.  The five pre-pulse states are
\begin{equation}
 0,\quad\frac34,\quad\frac{111}{64},\quad
 \frac{31}{256},\quad-\frac{513}{1024}.
\end{equation}
With threshold $z=-513/2048$, a sufficiently narrow continuous ramp readout
makes only the final pulse see value $-1$, so the limiting cost is exactly
$-1$.  This covers the only singular shift not included above.

\subsection{A completely positive tail whose negative shift fails}

The following is the certificate behind
\Cref{prop:cp-tail-shift-fails}.  The complement in
\eqref{eq:cp-tail-counterexample-complement} is nonnegative and
nonincreasing, so the normalized complementary-kernel characterization makes
$H$ completely positive on every finite horizon.  Equivalently, $H$ is the
unique continuous solution of
\begin{equation}
 H(t)+\frac15\int_0^tH(s)\dd s
 +5\int_{(t-1)_+}^tH(s)\dd s=1.
 \label{eq:cp-tail-volterra-equation}
\end{equation}
The method of steps gives $H(0)=1$ and
\begin{equation}
 H'(t)=-\frac{26}{5}H(t)+5H(t-1),
 \qquad H(t)=0\quad(t<0).
 \label{eq:cp-tail-homogeneous-delay}
\end{equation}
Thus $H(t)=e^{-26t/5}$ on $[0,1]$, and variation of constants, followed
inductively over the unit intervals, proves $H(t)>0$ for every $t\ge0$.
The same induction gives exponential decay without a spectral assertion.
Indeed, for $W(t)=2e^{-t/100}$,
\begin{equation}
 \frac{26}{5}-\frac1{100}>5e^{1/100};
 \label{eq:cp-tail-comparison-inequality}
\end{equation}
here $e^{1/100}<100/99$ and $519/100>500/99$ give a purely rational
verification of the strict inequality.  Hence
$W'\ge-(26/5)W+5W(t-1)$.  Since $H\le W$ on $[0,1]$, scalar
comparison on successive unit intervals yields
\begin{equation}
 0<H(t)\le2e^{-t/100},\qquad t\ge0.
 \label{eq:cp-tail-positive-decay}
\end{equation}
Thus $H$ is a nonzero positive exponentially decaying completely positive
kernel.  Its Laplace transform is
\begin{equation}
 \laplace H(p)=\frac1{p+1/5+5(1-e^{-p})}.
 \label{eq:cp-tail-laplace}
\end{equation}
With $g=-1/2$, use four consecutive half-unit rate blocks
$(-1,-4,4,1)$.  Writing $x=H*v$, $Q=1*v$, and $D=x-Q/2$, the delay
equation is
\begin{equation}
 x'=v-\frac{26}{5}x+5x(t-1),\qquad x(t)=0\quad(t\le0).
 \label{eq:cp-tail-delay-equation}
\end{equation}
Let $E=e^{-13/5}$.  Solving the first two blocks and using
$\log2>2/3$, $\log(3/2)\ge2/5$ gives
\begin{equation}
 D(t)>-\frac1{25},\qquad 0\le t\le1.
 \label{eq:cp-shift-negative-block-bound}
\end{equation}
Solving the last block, bounding $E<1/10$ and
$e^{-52/25}<1/8$ by finite exponential Taylor sums, gives
\begin{equation}
 D(t)<-\frac15,\qquad \frac{19}{10}\le t\le2.
 \label{eq:cp-shift-terminal-bound}
\end{equation}

Set $c=1/25$, $\kappa=1000$, and define the analytic strictly increasing
normalized soft hinge
\begin{equation}
 h(q)=-\frac1\kappa\log(1+e^{-\kappa(q+c)})
 +\frac1\kappa\log(1+e^{-\kappa c}).
 \label{eq:cp-shift-soft-hinge}
\end{equation}
It satisfies $h(q)>-1/\kappa$ for $q\ge-c$,
$h(q)<1/\kappa$ everywhere, and
$h(q)<-4/25+1/\kappa$ for $q\le-1/5$.  Therefore the two negative-rate
blocks contribute at most $5/(2\kappa)$, the positive-rate portion before
$19/10$ contributes at most $12/(5\kappa)$, and the final tenth contributes
less than $-2/125+1/(10\kappa)$.  Summing gives
\begin{equation}
 \cost< -\frac2{125}+\frac5{1000}
 =-\frac{11}{1000}.
 \label{eq:cp-shift-exact-final-margin}
\end{equation}
For a cap $B>0$, replace $v$ by $Bv/4$ and $h(q)$ by $h(4q/B)$; the cost is
multiplied by $B/4$ and remains negative.

\subsection{Finite-Prony hierarchy: closure in both directions}
\label{app:finite-prony-closure}

We give the closure details used in \Cref{thm:finite-prony-hierarchy}.
Throughout this subsection, $G$, $\beta=1/b$, $\ell$, and the independent
real moments are as in \eqref{eq:finite-prony}--\eqref{eq:generalized-moments}.
Every exponential-polynomial term of $\ell$ is extended to all real
arguments when it appears in a two-sided expression.  For a state $D$ on
$[0,\infty)$, write
\[
 m(D)=\left(\int_0^\infty\phi_k(s)D(s)\dd s\right)_{k=1}^r.
\]
The moment functions are bounded because their exponents have positive
real part.  Thus $m$ is a bounded linear map on $L^1(0,\infty)$.
For $r=0$, all moment corrections below are absent.

\begin{proposition}[Cell realization and moment-exact approximation]
\label{prop:prony-cell-realization}
Under the preceding finite-Prony hypotheses:
\begin{enumerate}[label=\textup{(\roman*)}]
\item If $x\in\ker W_\Pi$ and $h\in\readouts$ satisfy
$h(x)^\transpose A_\Pi x<0$, an ordinary finite round trip has
negative cost for the same $h$.  Joint scaling gives a witness in every
positive cap.
\item For every finite round trip $v$ and every $h\in\readouts$,
its state $D=G*v$ admits compact cell approximants $D_n$ with $m(D_n)=0$,
$\|D_n-D\|_1+\|D_n-D\|_\infty\to0$, and
\[
 \int_0^\infty h(D_n)(\ell*D_n)'\dd t
 \longrightarrow\cost^{\mathrm{out}}_{G,h}[v].
\]
\end{enumerate}
Repeated and complex inverse poles are included through the stated real
moment basis.
\end{proposition}

\begin{proof}
\emph{Independent corrections.}
The functions $s^ke^{-\zeta s}$ are independent on every open interval.
Indeed, applying $\prod_{\eta\ne\zeta}(\mathrm d/\mathrm ds+\eta)^{m_\eta}$
to an analytic relation eliminates the other exponents and acts invertibly
on the polynomial multiplying $e^{-\zeta s}$.  Its coefficients must
vanish, for each $\zeta$.
Thus $r$ evaluation vectors can be chosen independently in any open
interval.  Small disjoint intervals $J_j$ around their points give an
invertible matrix
\begin{equation}\label{eq:prony-fixed-moment-matrix}
 M_{kj}=\int_{J_j}\phi_k(s)\dd s:
\end{equation}
after column normalization it converges to the evaluation matrix.
Smooth approximations to $\one_{J_j}$ also provide fixed smooth
correction functions with invertible moment matrix.

\paragraph{A global memory estimate.}
Let $\ell_s$ contain all terms with negative real exponent and $\ell_u$
all terms with positive real exponent.  For a compactly supported or
exponentially decaying state with $m(D)=0$, expansion of each polynomial
factor gives
\begin{equation}\label{eq:prony-future-inverse}
 \int_0^\infty\ell_u(t-s)D(s)\dd s=0,
 \qquad
 (\ell_u*D)(t)=-\int_t^\infty\ell_u(t-s)D(s)\dd s
 \quad(t\ge0).
\end{equation}
For example, a term $e^{\zeta(t-s)}(t-s)^j$ contributes a linear
combination of
$e^{\zeta t}t^{j-k}\int s^ke^{-\zeta s}D(s)\dd s$, all zero.
Define
\begin{equation}\label{eq:prony-two-sided-kernel}
 K(r)=\ell_s'(r)\one_{(0,\infty)}(r)
      -\ell_u'(r)\one_{(-\infty,0)}(r)\in L^1(\R).
\end{equation}
Stable modes decay forwards and unstable modes backwards, including their
polynomial factors.  Differentiating \eqref{eq:prony-future-inverse} gives,
for the zero extension of $D$,
\begin{equation}\label{eq:prony-global-memory-bound}
 \mathcal T D:=(\ell*D)'=\ell(0)D+K*D,
 \qquad\|\mathcal T D\|_1\le C_\ell\|D\|_1,
\end{equation}
where $C_\ell=|\ell(0)|+\|K\|_1$.
The identity holds distributionally and almost everywhere; the moment
expansion also makes its right side zero for $t<0$.
Young's inequality applied to differences yields
\begin{equation}\label{eq:prony-memory-difference-bound}
 \|\mathcal T(D_1-D_2)\|_1\le C_\ell\|D_1-D_2\|_1.
\end{equation}
The constant does not depend on the support length, which is what the
converse needs.

For a compact step state $D_x$, the convolution $Y=\ell*D_x$ is
absolutely continuous.  Its cell increments are $A_\Pi x$, whence
\begin{equation}\label{eq:prony-step-memory-cost}
 \int_0^\infty h(D_x(t))\mathcal T D_x(t)\dd t
 =h(x)^\transpose A_\Pi x.
\end{equation}
The contribution outside the state support vanishes because $h(0)=0$.
The feedthrough is handled by compact smooth loops, not by assigning an
arbitrary value to $h$ at a state jump.

\paragraph{Realization of a negative cell form.}
Let $x$ and $h$ satisfy \textup{(i)}.  The compact step state $D_x$ has
exact moments and a strictly negative value in
\eqref{eq:prony-step-memory-cost}.  Smooth it to compact states
$Y_\varepsilon\in C_c^\infty(0,\infty)$ with supports contained in a fixed finite interval $[0,S]$,
a common bound, almost-everywhere convergence, and
$\|Y_\varepsilon-D_x\|_1\to0$.  For instance, a one-sided smoothing
followed by a vanishing right translation gives this approximation also
when the first state cell begins at zero.  Its moment error tends to zero.
Use fixed smooth correction functions $\psi_j$ with invertible moment
matrix $M^\psi$, enlarge $S$ once to contain their supports, and set
\begin{equation}\label{eq:prony-smooth-moment-correction}
 D_\varepsilon=Y_\varepsilon-
 \sum_{j=1}^r\bigl[(M^\psi)^{-1}m(Y_\varepsilon)\bigr]_j\psi_j.
\end{equation}
These states have exact moments and retain all the stated convergence and
boundedness properties.  By \eqref{eq:prony-memory-difference-bound},
$\mathcal T D_\varepsilon\to\mathcal T D_x$ in $L^1$.
Continuity of $h$, dominated convergence against
$|\mathcal T D_x|\dd t$, and the common bound on $h(D_\varepsilon)$ show
\begin{equation}\label{eq:prony-smoothed-cost-limit}
 \int h(D_\varepsilon)\mathcal T D_\varepsilon\dd t
 \longrightarrow h(x)^\transpose A_\Pi x<0.
\end{equation}
Fix one sufficiently small positive $\varepsilon$; all subsequent steps
use this single smooth state $D=D_\varepsilon$.

Set
\begin{equation}\label{eq:prony-realizing-inventory}
 Q=\beta D+\ell*D,\qquad v=Q'.
\end{equation}
The inverse transfer identity gives $G*v=D$.  The state starts and ends at
zero, so the feedthrough term integrates to
$\beta\int h(D)D'\dd t=\beta[F_h(D(\infty))-F_h(D(0))]=0$.
Thus its cost is the negative value in
\eqref{eq:prony-smoothed-cost-limit}.  The expansion in
\eqref{eq:prony-future-inverse} cancels every unstable tail of $Q$ and
$v$ after the state support.  Only stable exponential-polynomial tails
remain: $v$ is bounded and integrable, $Q(\infty)=Q(0)=0$, and
$\int_0^\infty v=0$.

Let $R$ exceed the state support, put
$\eta_R=\int_R^\infty|v|$ and
$q_R=\int_0^R v=-\int_R^\infty v$, and define
\begin{equation}\label{eq:prony-finite-rt-compensation}
 u_R=v\one_{[0,R]}-q_R\one_{[R,R+1]}.
\end{equation}
This is a bounded, compactly supported round trip.  Up to $R$ its
state equals $D$, so its cost there is the fixed strictly negative loop
cost.  For $R\le t\le R+1$, boundedness of the finite-Prony kernel and
$D(t)=0$ give
\begin{equation}\label{eq:prony-compensation-state-bound}
 |(G*u_R)(t)|\le\|G\|_\infty(\eta_R+|q_R|)
 \le2\|G\|_\infty\eta_R.
\end{equation}
The compensation cost is at most, in absolute value,
\begin{equation}\label{eq:prony-compensation-cost-bound}
 |q_R|\sup_{|z|\le2\|G\|_\infty\eta_R}|h(z)|\longrightarrow0.
\end{equation}
Choose a finite $R$ preserving strict negativity and apply
\Cref{lem:finite-step-closure}.  This yields an ordinary finite
piecewise-constant round trip.  Finally apply \Cref{lem:cap-scaling} to
obtain a witness in every positive cap.

\paragraph{Approximation of ordinary round-trip states.}
Let $v\in\rt(T_0)$ and $h\in\readouts$, extend $v$
by zero after $T_0$, and put $Q(t)=\int_0^t v$ and $D=G*v$.
The inventory is compactly supported.  Since the total volume is zero,
the permanent part of $D$ vanishes after $T_0$ and the remaining state
and its derivative decay exponentially.  The state is continuous,
$D\in W^{1,1}(0,\infty)$, and $D(0)=D(\infty)=0$.
For $\re p>0$,
\begin{equation}\label{eq:prony-actual-state-transform}
 \laplace D(p)=R(p)\laplace Q(p).
\end{equation}
Every unstable actual inverse pole $\zeta$ of multiplicity $m$ is a zero
of $R$ of the same multiplicity.  Differentiating at $\zeta$ therefore
gives
\begin{equation}\label{eq:prony-actual-state-moments}
 \int_0^\infty s^ke^{-\zeta s}D(s)\dd s
 =(-1)^k\laplace D^{(k)}(\zeta)=0,
 \qquad 0\le k<m.
\end{equation}
Hence $m(D)=0$, and the inverse relation
$v=\beta D'+\mathcal T D$ holds.  The primitive chain rule gives
\begin{equation}\label{eq:prony-actual-memory-cost}
 \cost^{\mathrm{out}}_{G,h,T_0}[v]
 =\int_0^\infty h(D)\mathcal T D\dd t.
\end{equation}
All these integrals converge absolutely: $D'$ and $\mathcal T D$ are
integrable and $h$ is bounded on the state range.

Choose compact cell states $Y_n$ with
$\|Y_n-D\|_1+\|Y_n-D\|_\infty\to0$; continuity and exponential decay
of $D$ give this by truncation and successively finer partitions.
Fix the correction intervals $J_j$ of
\eqref{eq:prony-fixed-moment-matrix} once and for all, and include their
endpoints in every partition.  The $J_j$ are allowed to be unions of
refined cells, so refinement inside them is not restricted.  With
\begin{equation}\label{eq:prony-cell-moment-correction}
 D_n=Y_n-\sum_{j=1}^r[M^{-1}m(Y_n)]_j\one_{J_j},
\end{equation}
one has $m(D_n)=0$.  Boundedness of the moment functions gives
$|m(Y_n)|=|m(Y_n-D)|\le C\|Y_n-D\|_1\to0$, so
\begin{equation}\label{eq:prony-corrected-cell-convergence}
 \|D_n-D\|_1+\|D_n-D\|_\infty\longrightarrow0.
\end{equation}
Let $B$ bound their common state range, let
$H_B=\sup_{|z|\le B}|h(z)|$, and let $\omega_h$ be the modulus of
continuity of $h$ on $[-B,B]$.  The support-independent estimate gives
\begin{align}
 \left|\int h(D_n)\mathcal T D_n\dd t
            -\int h(D)\mathcal T D\dd t\right|
 &\le H_B C_\ell\|D_n-D\|_1\nonumber\\
 &\quad+\omega_h(\|D_n-D\|_\infty)\|\mathcal T D\|_1
 \longrightarrow0.
 \label{eq:prony-global-cost-convergence}
\end{align}
This proves \textup{(ii)}.  In particular, every strict negative
ordinary witness has a negative finite cell form for the same readout.
\end{proof}

\subsection{Full-rank refinement and rational endpoints}
\label{app:prony-rational-refinement}

Consider a feasible cut on an arbitrary finite partition.  First multiply
$x,z$ by a factor greater than one, so every threshold gap and the absolute
objective margin exceed one.  Choose the small intervals $J_j$ used in
\eqref{eq:prony-fixed-moment-matrix} inside the existing cells, avoiding
all original endpoints, and insert their endpoints into the partition.
The independence argument applies on the complement of finitely many
points and therefore permits this choice.  The new moment matrix has full
row rank because it contains the invertible $J_j$ submatrix.

Copy each original state value to its subcells and assign those subcells
the original cut membership.  This leaves the state function and its
moments unchanged.  It also leaves the cut objective unchanged: on each
original cell the indicator coefficient is constant, and the increments
of $\ell*D_x$ telescope over its subcells.  Thus full row rank is obtained
by an exact refinement of the given certificate, without adding a new
state level or changing its threshold.

Both $W_\Pi$ and $A_\Pi$ depend continuously on the ordered endpoints.
Perturb the nonzero endpoints to nearby rationals, maintaining their order,
and denote the resulting matrices by $W',A'$.  Full row rank persists.
Set
\begin{equation}\label{eq:rational-moment-projection}
 x'=x-W'^\transpose(W'W'^\transpose)^{-1}W'x.
\end{equation}
Then $W'x'=0$ exactly and $x'\to x$, because $W'x\to W_\Pi x=0$.
Keep $z$ fixed.  Every threshold gap and the signed objective retain a
margin greater than one for a sufficiently small perturbation.  If there
are no unstable moments no refinement or projection is needed.  This
establishes the rational-partition assertion, including the exact moment
equalities, without assuming that the state coordinates are rational.

\subsection{Effective enumeration of negative finite witnesses}
\label{app:prony-effective-detection}

We prove \Cref{cor:prony-unsafety-semialgorithm} directly in control space.
This also separates the countable mathematical cut hierarchy from questions
about deciding equalities among real coefficients or inverse poles.

\paragraph{Density with exact volume.}
Start from any strict finite witness $(v,h)$.  Extend $v$ by zero to a
fixed rational horizon $T_1$ larger than its support.  Approximate its
finitely many interval endpoints and rate values by rationals, obtaining
rational step controls $w_n$ with a common bound and
$\|w_n-v\|_1\to0$ on $[0,T_1]$.  Their volumes
$q_n=\int_0^{T_1}w_n$ are rational and tend to zero.  Add the rational
rate $-q_n$ on the unit interval $[T_1,T_1+1]$.  The resulting $v_n$
have rational data, exact zero volume, a common bound, and converge in
$L^1$ to the zero extension of $v$ on this fixed horizon.
Since $G$ is bounded,
\[
 \|G*(v_n-v)\|_\infty\le\|G\|_\infty\|v_n-v\|_1\to0.
\]
Continuity of $h$ gives convergence of the state-outside costs.

On a common compact state interval, approximate $h$ uniformly by
continuous nondecreasing piecewise-linear functions with rational
breakpoints and rational values, including the vertex $(0,0)$, and extend
them constantly beyond the outer breakpoints.  Such approximations exist
by uniform continuity and order-preserving rational approximations to the
finitely many sampled values.  For example, values to the right of zero
can be rounded down and those to its left rounded up on a common rational
value grid.  Their order and normalization are preserved, and the rounding
error tends uniformly to zero.  If $h_n$ is one such approximation,
\[
 |\cost^{\mathrm{out}}_{G,h_n}[v_n]
       -\cost^{\mathrm{out}}_{G,h}[v_n]|
 \le\|v_n\|_1\sup_{|z|\le B}|h_n(z)-h(z)|\to0.
\]
A strict negative witness therefore exists among the enumerated rational
pairs whenever safety fails.

\paragraph{Certified integration.}
A computable presentation gives every real parameter to arbitrarily small
certified rational error.  For each rational pair $(v,h)$, the state on a
finite horizon is a computable continuous function: integrate the finite
sum of exponentials on each rate cell.  Positivity of each $\lambda_j$
permits a certified positive lower bound when division by a rate is used.
For a horizon $T$ and rate bound $V$, one may use the explicit derivative
bound
\begin{equation}\label{eq:prony-effective-state-lipschitz}
 \|D'\|_\infty
 \le V\left(|b|+T\sum_{j=1}^n|a_j|\lambda_j\right).
\end{equation}
A rational upper bound for the right side is computable from the input
presentation.  If $L_h$ is the largest slope of the rational
piecewise-linear readout, then on rate cell $i$ the integrand
$v_i h(D(t))$ has Lipschitz constant at most
$|v_i|L_h\|D'\|_\infty$.  Split integration at the finitely many rate
jumps; certified exponential evaluation and Riemann sums with this
Lipschitz remainder produce rational cost enclosures of arbitrarily small
width.

Enumerate all rational finite controls, retain those with exactly zero
rational volume, and pair them with all the rational readouts just
described.  Dovetail the enclosure refinements over these pairs and halt
when an upper endpoint is strictly negative.  The returned finite pair and
enclosure certify an actual negative round trip.  Density ensures
termination on every unsafe input, while on a safe input no negative
enclosure can occur.  The procedure does not require a pole reduction,
a multiplicity test, or an approximate substitute for an exact moment
equality.

\section{Matrix orientation and nonstationary closure}
\label{app:matrix-details}

\subsection{Why Loewner-positive complements do not suffice}

Take
\begin{equation}
 B=\begin{pmatrix}1&0\\0&2\end{pmatrix},\qquad
 G(t)\equiv B^{-1},\qquad
 h(x)=Hx,\qquad
 H=\begin{pmatrix}2&1\\1&2\end{pmatrix}.
 \label{eq:loewner-false-data}
\end{equation}
Both $B$ and $H$ are symmetric positive definite, $B\delta_0$ is a positive
matrix complement of $G$, and
$h=\nabla(\tfrac12x^\transpose Hx)$ is a strongly monotone convex gradient.
Nevertheless, let the state traverse the clockwise unit square.  Since
$D=B^{-1}q$ and $q'=v$, the realizing rate $v=BD'$ has four constant blocks
and zero vector volume.  Its cost is
\begin{equation}
 \cost=\oint \dd D^\transpose BH D,\qquad
 BH=\begin{pmatrix}2&1\\2&4\end{pmatrix}.
 \label{eq:loewner-square-line-integral}
\end{equation}
Green's theorem gives cost $+1$ in the counterclockwise orientation and
\begin{equation}
 \cost=-1
 \label{eq:loewner-square-cost}
\end{equation}
clockwise.  The defect is that the transformed one-form $Bh$ is not
conservative.  Hence positive matrix coefficients and an ordinary Euclidean
convex-gradient readout cannot replace the compatibility condition
\eqref{eq:matrix-conic-compatibility}.

For a symmetric coefficient $M$ and $h=\nabla\Psi$ with $\Psi\in C^2$, the
compatibility test is
\begin{equation}
 Mh\text{ conservative}
 \quad\Longleftrightarrow\quad
 M\nabla^2\Psi(x)=\nabla^2\Psi(x)M
 \quad\text{for every }x.
 \label{eq:matrix-commutant-test}
\end{equation}
Thus a scalar common commutant of the Hessians forces every symmetric
complement coefficient in the conic theorem to be scalar.  Fixed common
eigenspaces permit block coefficients, and a degenerate Hessian block may
support noncommuting matrices.

\subsection{A noncommuting conic complement}

The conic theorem itself does not impose commutation.  In dimension three,
set
\begin{equation}
 M_1=\begin{pmatrix}1&0&0\\0&1&0\\0&0&2\end{pmatrix},\qquad
 M_2=\begin{pmatrix}2&0&0\\0&2&1\\0&1&2\end{pmatrix}.
 \label{eq:noncommuting-conic-matrices}
\end{equation}
Both are positive definite, but
\begin{equation}
 [M_1,M_2]=
 \begin{pmatrix}0&0&0\\0&0&-1\\0&1&0\end{pmatrix}\ne0.
\end{equation}
Let
\begin{equation}
 R=I\delta_0+e^{-t}M_1\dd t+e^{-2t}M_2\dd t,\qquad
 \laplace G(p)=\frac1p
 \left(I+\frac{M_1}{p+1}+\frac{M_2}{p+2}\right)^{-1},
 \label{eq:noncommuting-conic-transfer}
\end{equation}
and choose
\begin{equation}
 h(x)=x+x_1^3e_1.
 \label{eq:noncommuting-conic-readout}
\end{equation}
The three transformed fields are convex gradients:
\begin{align}
 h(x)&=\nabla\left(\tfrac12\|x\|^2+\tfrac14x_1^4\right),\\
 M_1h(x)&=\nabla\left(\tfrac12x^\transpose M_1x+\tfrac14x_1^4\right),\\
 M_2h(x)&=\nabla\left(\tfrac12x^\transpose M_2x+\tfrac12x_1^4\right).
\end{align}
Therefore \Cref{thm:matrix-conic} gives a nonnegative Bregman
decomposition for every input.

This example is noncommuting in the time domain.  With $S_0=M_1+M_2$ and
$S_1=M_1+2M_2$, the large-$p$ expansion
$\laplace G(p)=I/p-S_0/p^2+(S_0^2+S_1)/p^3+O(p^{-4})$ gives
$[G(t),G(s)]=\tfrac12ts(t-s)[S_0,S_1]+O((s+t)^4)$ for small distinct
$s,t>0$, and $[S_0,S_1]=[M_1,M_2]\ne0$.  Each $G(t)$ is symmetric, since
the transfer is symmetric for real $p>0$, but matrices at different times
do not commute.

\subsection{Finite-step closure for matrix path tests}

In the fully actuated model, a piecewise $C^1$ path $x$ is generated by
$v=B^{-1}(x'+Ax)$.  Let $v_n$ be its cell averages.  Then
$v_n\to v$ in $L^1$, and variation of constants gives
\begin{equation}
 \sup_{t\le T}\|D_n(t)-D(t)\|
 \le C_{A,B,T}\|v_n-v\|_{L^1(0,T)}\to0.
 \label{eq:matrix-path-closure}
\end{equation}
If the limiting path is a vector round trip, cell averaging preserves each
coordinate of its total volume exactly.  Continuity of the fixed readout on
the common compact state range then gives convergence of costs.  Hence the
compressed-loop, long-hold, and remote-compensation arguments in
\Cref{sec:matrix} all return to literal finite piecewise-constant controls.

\section{Proofs for Section~\ref{sec:friction-complexity}}
\label{app:friction-details}

\subsection{Vertical reachability}

\begin{proof}[Proof of \Cref{thm:vertical-reachability}]
Choose a finite step function $q$ on a source interval, with values in
$V_f$, and set $q=f(0)$ near both endpoints.  By
\Cref{lem:vertical-section}, each source cell admits finite mixtures of
ordinary rates whose rate mean is exactly zero and whose impact mean tends
to the prescribed value of $q$.  Repeating those mixtures at increasing
frequency, as in \Cref{lem:finite-support-realization}, yields finite
piecewise-constant controls $v_n$ with exactly zero volume on each source
cell and
\begin{equation}
 v_n\weakstar0,
 \qquad
 f(v_n)\weakstar q.
 \label{eq:vertical-chattering-limit}
\end{equation}
The causal two-variable kernel
$\one_{\{s<t\}}G(t-s)$ lies in $L^1((0,T)^2)$.  Finite-rank $L^1$
approximation therefore turns the second convergence in
\eqref{eq:vertical-chattering-limit} into strong $L^1$ convergence of the
states.  Uniform boundedness of the states and continuity of $h$ then give
convergence of the costs.

For the target test, fix $0<t<T$ and such a selection $q$ on $[0,t]$.
For each fixed sufficiently small $\varepsilon>0$, place rates $\sigma a$
on $[0,\varepsilon]$ and $-\sigma a$ on $[t,t+\varepsilon]$, where
$0<a\le B$.  Chatter only between these blocks; all their volumes cancel
exactly.  The preceding convergence gives, at this fixed width,
$\cost[v_{\varepsilon,n}]\to C_\varepsilon$, with zero limiting pump cost.
The early target sees a state bounded by
$\|f\|_\infty\int_0^\varepsilon|G|=o(1)$.  The late target sees
$z+o(1)$, where $z=\int_0^tG(t-s)q(s)\dd s$: convolution with a bounded
step source is continuous, and changing the source on either target block
has uniformly vanishing effect by absolute continuity of the $L^1$ integral.
Consequently
\begin{equation}
 \lim_{\varepsilon\downarrow0}\frac1\varepsilon
 \lim_{n\to\infty}\cost[v_{\varepsilon,n}]
 =-\sigma a k(z).
 \label{eq:thin-baseline-cost}
\end{equation}
If $k(z)\ne0$, choose $\sigma=\operatorname{sgn}k(z)$, then a fixed width
with $C_\varepsilon<-a\varepsilon|k(z)|/2$, and finally a finite chatter
with $|\cost[v_{\varepsilon,n}]-C_\varepsilon|
<a\varepsilon|k(z)|/4$.  This is a negative ordinary round trip.  Thus
$z\in Z_h$.

Step selections are dense in the finite measure $|G(t-s)|\dd s$.  Pointwise
selection of $\nu_+$ on the positive part of $G$ and $\nu_-$ on the negative
part, or conversely, gives the two endpoints in
\eqref{eq:vertical-reachable-interval}; closure gives the whole interval.
Conversely, convex separation of $K_f$ shows that any weak-star limit of
$(v_n,f(v_n))$ for vanishing-mesh, cellwise zero-volume pumps belongs to
$\{0\}\times V_f$ almost everywhere.  Pairing the source limit with
$G(t-\cdot)\in L^1(0,t)$ proves that no larger limiting state range occurs.
Since $Z_h$ is closed, boundary points cause no loss.
\end{proof}

\subsection{The critical cubic constant}

\begin{proof}[Proof of \Cref{thm:exact-cubic-repair}]
Time scaling reduces the proof to $T=1$.  Define the scalar Hamiltonian
\begin{equation}
 \mathfrak h(c):=\sup_{z\in\R}
 \{z|z|-|z|^3-3cz\}.
 \label{eq:cubic-hamiltonian}
\end{equation}
On the positive branch $z=x>0$, stationarity gives
$c=2x/3-x^2$ and the attained value $x^2(2x-1)$.  On the negative branch
$z=-y<0$, it gives $c=2y/3+y^2$ and value $y^2(2y+1)$.  The envelope
derivatives are $-3x$ and $3y$, respectively, so the two positive branch
values cross exactly once.  Equality of their parameter and value is
equivalent to
\begin{equation}
 \frac23x-x^2=\frac23y+y^2,
 \qquad x^2(2x-1)=y^2(2y+1),
 \label{eq:hamiltonian-crossing}
\end{equation}
whose unique solution with $1/2<x<2/3$ and $y>0$ is
\eqref{eq:xy-star}.  Direct integration along the two branches gives the
complete flight time
\begin{equation}
 \int_{-\infty}^{\infty}\frac{\dd c}{\mathfrak h(c)}=C_*.
 \label{eq:full-flight-time}
\end{equation}

For $K>0$ put
\begin{equation}
 H_K(a):=\sup_{r\in\R}
 \{r|r|-K|r|^3-3ar\}=K^{-2}\mathfrak h(Ka).
 \label{eq:scaled-hamiltonian}
\end{equation}
If $K>\kappa_*$, the complete flight of $a'=H_K(a)$ lasts
$KC_*>1$.  Hence there is a finite solution on $[0,1]$.  From the definition
of $H_K$, for arbitrary $U,W\in\R$,
\begin{equation}
 UW|W|-K|W|^3-3aU|U|W\le H_K(a)|U|^3.
 \label{eq:bellman-pointwise}
\end{equation}
Take $U=X(t)$, $W=X'(t)$ and use
\begin{equation*}
 \frac{\dd}{\dd t}\{a(t)|X(t)|^3\}
 =a'|X|^3+3aX|X|X'.
\end{equation*}
Integration of \eqref{eq:bellman-pointwise} and the zero endpoints gives
\begin{equation}
 \int_0^1XX'|X'|\dd t\le K\int_0^1|X'|^3\dd t.
\end{equation}
Letting $K\downarrow\kappa_*$ proves the upper bound.  The same pointwise
inequality may be averaged over a Young measure in $W$, so relaxed
chattering cannot improve it.

For attainment, define $c:(0,1)\to\R$ by
\begin{equation}
 t=\kappa_*\int_{-\infty}^{c(t)}
 \frac{\dd s}{\mathfrak h(s)}
\end{equation}
and set $X(t)=A\mathfrak h(c(t))^{-1/3}$ for any $A>0$.  Below the crossing,
the envelope identity gives $X'/X=x(c)/\kappa_*$; above it,
$X'/X=-y(c)/\kappa_*$.  Thus $X$ is a positive one-hump inventory.  The
branch formulas give
$\mathfrak h(c)=2|c|^{3/2}(1+O(|c|^{-1/2}))$ and
$x(c),y(c)\sim |c|^{1/2}$ at their respective infinite endpoints.
Thus $t\sim\kappa_*/\sqrt{-c(t)}$ near zero and
$1-t\sim\kappa_*/\sqrt{c(t)}$ near one.  In particular,
$X(t)\sim A2^{-1/3}t/\kappa_*$ at zero, with the analogous expression
in $1-t$ at one, and $X'$ tends to $\pm A2^{-1/3}/\kappa_*$.
Hence $X\in W^{1,\infty}_0(0,1)$ and
$c|X|^3=O(|c|^{-1/2})\to0$ at both endpoints.  With $a=c/\kappa_*$, the chosen rate
$X'/X$ attains the supremum in \eqref{eq:scaled-hamiltonian} almost
everywhere, so equality holds in \eqref{eq:bellman-pointwise}.  The boundary
term vanishes and gives equality in \eqref{eq:exact-cubic-constant}.
Amplitude scaling fits any positive cap without changing the homogeneous
quotient.  Density of polygonal inventories in $W^{1,3}_0$ gives the last
claim.
\end{proof}

\begin{proposition}[Two blocks are not globally sufficient]
\label{prop:four-block-friction}
For $T=B=1$, $\delta=2$, and a cubic penalty, the two-block optimum is
\begin{equation}
 \max_{0<r<1}
 \frac{r(1-r)}{2(1+r)(1+r^2)}
 =0.0750707765001944\ldots.
 \label{eq:two-block-cubic-optimum}
\end{equation}
Four equal-duration blocks with rates
\begin{equation}
 \left(\frac9{10},-\frac14,-\frac7{25},-\frac{37}{100}\right)
 \label{eq:four-block-friction-rates}
\end{equation}
have zero sum and ratio
\begin{equation}
 \frac{-\cost_2[v]}{\int_0^1|v|^3\dd t}
 =\frac{10448}{136205}
 =0.0767079035277706\ldots,
 \label{eq:four-block-friction-ratio}
\end{equation}
strictly larger than \eqref{eq:two-block-cubic-optimum}.  An exact rational
eighty-block certificate reaches
\begin{equation*}
 0.08810702011840479\ldots,
\end{equation*}
approaching \eqref{eq:exact-cubic-constant} from below.
\end{proposition}

\begin{proof}
Zero volume fixes the two active durations in the inverse ratio of the two
rate magnitudes, giving \eqref{eq:two-block-cubic-optimum}.  Its maximizer is
the root in $(0,1)$ of
$r^4-2r^3-2r^2-2r+1=0$.  The Bernstein coefficients give the uniform
upper bound $153/2000$ for the displayed two-block ratio.  Direct rational
arithmetic for \eqref{eq:four-block-friction-rates} gives
$\cost_2=-1959/125000$ and
$\int|v|^3=81723/400000$, hence
\eqref{eq:four-block-friction-ratio}; moreover
\begin{equation*}
 \frac{10448}{136205}-\frac{153}{2000}
 =\frac{11327}{54482000}>0.
\end{equation*}
The eighty-block ratio is likewise an exact rational evaluation; the
eighty rates are supplied in the computational supplement
(\Cref{app:verification}).
\end{proof}

\subsection{A dwell-dependent safe neighborhood}

\begin{proposition}[Dwell-dependent safe neighborhood]
\label{prop:dwell-safe-neighborhood}
Fix $N\in\N$, $\tau>0$, a horizon $T$, and $0<\gamma<1$.  Consider the
normalized class of nonzero fractional round trips having at most $N$
blocks, each positive-duration block of length at least $\tau$, with
$\max_i|v_i|=1$.  Blocks are the intervals of a partition of $[0,T]$, so
zero-rate blocks also count.  If this class is nonempty, define, with
$T$ and $\gamma$ fixed throughout,
\begin{equation}
 \mu_{N,\tau}:=\min_v\cost_1[v]>0.
 \label{eq:linear-energy-gap}
\end{equation}
For $\delta>0$, every member remains safe whenever
\begin{equation}
 |\delta-1|<
 \frac{e\mu_{N,\tau}}{2\|H\|_{L^1(\DeltaT)}},
 \qquad
 \|H\|_{L^1(\DeltaT)}
 =\frac{T^{2-\gamma}}{(1-\gamma)(2-\gamma)}.
 \label{eq:dwell-safe-neighborhood}
\end{equation}
\end{proposition}

\begin{proof}
Put $q=2-\gamma$, $L=\|H\|_1$, and $m=\mu_{N,\tau}/L$.
Nonemptiness implies $N\ge2$ and $\tau\le T/2$.  For each possible block
count, the duration/rate parameters satisfy
\begin{equation*}
 d_i\ge\tau,\quad \sum d_i=T,\quad |a_i|\le1,\quad
 \max|a_i|=1,\quad \sum d_i a_i=0.
\end{equation*}
They form a compact set, mapped continuously into $L^1(0,T)$.  The finite
union over block counts is compact and excludes zero.  Since
$|\cost_1[v]-\cost_1[w]|\le
2T^{1-\gamma}\|v-w\|_1/(1-\gamma)$ on the unit ball,
\Cref{prop:fractional-positive} gives the positive minimum.

For $0<\delta<1$, direct maximization gives
\begin{equation*}
 \sup_{|x|\le1}|f_\delta(x)-x|
 =(1-\delta)\delta^{\delta/(1-\delta)}.
\end{equation*}
The last factor decreases in $\delta$, as its logarithmic derivative is
$(\log\delta+1-\delta)/(1-\delta)^2<0$; at $\delta=1/4$ it is
$4^{-1/3}<2/e$.  For $\delta\ge1$, differentiation in the exponent gives
the smaller coefficient $1/e$.  Therefore
\begin{align}
 |f_\delta(x)-x|&\le(2/e)|\delta-1|,
 && |x|\le1,\quad \delta\ge1/4,
 \label{eq:dwell-power-modulus}\\
 \cost_\delta[v]&\ge\mu_{N,\tau}-(2/e)L|\delta-1|.
 \label{eq:dwell-modulus-energy-bound}
\end{align}
To cover the entire stated radius, if $\tau\le T/3$, use the admissible
rates $1,-1/2$ for durations $T/3,2T/3$.  Dropping their negative cross
term yields $m<3^{-q}+(2/3)^q/4<1/2$.  If $\gamma\le1/2$, the balanced
half-horizon rates $1,-1$ instead give
$m\le2^\gamma-1\le\sqrt2-1<1/2$.  In either case the stated radius forces
$\delta>1-e/4>1/4$, so \eqref{eq:dwell-modulus-energy-bound} applies.

Only $\tau>T/3$, $\gamma>1/2$ remains.  There are exactly two blocks,
with magnitude ratio $r\ge\tau/(T-\tau)>1/2$ and $1<q<3/2$.
The balanced case $r=1$ is exponent-invariant.  For $r<1$ and
$\delta\ge1$ they are safe by \Cref{thm:fractional-two-block-phase}.  For $0<\delta<1$, subadditivity
gives $A_q(r)\le qr$, and their adverse-order factor satisfies
\begin{equation*}
 \frac{S_{q,\delta}(r)}{r^\delta}
 =r^{q-\delta}+r-A_q(r)
 \ge r(1+r^{q-1}-q)
 >r(1+2^{-1/2}-3/2)>0.
\end{equation*}
Thus this last class is safe for every positive exponent, completing the
proof without reducing the radius.
\end{proof}

\section{Reproducible finite certificates}
\label{app:verification}

All universal implications in the paper are proved analytically.  The
computational supplement checks transform algebra, evaluates finite
strategies exactly or in interval arithmetic, and computes decimal values
of analytically defined constants.  It contains the Python sources, input
data, recorded outputs, and a file manifest.  Running \texttt{python
verify.py} executes every checker in a temporary workspace and compares its
output across two working directories.  The only nonstandard dependencies
are SymPy and mpmath; the versions used for the recorded run are listed.
The accompanying supplement is distributed with the arXiv version as
ancillary material.

The checks cover the following items.
\begin{itemize}
\item \emph{Sparse-comb manipulations} (\Cref{tab:certificates}): the cost
  of each comb is a finite sum of terms $v_if_\delta(v_j)K(r)$ with
  rational $v_i$, $r$ and $K$ from \eqref{eq:comb-kernel-primitive}.  The
  checker evaluates the sum in directed interval arithmetic at forty
  decimal digits; the enclosures have width below $10^{-30}$ and their
  upper endpoints are the certified bounds printed in the table.
\item \emph{Two-mode kernels} (\Cref{thm:two-mode}, \Cref{ex:safe-signed-two-mode}, \Cref{fig:two-mode}):
  the inverse of $G(t)=1+e^{-t}-e^{-2t}/8$, its poles, zero, residues,
  $\ell(0+)=16/75$ and $-\ell'(0+)=64/1125$ in exact arithmetic; the
  boundary $a_2=-a_1/(a_1+4)$ and the discriminant curve for $g=1$; the
  wedge for $g=-1$; and the residue test on a grid against a direct
  evaluation of the inverse density.
\item \emph{Separators}: the opposite-sector four-cell cost $-259/2304$,
  the completely positive tail whose shift fails with margin $-11/1000$,
  the five-pulse state tables of \Cref{prop:linear-not-nonlinear}, the
  singular boundaries, and the three-mode residue signs.
\item \emph{Cut hierarchy}: generalized Prony moments and individual cut
  programs from \Cref{thm:finite-prony-hierarchy}.
\item \emph{Nonstationary and matrix identities}: the two witnesses of
  \Cref{ex:time-diagnostics} and the orientation, drift, and noncommuting
  examples of \Cref{sec:matrix}.
\item \emph{Friction constants}: the defining quartic and the complete
  flight integral of \Cref{thm:exact-cubic-repair}, evaluated in two
  parametrizations, giving
  $C_*=11.34647094581980822528319769291547\ldots$ and
  $\kappa_*=C_*^{-1}=0.08813313009613913350810867462130\ldots$; the
  two-block optimum of \Cref{prop:four-block-friction}; and the exact
  rational eighty-block certificate.  The eighty rates $r_i$, with
  $v(t)=r_i$ on $[(i-1)/80,i/80)$, are supplied as a data file; they sum to
  zero, their maximum modulus is $0.968318$, and rational arithmetic gives
  \begin{equation}\label{eq:eighty-block-exact-ratio}
   \frac{-\cost_2[v]}{\int_0^1|v|^3\dd t}
   =\frac{46525312884098154559}{528054550268230209280}
   =0.0881070201184047938\ldots.
  \end{equation}
\end{itemize}
These computations certify their specific formulas and witnesses.  They do
not prove hierarchy completeness or infer safety from a finite search;
optimality of $\kappa_*$ is established by the Bellman inequality and its
extremizer, not by quadrature.

\bibliographystyle{plainnat}
\bibliography{references}

\end{document}